\documentclass[letterpaper, 12pt]{article}
\usepackage{graphicx}
\usepackage{etoolbox}
\usepackage{amsmath}
\usepackage{amsfonts}
\usepackage{hyperref}
\usepackage{amsthm}
 \usepackage{amscd}
\usepackage{mathrsfs}
\usepackage{caption}
\usepackage{subcaption}
\usepackage{float}
\usepackage[top=1in,left=1in,right=1in,bottom=1in]{geometry}

\newtheorem{theorem}{Theorem}[section]

\newtheorem{example}[theorem]{Example}

\hypersetup{
	colorlinks = true,
	linkcolor  = red,
	citecolor  = green,
	filecolor  = cyan,
	urlcolor   = magenta,}
\newcommand \bR {\mathbb{R}}

\newcommand{\bCp}{\mathbb{C}^+}

\newcommand {\bCpb} {\overline{\mathbb{C}^+}}

\newcommand{\ds}{\displaystyle}

\numberwithin{equation}{section}

\title{Changing the discrete spectrum of the half-line matrix Schr\"odinger operator}
\author{Tuncay Aktosun\thanks{aktosun@uta.edu}\\
Department of Mathematics\\
University of Texas at Arlington\\
Arlington, TX 76019-0408, USA\\
\\
Ricardo Weder\thanks{weder@unam.mx}
\\
Departamento de F\'\i sica Matem\'atica\\
Instituto de Investigaciones en
Matem\'aticas Aplicadas y en Sistemas\\
Universidad Nacional Aut\'onoma de M\'exico\\
Apartado Postal 20-126, IIMAS-UNAM\\
 Ciudad de M\'exico, CP 01000, M\'exico}

\date{}

\begin{document}

\maketitle

\begin{abstract}
We consider the matrix-valued Schr\"odinger 
operator on the half line with the general selfadjoint boundary condition.
When the discrete spectrum is changed without changing the continuous spectrum, we present
a review of the transformations of the relevant quantities including the regular solution, the Jost solution,
the Jost matrix, the scattering matrix, and the boundary matrices used to describe the selfadjoint
boundary condition. The changes in the discrete spectrum are considered when an existing bound state
is removed, a new bound state is added, and the multiplicity of a bound state is decreased or increased 
without removing the bound state. We provide various explicit examples to illustrate the theoretical results
presented.
\end{abstract}



\newpage

\section{Introduction}
\label{section1}

In \cite{AW2025} we have presented the Gel'fand--Levitan method for the matrix-valued Schr\"odinger equation on the
half line with the general selfadjoint boundary condition. By using that method we have then provided
the transformations of the relevant quantities when the discrete spectrum of the corresponding Schr\"odinger operator
is changed without changing the continuous spectrum.
In particular, we have presented the transformation formulas for the relevant quantities when the discrete spectrum is changed in four ways, which are
the complete removal of a bound state, the decrease of the multiplicity
of an existing bound state without removing that bound state, the addition of a new bound state, and the
increase of the multiplicity of an existing bound state, respectively.
These four types of elementary operations can be applied successively in any order to change the discrete spectrum by removing or adding any number of bound states or by decreasing or increasing of the multiplicities of any existing bound states.

The goal of the present paper is to illustrate the primary results of \cite{AW2025} by providing various explicit examples.
Such examples are expected to help the reader understand the theory better by clarifying
the theoretical results presented in \cite{AW2025}. Our paper is organized as follows.
In Section~\ref{section2} we present the preliminaries related
to the matrix-valued Schr\"odinger equation on the
half line with the general selfadjoint boundary condition. 
In Section~\ref{section3} we provide the basic material on the
bound states of the relevant Schr\"odinger operator.
In Section~\ref{section4} we describe the perturbation initiated by a change
in the spectral measure, we introduce the Gel'fand--Levitan system of integral equations, and
we describe how the solution to  the Gel'fand--Levitan system
yields the corresponding perturbation in the potential, in the regular solution, and in the boundary matrices. 
In Section~\ref{section5} we briefly describe the method to remove a bound state from the discrete
spectrum with the help of the Gel'fand--Levitan method.
In Section~\ref{section6} we provide a summary of the method to decrease the multiplicity
of an existing bound states with the help of the Gel'fand--Levitan procedure.
In Section~\ref{section7} we summarize the process to add a new bound state by using
 the Gel'fand--Levitan method.
 In Section~\ref{section8} we briefly describe the method to increase the multiplicity
of an existing bound state by the Gel'fand--Levitan procedure.
 In Section~\ref{section9} we illustrate the evaluation of various relevant quantities described in Sections~\ref{section2}
 and \ref{section3}
 associated with the Schr\"odinger operator related to \eqref{2.1} and \eqref{2.7}.
 Finally, in Section~\ref{section10} we illustrate the mathematical procedures
 described in Sections~\ref{section5}--\ref{section8}. Except for the proofs of Theorems~\ref{theorem8.1},
 \ref{theorem9.9}, and \ref{theorem9.10},
 we do not include any proofs in our paper.
 For the proofs of the mathematical techniques used here we refer the reader to \cite{AW2025}.
The reader can consult \cite{AW2025} also for a brief review of the literature on transformations to remove and
to add bound states and for a discussion of important applications of matrix-valued Schr\"odinger
operators.

\section{The matrix Schr\"odinger operator on the half line}
\label{section2}

In this section we present the preliminaries related to the matrix-valued Schr\"odinger 
equation on the half line with the selfadjoint boundary condition.
We refer the reader to \cite{AW2021} for any  elaborations on the material presented here.

We consider the half-line Schr\"odinger equation
\begin{equation}
\label{2.1}
-\psi''+V(x)\,\psi=k^2\psi,\qquad x\in\bR^+,
\end{equation}
where $x$ is the independent variable, the prime denotes the $x$-derivative, $k^2$ is the spectral parameter,
$\bR^+$ denotes the interval $(0,+\infty),$ and the
potential $V$ is an $n\times n$ selfadjoint matrix-valued function
of $x,$ with $n$ being a fixed positive integer.
The wavefunction
$\psi$ is either an $n\times n$ matrix or
a column vector with
$n$ components.
The selfadjointness is expressed as
\begin{equation}
\label{2.2}
V(x)^\dagger=V(x),\qquad x\in\bR^+,
\end{equation}
with the dagger denoting the matrix adjoint
(matrix transpose and complex conjugation).
 We use the terms selfadjoint and hermitian interchangeably.
 We use both $V$ and $V(x)$ to refer to the potential
with the latter use for emphasis on the $x$-dependence.
We let $\bCp$ denote the upper-half complex plane, use
$\bR$ for the real axis, and let $\bCpb:=\bCp\cup\bR.$

We always assume that the potential $V$ is at least integrable, i.e. we have
\begin{equation}\label{2.3}
\int_0^\infty dx\,|V(x)|<+\infty,
\end{equation}
where $|V(x)|$ denotes the
operator norm of the $n\times n$ matrix $V(x).$
We mention explicitly when we need a stronger condition on the potential.
For example, at times we require that the potential $V$ also has the first moment, i.e. we require that
we have
\begin{equation}\label{2.4}
\int_0^\infty dx\,(1+x)\,|V(x)|<+\infty.
\end{equation}
When \eqref{2.4} is satisfied, we say that the potential $V$ belongs
to $L_1^1(\mathbb R^+).$
In general, we say that the potential $V$ 
belongs to $L^1_\epsilon(\bR^+)$ for some nonnegative constant
$\epsilon$ if we have
\begin{equation}\label{2.5}
\int_0^\infty dx\, (1+x)^\epsilon   \,|V(x)|<+\infty.
\end{equation}
Thus, the potential $V$ is integrable when $\epsilon=0$ in \eqref{2.5}, in which case
we also say that $V$ belongs to $L^1(\mathbb R^+).$
Because all matrix norms are equivalent in a
finite-dimensional vector space, any matrix norm can be used in \eqref{2.3}--\eqref{2.5} instead of 
the operator norm.

We obtain our selfadjoint Schr\"odinger operator on the half line
by supplementing \eqref{2.1} with a
selfadjoint boundary condition at $x=0.$
We find it the most convenient to describe that selfadjoint boundary condition 
in terms of the two constant $n\times n$ matrices $A$ and $B$
used in the definition of the regular solution $\varphi(k,x)$ to \eqref{2.1}
satisfying the initial conditions
\begin{equation}
\label{2.6}
\varphi(k,0)=A,\quad \varphi'(k,0)=B.
\end{equation}
We then state the selfadjoint boundary condition as
\begin{equation}
\label{2.7}
-B^\dagger \psi(0)+A^\dagger \psi'(0)=0,
\end{equation}
with the understanding that the matrices $A$ and $B$ satisfy
\begin{equation}
\label{2.8}
-B^\dagger A+A^\dagger B=0,\end{equation}
\begin{equation}
\label{2.9}
A^\dagger A+B^\dagger B>0.\end{equation}
We refer to $A$ and $B$ as the boundary matrices.
We use the positivity of a matrix to mean the positive definiteness, and we use the nonnegativity of
a matrix to mean the positive semidefiniteness.
We recall that an $n\times n$
matrix is positive if and only if
all its eigenvalues are positive and that it is 
nonnegative if and only if all its eigenvalues
are real and nonnegative.
The boundary condition \eqref{2.7} is uniquely determined by
the boundary matrices $A$ and $B.$ On the other hand, the 
boundary matrices $A$ and $B$ can be simultaneously
multiplied on the right by an $n\times n$ invertible matrix $T$ without affecting \eqref{2.7}--\eqref{2.9}.
In fact, this is the only freedom we have in choosing the boundary matrices $A$ and $B.$

By using $A=0$ and $B=I$ in the initial conditions \eqref{2.6}, we obtain the Dirichlet boundary condition
$\psi(0)=0.$ We remark that we use $0$ for the scalar zero, the zero vector with $n$ components, and
the $n\times n$ zero matrix, depending on the context, and we use $I$ to denote the $n\times n$ identity matrix.
We use the terms Dirichlet and purely Dirichlet interchangeably, where the latter
is used to make a contrast
with a boundary condition consisting of a mixture of Dirichlet and non-Dirichlet components.

We recall that we always assume that the potential $V$ is integrable and satisfies
\eqref{2.2}.
In that case, for each $k \in \overline{\mathbb C^+}\setminus\{0\}$ the Schr\"odinger equation \eqref{2.1} has 
a unique $n\times n$ matrix-valued particular solution, known as the Jost solution $f(k,x),$
satisfying the spacial
asymptotics
\begin{equation}
\label{2.10}
f(k,x)=e^{ikx}\left[ I+o(1)\right], \quad f'(k,x)= e^{ikx} \left[ik I+o(1)\right],
\qquad x\to +\infty.
\end{equation}
The regular solution $\varphi(k,x)$ to \eqref{2.1} satisfying the
initial conditions \eqref{2.6} is entire in 
$k\in\mathbb C$ for each fixed $x\in[0,+\infty).$
In general, the Jost solution $f(k,x)$ does not satisfy
the boundary condition \eqref{2.7}, but the regular solution
$\varphi(k,x)$ satisfies \eqref{2.7} at any $k\in\mathbb C.$
The Jost matrix $J(k)$ associated with
\eqref{2.1} and \eqref{2.7} is defined as
\begin{equation}
\label{2.11}
J(k):=f(-k^*,0)^\dagger B-f'(-k^*,0)^\dagger A,\qquad k\in\mathbb R\setminus\{0\},
\end{equation}
where the asterisk denotes complex conjugation. The $n\times n$ matrix-valued
 $J(k)$ has an extension from $ k\in\mathbb R\setminus\{0\}$
 to $k\in\bCp\setminus\{0\},$ and in fact we use the asterisk in
\eqref{2.11} to indicate how that
extension occurs.
The Jost matrix $J(k)$ is analytic in $k\in \bCp$ and continuous in $\bCpb\setminus\{0\}.$
The scattering matrix $S(k)$ associated with \eqref{2.1} and \eqref{2.7}
is defined as
\begin{equation}
\label{2.12}
S(k):=-J(-k)\,J(k)^{-1},\qquad k\in\bR\setminus\{0\},
\end{equation}
and it is continuous in $k\in\mathbb R\setminus\{0\}.$
The physical solution $\Psi(k,x)$ associated with
\eqref{2.1} and \eqref{2.7} is defined as
\begin{equation}
\label{2.13}
\Psi(k,x):=f(-k,x)+f(k,x)\,S(k),\qquad k\in\mathbb R\setminus\{0\}.
\end{equation}
For each fixed $x\in\mathbf R^+,$ the physical
solution $\Psi(k,x)$ has an  extension from $k\in\mathbb R\setminus\{0\}$ to
$k\in\mathbb C^+,$ and that extension is meromorphic in $k\in\bCp.$ 

If the potential $V$ is further assumed to belong to
$L^1_1(\mathbb R^+),$ then the 
Jost solution
$f(k,x),$ the Jost matrix $J(k),$ the scattering matrix $S(k),$ and the physical solution
$\Psi(k,x)$ are also defined at $k=0,$ and those four quantities are continuous
at $k=0.$
On the other hand, if the potential $V$ satisfies \eqref{2.2} and \eqref{2.3} but not \eqref{2.4}, as illustrated
in Example~\ref{example9.11} those four quantities may not have continuous extensions
to $k=0.$

The regular solution $\varphi(k,x)$
can be expressed in terms of the physical solution $\Psi(k,x)$ and the Jost matrix $J(k)$ as
\begin{equation}
\label{2.14}
\varphi(k,x)=-\ds\frac{1}{2ik}\,\Psi(k,x)\,J(k),\qquad k\in\bCpb\setminus\{0\},
\end{equation}
or in terms of the Jost solution $f(k,x)$ and the Jost matrix $J(k)$ as
\begin{equation}
\label{2.15}
\varphi(k,x)=\ds\frac{1}{2ik}\left[f(k,x)\,J(-k)-f(-k,x)\,J(k)\right],
\qquad k\in\mathbb R\setminus\{0\}.
\end{equation}
When the potential $V$ is locally integrable for $x\in[0,+\infty),$ it follows \cite{AW2021} that
the regular solution is defined for $k\in\mathbb C$ and in fact it is entire in $k$ for each $x\in[0,+\infty).$
Thus, when $V$ belongs to $L^1(\mathbb R^+)$ but not
to $L^1_1(\mathbb R^+),$ even though the zero-energy Jost solution $f(0,x)$ or
the zero-energy physical solution $\Psi(0,x)$ may not exist and the right-hand sides may not be defined at $k=0,$
the regular solution $\varphi(k,x)$ is well defined at $k=0.$

\section{The bound states of the Schr\"odinger operator}
\label{section3}

In this section we assume that 
the potential is integrable and satisfies \eqref{2.2}. We present the relevant preliminary material related to the
bound states of the half-line matrix Schr\"odinger operator. 

A bound-state solution
corresponds to a square-integrable column-vector solution
to \eqref{2.1} satisfying the boundary condition \eqref{2.7}. We let
$\lambda:=k^2.$ When the potential satisfies \eqref{2.2} and \eqref{2.4}, 
there are no bound states when $\lambda\ge 0,$ and there are at most a finite number 
of bound states when $\lambda<0$ occurring at the $k$-values on the positive imaginary axis 
of the complex $k$-plane. The bound-state $k$-values correspond to
the zeros of the determinant $\det[J(k)]$ on the positive imaginary axis.
We use
$N$ to denote the number of zeros of $\det[J(k)]$ occurring when $k=i\kappa_j$ for
$1\le j\le N,$ with
$\kappa_j$ being distinct positive
constants.
We remark that $N$ may be zero, finite, or infinite.
 If there are no bound states, then we have $N=0.$
Thus, $\det[J(i\kappa_j)]=0$ and we use
$m_j$ to denote the number of linearly independent vectors
in $\mathbb C^n$ spanning the kernel $\text{\rm{Ker}}[J(i\kappa_j)].$
We remark that $m_j$ is a positive integer
satisfying $1\le m_j\le n,$ and it corresponds to
the multiplicity of the bound state at $k=i\kappa_j.$
Thus, $N$ describes the number of bound states
without counting the multiplicities. The total number of bound states
including the multiplicities, denoted by $\mathcal N,$ is given by
\begin{equation*}
\mathcal N:=\ds\sum_{j=1}^N m_j.
\end{equation*}

The $n\times n$ matrices
$J(i\kappa_j)$ and $J(i\kappa_j)^\dagger$ each have the rank equal to $n-m_j.$ 
When $n=1$ we must have $m_j=1,$ in which case
$J(i\kappa_j)$ and $J(i\kappa_j)^\dagger$ are both equal to the
scalar zero. If $n\ge 2,$
then the kernels
$\text{\rm{Ker}}[J(i\kappa_j)]$ and
$\text{\rm{Ker}}[J(i\kappa_j)^\dagger]$
are in general different from each other, but they have the same
dimension equal to $m_j.$
We use $Q_j$ and $P_j$ to denote the orthogonal projections
onto $\text{\rm{Ker}}[J(i\kappa_j)]$ and
$\text{\rm{Ker}}[J(i\kappa_j)^\dagger],$ respectively.
The definition of the orthogonal projection indicates that we have
\begin{equation}
\label{3.2}
P_j^2=P_j^\dagger=P_j,\quad Q_j^2=Q_j^\dagger=Q_j.
\end{equation}

The bound-state solutions to
\eqref{2.1} can be described either via
the Marchenko theory or the Gel'fand--Levitan theory.
In the Marchenko theory the normalization
matrices for the bound states are obtained by
normalizing the Jost solution $f(k,x)$ at the bound states,
whereas in the Gel'fand--Levitan theory
the normalization
matrices for the bound states are obtained by
normalizing the regular solution $\varphi(k,x)$ at the bound states.
We use 
$M_j$ and $\Psi_j(x)$ to denote the $n\times n$
Marchenko normalization matrix and
the corresponding $n\times n$ normalized matrix solution
at the bound state with $k=i\kappa_j,$ respectively.
Similarly, we use
$C_j$ and $\Phi_j(x)$ to denote the
Gel'fand--Levitan normalization matrix and
the corresponding normalized matrix solution
at the bound state with $k=i\kappa_j,$ respectively.

At the bound state $k=i\kappa_j,$ the Schr\"odinger equation \eqref{2.1} is given by
\begin{equation} 
\label{3.3}
-\psi''+V(x)\,\psi=-\kappa_j^2\,\psi,\qquad x\in\mathbb R^+,\quad 1\le j\le N.
\end{equation}
In general, the solution 
$f(i\kappa_j,x)$ to \eqref{3.3} does not satisfy the boundary condition \eqref{2.7}. On the other hand, $f(i\kappa_j,x)$ is exponentially decaying
asymptotically, i.e. we have
\begin{equation}
\label{3.4}
f(i\kappa_j,x)=e^{-\kappa_j x}\left[I+o(1)\right],\qquad x\to+\infty,\quad 1\le j\le N.
\end{equation}
Hence, for large $x$-values,
the $n$ columns of
$f(i\kappa_j,x)$ are linearly independent and they are each exponentially
decaying solutions to \eqref{3.3}. However, those columns do not, in general, 
satisfy the boundary condition \eqref{2.7}.
We construct an $n\times n$ matrix solution to
\eqref{3.3} satisfying the boundary condition \eqref{2.7} and exponentially decaying as $x\to+\infty.$
Such a matrix solution, denoted by $\Psi_j(x),$ is obtained with the help of $f(i\kappa_j,x)$ and
the $n\times n$ Marchenko normalization matrix $M_j,$ and it is given by
\begin{equation}
\label{3.5}
\Psi_j(x):=f(i\kappa_j,x)\,M_j,\qquad 1\le j\le N,
\end{equation}
where we refer to the $n\times n$ matrix
$\Psi_j(x)$ as the Marchenko
normalized  bound-state solution
to the Schr\"odinger equation \eqref{3.3}
at $k=i\kappa_j.$ 
Although each column of the $n\times n$ matrix $\Psi_j(x)$ is a solution to
\eqref{3.3} and satisfies \eqref{2.7}, among its $n$ columns
we only have $m_j$ linearly independent column-vector solutions to \eqref{3.3} satisfying
the boundary condition \eqref{2.7}. 
Any nontrivial
square-integrable column-vector solution to \eqref{3.3} satisfying \eqref{2.7} can be written as $\Psi_j(x)\,\nu$ for a constant nonzero
column vector $\nu \in \mathbb C^n.$ 
The Marchenko normalized bound-state solutions $\Psi_j(x)$ for $1\le j\le N$
satisfy the normalization condition given by
\begin{equation}
\label{3.6}
\int_0^\infty dx\,\Psi_j(x)^\dagger\,\Psi_j(x)=P_j,\qquad 1\le j\le N,
\end{equation}
as well as the orthogonality condition given by
\begin{equation}
\label{3.7}
\int_0^\infty dx\,\Psi_j(x)^\dagger\,\Psi_l(x)=0,\qquad j\ne l,
\end{equation} 
where we have $1\le j\le N$ and $1\le l\le N.$
For further information on the description of bound states
in terms of the Marchenko normalization matrices, we refer the reader to
\cite{AW2021} and \cite{AW2025}.

The Marchenko normalization matrix $M_j$ is constructed
as follows.
Using the orthogonal projection matrix
$P_j$ and the Jost solution $f(i\kappa_j,x),$ we obtain the $n\times n$ hermitian matrix $\mathbf A_j$ given by
\begin{equation}
\label{3.8}
\mathbf A_j:=\int_0^\infty dx\,P_j\,f(i\kappa_j,x)^\dagger\,f(i\kappa_j,x)\,P_j,\qquad 1\le j\le N.
\end{equation}
Using $\mathbf A_j$ and $P_j,$ we define the $n\times n$ hermitian matrix $\mathbf B_j$ as
\begin{equation}
\label{3.9}
\mathbf B_j:=I-P_j+\mathbf A_j,\qquad 1\le j\le N.
\end{equation}
It is known that $\mathbf B_j$ is a positive matrix, and hence 
there exists a unique positive hermitian matrix $\mathbf B_j^{1/2}$ satisfying
\begin{equation}
\label{3.10}\mathbf B_j^{1/2}\,\mathbf B_j^{1/2}=\mathbf B_j,\qquad 1\le j\le N.
\end{equation}
Since $\mathbf B_j^{1/2}$ is positive, its inverse exists and is denoted by
$\mathbf B_j^{-1/2}.$ It is known that each of
$\mathbf B_j,$ $\mathbf B_j^{1/2},$ and $\mathbf B_j^{-1/2}$ commutes with
$P_j.$ The Marchenko normalization matrix $M_j$ is defined in terms of  $\mathbf B_j^{-1/2}$ and $P_j$
as
\begin{equation}
\label{3.11}M_j:=\mathbf B_j^{-1/2}\,P_j,\qquad 1\le j\le N.
\end{equation}
The Marchenko normalization matrix
 $M_j$ is an $n\times n$ nonnegative matrix of rank $m_j.$
Hence, exactly $m_j$ of its eigenvalues are positive and the remaining $n-m_j$ eigenvalues are zero.

The description of the bound states pertinent to the Gel'fand--Levitan method is as follows.
We construct the $n\times n$ matrix-valued Gel'fand--Levitan normalized bound-state solutions
$\Phi_j(x)$ with the help of the regular solution
$\varphi(k,x)$ and the Gel'fand--Levitan normalization matrices $C_j$ as
\begin{equation}
\label{3.12}
\Phi_j(x):=\varphi(i\kappa_j,x)\,C_j,\qquad 1\le j\le N,
\end{equation}
where $\varphi(i\kappa_j,x)$ is the value of the regular solution evaluated at $k=i\kappa_j.$
We refer to the $n\times n$ matrix
$\Phi_j(x)$ as the Gel'fand--Levitan normalized bound-state solution
 to the Schr\"odinger equation \eqref{3.3} at $k=i\kappa_j.$
 Although each column of the regular solution $\varphi(i\kappa_j,x)$ to \eqref{3.3} satisfies the boundary condition \eqref{2.7}, 
 unless we have $m_j=n$ those columns in general become unbounded
as $x\to+\infty,$ and we have
\begin{equation*}
\varphi(i\kappa_j,x)=O(e^{\kappa_j x}),\qquad x\to+\infty.
\end{equation*}
The presence of the $n\times n$ matrix $C_j$ of rank $m_j$ in \eqref{3.12} ensures that 
each column of $\Phi_j(x),$ besides being a solution to \eqref{3.3} and satisfying the boundary condition
\eqref{2.7}, decays exponentially as $O(e^{-\kappa_j x})$ as $x\to+\infty.$
There are exactly $m_j$ linearly independent columns of $\Phi_j(x)$ among its $n$ columns.
The Gel'fand--Levitan normalized bound-state solutions $\Phi_j(x)$ satisfy the normalization property given by
\begin{equation}
\label{3.14}
\int_0^\infty dx\,\Phi_j(x)^\dagger\,\Phi_j(x)=Q_j,\qquad 1\le j\le N,
\end{equation}
as well as the orthogonalization property
\begin{equation}
\label{3.15}
\int_0^\infty dx\,\Phi_j(x)^\dagger\,\Phi_l(x)=0,\qquad j\ne l,
\end{equation} 
for $1\le j\le n$ and $1\le l\le N.$

The Gel'fand--Levitan normalization matrix $C_j$ is constructed as follows. Using the orthogonal projection matrix
$Q_j$ and the regular solution $\varphi(i\kappa_j,x)$ to \eqref{3.3}, we form the $n\times n$ hermitian matrix $\mathbf G_j$ as
\begin{equation}
\label{3.16}\mathbf G_j:=\int_0^\infty dx\,Q_j\,\varphi(i\kappa_j,x)^\dagger\,\varphi(i\kappa_j,x)\,Q_j,\qquad 1\le j\le N.
\end{equation}
Then, using $\mathbf G_j$ and $Q_j$ we form the $n\times n$ hermitian matrix
$\mathbf H_j$ as
\begin{equation}
\label{3.17}
\mathbf H_j:=I-Q_j+\mathbf G_j,\qquad 1\le j\le N.
\end{equation}
It is known that $\mathbf H_j$ is a positive matrix. Hence, 
there exists a unique positive hermitian matrix $\mathbf H_j^{1/2}$ so that
\begin{equation}
\label{3.18}\mathbf H_j^{1/2}\,\mathbf H_j^{1/2}=\mathbf H_j,\qquad 1\le j\le N.
\end{equation}
Since the matrix $\mathbf H_j^{1/2}$ is positive, its inverse exists and is denoted by
$\mathbf H_j^{-1/2}.$ Each of
$\mathbf H_j,$ $\mathbf H_j^{1/2},$ and $\mathbf H_j^{-1/2}$ commutes with
$Q_j.$ 
The Gel'fand--Levitan normalization matrix $C_j$ is constructed using  $\mathbf H_j^{-1/2}$ and
$Q_j$ as
\begin{equation}
\label{3.19}
C_j:=\mathbf H_j^{-1/2}\,Q_j,\qquad 1\le j\le N.
\end{equation}
The $n\times n$ matrix $C_j$ is hermitian and nonnegative, and it has rank
equal to $m_j,$ which is the same as the rank of $Q_j.$

The Marchenko normalized bound-state solution $\Psi_j(x)$ is not identical to 
the Gel'fand--Levitan normalized bound-state solutions $\Phi_j(x),$ but each of them contains
$m_j$ linearly independent columns. Those two $n\times n$ matrix-valued bound-state solutions can be expressed 
in terms of each other with the help of the $n\times n$ dependency
matrix $D_j$ via
\begin{equation}
\label{3.20}
\Phi_j(x)=\Psi_j(x)\,D_j, \qquad 1\le j\le N.
\end{equation}
In order to define $D_j$ uniquely, we need to impose one further restriction on it.
For example, we can require that $D_j$ satisfies
\begin{equation}
\label{3.21}
D_j=P_j\,D_j, \qquad 1\le j\le N,
\end{equation}
where we recall that the $n\times n$ matrix $P_j$ is the orthogonal projection onto $\text{\rm{Ker}}[J(i\kappa_j)^\dagger]$ and
appearing in \eqref{3.2}. 
Then, the dependency matrix satisfying \eqref{3.20} and \eqref{3.21} is uniquely determined.
In other words, when we are given
the Schr\"odinger equation \eqref{2.1} with the potential
$V$ satisfying
\eqref{2.2} and \eqref{2.3} and we are also given the selfadjoint boundary condition \eqref{2.7}, then there exists
a unique dependency matrix $D_j$ satisfying \eqref{3.20} and \eqref{3.21} associated with
each bound state at $k=i\kappa_j$ for $1\le j\le N.$
The dependency matrix $D_j$ can be constructed by using the fact that
there exists a bijection $\alpha\mapsto\beta$ from the kernel of $J(i\kappa_j)$ onto the kernel of $J(i\kappa_j)^\dagger$ 
in such a way that
\begin{equation*}
\varphi(i\kappa_j,x)\,\alpha= f(i\kappa_j,x)\,\beta,\qquad 1\le j\le N.
\end{equation*}
For the explicit construction
of $D_j$ associated with the corresponding Schr\"odinger operator, 
we refer the reader to \cite{AW2025}.
From \eqref{3.4} we know that the Jost solution $f(i\kappa_j,x)$ is an invertible
$n\times n$ matrix when $x$ is large. We can express
$D_j$ in terms of the Marchenko normalization matrix $M_j,$ the Gel'fand--Levitan normalized bound-state wavefunction
$\Phi_j(x),$ and the matrix quantity $f(i\kappa_j,x)$ as \cite{AW2025}
\begin{equation}
\label{3.23}
D_j= M^+_j\, f(i\kappa_j,x)^{-1}\,   \Phi_j(x), \qquad 1\le j\le N,
\end{equation}
where we evaluate the right-hand side at any $x$-value at which the matrix $f(i\kappa_j,x)$ is invertible and 
we note that  $M^+_j$ denotes \cite{BG2003} the Moore--Penrose inverse of $M_j.$ 
As seen from \eqref{3.20}, we can write \eqref{3.23} in the equivalent form as
\begin{equation}
\label{3.24}
D_j= \Psi_j(x)^+\,\Phi_j(x), \qquad 1\le j\le N.
\end{equation}
If the matrix $f(i\kappa_j,x)$ is not invertible
at a particular $x$-value
but the matrix
$f'(i\kappa_j,x)$ is invertible there, then we can evaluate
$D_j$ by using
\begin{equation}
\label{3.25}
D_j= M^+_j f'(i\kappa_j,x)^{-1} \,  \Phi'_j(x), \qquad 1\le j\le N,
\end{equation}
where we evaluate the right-hand side at any $x$-value at which the matrix $f'(i\kappa_j,x)$ is invertible.
From \eqref{3.20}, we see that we can write \eqref{3.25} in the equivalent form as
\begin{equation*}
D_j= [\Psi'_j(x)]^+\,\Phi_j(x), \qquad 1\le j\le N.
\end{equation*}

We remark that $D_j$ satisfies various properties such as
\begin{equation}
\label{3.27}
D_j=D_j\,Q_j, \qquad 1\le j\le N,
\end{equation}
where we recall that the $n\times n$ matrix $Q_j$ is the orthogonal projection
onto $\text{\rm{Ker}}[J(i\kappa_j)].$ 
Each of the matrix products $D_j^\dagger D_j$ and $D_j D_j^\dagger$
is an orthogonal projection, and we have
\begin{equation}
\label{3.28}
D_j^\dagger D_j=Q_j,\quad D_j D_j^\dagger=P_j, \qquad 1\le j\le N
.\end{equation}
The matrices $D_j$ and $D_j^\dagger$
each have rank $m_j,$ which is the same as the common rank
of the orthogonal projection matrices $P_j$ and $Q_j.$
The adjoint matrix $D_j^\dagger$
coincides with 
the Moore--Penrose inverse of $D_j,$ i.e. we have
\begin{equation}
\label{3.29}
D_j^+=D_j^\dagger, \qquad 1\le j\le N.\end{equation}
Analogous to \eqref{3.20}, where $\Phi_j(x)$ is expressed in terms of $\Psi_j(x),$ we have 
\begin{equation}
\label{3.30}
\Psi_j(x)=\Phi_j(x) D_j^\dagger, \qquad 1\le j\le N,\end{equation}
which expresses $\Psi_j(x)$ in terms of $\Phi_j(x).$

\section{The Gel'fand--Levitan system of integral equations}
\label{section4}

In this section we assume that the potential $V$ satisfies
\eqref{2.2} and at least \eqref{2.4}. We present the basic relevant material related to the spectral measure for the half-line matrix
Schr\"odinger operator described by \eqref{2.1} and \eqref{2.7}. We also mention
Parseval's equality involving the regular solution and
the Gel'fand--Levitan normalized bound-state solutions as well as
another version of Parseval's equality involving the physical solution and
the Marchenko normalized bound-state solutions. We show how
a perturbation of the spectral measure results in the perturbation of all
relevant quantities related to the corresponding Schr\"odinger operator.
We also indicate how those perturbations can be expressed in terms of the solution to
a system of integral equations, which we refer to the matrix-valued Gel'fand--Levitan system.
For the details and further elaborations, we refer the reader to \cite{AW2025}.

We recall that when the potential $V$ in \eqref{2.1} satisfies \eqref{2.2} and \eqref{2.4}, the
 Schr\"odinger operator associated with \eqref{2.1} and \eqref{2.7} has at most a finite number of distinct bound states
occurring at $k=i\kappa_j$ for $1\le j\le N.$ It is understood that $N=0$ if there are no
bound states. The spectral measure $d\rho$ associated with the corresponding Schr\"odinger operator is given by
\begin{equation}
\label{4.1}
d\rho:=\begin{cases}
\ds\frac{\sqrt{\lambda}}{\pi}\,\left(J(\sqrt{\lambda})^\dagger\,J(\sqrt{\lambda})\right)^{-1}\,d\lambda,\qquad \lambda\ge 0,
\\
\noalign{\medskip}
\ds\sum_{j=1}^N C_j^2\,\delta(\lambda-\lambda_j)\,d\lambda,
\qquad \lambda<0,\end{cases}
\end{equation}
where $\delta(\cdot)$ denotes the Dirac delta distribution and we recall that 
$J(k)$ is the Jost matrix
given in \eqref{2.11}, the quantity $C_j$ 
is the Gel'fand--Levitan normalization matrix defined in \eqref{3.19},
$C_j^2$ is the matrix multiplication of $C_j$ with itself, and we let
$\lambda:=k^2$ and $\lambda_j:=-\kappa_j^2$ for $1\le j\le N.$
Using the spectral measure $d\rho,$ we express the completeness relation, known as 
 Parseval's equality, as
\begin{equation*}
\ds\int_{\lambda\in\mathbb R^+} \varphi(k,x)\,d\rho\,
\varphi(k,y)^\dagger+\ds\sum_{j=1}^N \Phi_j(x)\,\Phi_j(y)^\dagger=\delta(x-y)\,I,
\end{equation*}
where the integration is from $\lambda=0$ to $\lambda=+\infty,$ and we recall that
$\varphi(k,x)$ is the regular solution to \eqref{2.1} satisfying the initial
conditions \eqref{2.6} and that $\Phi_j(x)$ is the $n\times n$ Gel'fand--Levitan normalized bound-state  solution at $k=i\kappa_j$ and appearing in \eqref{3.12}.
We remark that Parseval's equality can also be expressed in terms of the
physical solution $\Psi(k,x)$ appearing in \eqref{2.13}
and the $n\times n$ Marchenko normalized bound-state  solutions $\Psi_j(x)$
appearing in \eqref{3.5}, and we have
\begin{equation*}
\ds\frac{1}{2\pi} \ds\int_0^\infty
dk\, \Psi(k,x)\,
\Psi(k,y)^\dagger+\ds\sum_{j=1}^N \Psi_j(x)\,\Psi_j(y)^\dagger=\delta(x-y)\,I.
\end{equation*}

Let us view $V(x),$ $\varphi(k,x),$ $(A,B),$ $f(k,x),$ $J(k),$ $S(k),$ and $d\rho$
as the potential, the regular solution, the pair of boundary
matrices, the Jost solution, the Jost matrix, the scattering matrix, and the spectral measure, respectively, corresponding to the unperturbed problem.
Let us assume that the unperturbed problem has $N$ bound states at
$k=i\kappa_j$ for $1\le j\le N,$ and that the quantities relevant to the bound state at $k=i\kappa_j$
are $\lambda_j:=-\kappa_j^2,$ the multiplicity $m_j,$ the Gel'fand--Levitan normalization matrix $C_j,$ the orthogonal projection
$Q_j$ onto $\text{\rm{Ker}}[J(i\kappa_j)],$ the orthogonal projection
$P_j$ onto $\text{\rm{Ker}}[J(i\kappa_j)^\dagger],$ and the Gel'fand--Levitan normalized bound-state  solution
 $\Phi_j(x).$ 
Thus, the regular solution $\varphi(k,x)$ 
satisfies the unperturbed Schr\"odinger equation \eqref{2.1}
and the unperturbed initial conditions \eqref{2.6}.

Let us view $\tilde V(x),$
$\tilde\varphi(k,x),$ $(\tilde A,\tilde B),$ $\tilde f(k,x),$ $\tilde J(k),$ $\tilde S(k),$ and
$d\tilde\rho$ as the potential, the regular solution, the pair of boundary
matrices, the Jost solution, the Jost matrix, the scattering matrix, and the spectral measure, respectively, corresponding to the perturbed problem.
Let us assume that the perturbed problem has $\tilde N$ bound states at
$k=i\tilde\kappa_j$ for $1\le j\le \tilde N,$ and that the quantities relevant to the bound state at $k=i\tilde\kappa_j$
are $\tilde\lambda_j:=-\tilde\kappa_j^2,$ the multiplicity $\tilde m_j,$ the Gel'fand--Levitan normalization matrix $\tilde C_j,$ the orthogonal projection
$\tilde Q_j$ onto $\text{\rm{Ker}}[\tilde J(i\tilde\kappa_j)],$ the orthogonal projection
$\tilde P_j$ onto $\text{\rm{Ker}}[\tilde J(i\tilde\kappa_j)^\dagger],$ and the Gel'fand--Levitan normalized bound-state  solution
 $\tilde\Phi_j(x).$ We remark that the perturbed quantities are defined in the same way the unperturbed quantities are defined, but by using 
 the relevant perturbed quantities as input to those definitions.
Thus, the perturbed Schr\"odinger equation is given by
\begin{equation}
\label{4.4}
-\tilde\psi''+\tilde V(x)\,\tilde\psi=k^2\tilde\psi,\qquad x\in\bR^+,
\end{equation}
and the perturbed boundary condition is described by
\begin{equation*}
-\tilde B^\dagger \tilde\psi(0)+\tilde A^\dagger \tilde\psi'(0)=0,
\end{equation*}
where the constant $n\times n$ perturbed boundary matrices $\tilde A$ and $\tilde B$ satisfy
\begin{equation}
\label{4.6}
-\tilde B^\dagger \tilde A+\tilde A^\dagger \tilde B=0,\end{equation}
\begin{equation}
\label{4.7}
\tilde A^\dagger \tilde A+\tilde B^\dagger \tilde B>0.\end{equation}
The perturbed Jost solution $\tilde f(k,x)$ is the solution to \eqref{4.4} satisfying the
spacial asymptotics
\begin{equation*}
\tilde f(k,x)=e^{ikx}\left[ I+o(1)\right], \quad \tilde f'(k,x)= e^{ikx} \left[ik I+o(1)\right],
\qquad x\to +\infty.
\end{equation*}
The perturbed Jost matrix $\tilde J(k)$ is defined as
\begin{equation*}
\tilde J(k):=\tilde f(-k^*,0)^\dagger\,\tilde B-\tilde f'(-k^*,0)^\dagger\,\tilde A,\qquad k\in\mathbb R,
\end{equation*}
and the perturbed scattering matrix $\tilde S(k)$ is given by
\begin{equation*}
\tilde S(k):=-\tilde J(-k)\,\tilde J(k)^{-1},\qquad k\in\bR.
\end{equation*}
The perturbed Gel'fand--Levitan normalized bound state solution at $k=i\tilde\kappa_j$ is defined as
\begin{equation*}
\tilde\Phi_j(x):=\tilde\varphi(i\tilde\kappa_j,x)\,\tilde C_j,\qquad 1\le j\le \tilde N,
\end{equation*}
and the perturbed spectral measure $d\tilde\rho$ is given by
\begin{equation}
\label{4.12}
d\tilde\rho:=\begin{cases}
\ds\frac{\sqrt{\lambda}}{\pi}\,\left(\tilde J(\sqrt{\lambda})^\dagger\,\tilde J(\sqrt{\lambda})\right)^{-1}\,d\lambda,\qquad \lambda\ge 0,
\\
\noalign{\medskip}
\ds\sum_{j=1}^{\tilde N} \tilde C_j^2\,\delta(\lambda-\tilde\lambda_j)\,d\lambda,
\qquad \lambda<0.\end{cases}
\end{equation}
The perturbed regular solution $\tilde\varphi(k,x)$ 
satisfies the perturbed Schr\"odinger equation \eqref{4.4} and also satisfies 
the perturbed initial conditions
\begin{equation}
\label{4.13}
\tilde\varphi(k,0)=\tilde A,\quad \tilde\varphi'(k,0)=\tilde B.\end{equation}

We assume that the unperturbed potential $V$ and the perturbed potential $\tilde V$
both satisfy \eqref{2.2} and 
we also have $V\in L^1_1(\mathbb R^+)$ and $\tilde V\in L^1(\mathbb R^+).$
We suppose that the perturbation is initiated by the change $d\tilde\rho-d\rho$ in the spectral measure.
We also assume that the perturbed regular solution $\tilde\varphi(k,x)$ is related 
 to the unperturbed regular solution $\varphi(k,x)$ through an $n\times n$ matrix-valued quantity $\mathcal A(x,y)$
via the relation given by 
\begin{equation}
\label{4.14}
\tilde\varphi(k,x)=\varphi(k,x)+\int_0^x dy\,\mathcal A(x,y)\,\varphi(k,y),
\end{equation}
and that the perturbed potential $\tilde V(x)$ is related to the unperturbed
potential $V(x)$ as
\begin{equation}
\label{4.15}
\tilde V(x)=V(x)+2\,\ds\frac{d\, \mathcal A(x,x)}{dx},
\end{equation}
 where we use $\mathcal A(x,x)$ to denote $\mathcal A(x,x^-).$ 
By choosing $\mathcal A(x,y)$ appearing in \eqref{4.14} appropriately, we would like
to relate the perturbed boundary matrices
$\tilde A$ and $\tilde B$ to the unperturbed 
boundary matrices
$A$ and $B$ through the quantity $\mathcal A(x,y).$
We also would like to determine 
how all other relevant perturbed quantities are related
to the corresponding unperturbed quantities.
It turns out that the appropriate $\mathcal A(x,y)$ is determined as
the solution to
 the Gel'fand--Levitan system of integral equations given by
\begin{equation}
\label{4.16}
\mathcal A(x,y)+G(x,y)+\int_0^x dz\,\mathcal A(x,z)\,G(z,y)=0,\qquad 0\le y< x,\end{equation}
where we have defined
\begin{equation}
\label{4.17}
G(x,y):=\int_{\lambda\in\mathbb R} \varphi(k,x)\left[d\tilde\rho-d\rho\right]\varphi(k,y)^\dagger,
\end{equation}
with $d\rho$ and $d\tilde\rho$ denoting the unperturbed and perturbed spectral measures, respectively.
and $\varphi(k,x)$ is the unperturbed regular solution satisfying \eqref{2.1} and
the initial conditions given in \eqref{2.6}.
Then, the constant $n\times n$ matrices $\tilde A$ and $\tilde B$ appearing in 
\eqref{4.13} are related to
the constant $n\times n$ matrices $A$ and $B$ appearing in \eqref{2.6} as
\begin{equation}
\label{4.18}
\tilde A=A,\quad 
\tilde B=B+\mathcal A(0,0)\,A.
\end{equation}

 \section{The removal of a bound state}
\label{section5}

In this section we assume that the potential $V$ satisfies \eqref{2.2} and at least \eqref{2.4}.
We describe the changes in the relevant quantities related to
the Schr\"odinger operator associated with \eqref{2.1} and \eqref{2.7}, when we perturb
the spectrum of that operator by removing one of the bound states completely without changing
the rest of the discrete spectrum and without changing the continuous spectrum.
Without loss of generality, we remove the bound state at $k=i\kappa_N$ with the
Gel'fand--Levitan normalization matrix $C_N$ without changing the rest of the spectrum.

Toward our goal, we find it convenient to use the Gel'fand--Levitan system \eqref{4.16} and its solution
$\mathcal A(x,y)$ by choosing the change in the spectral measure as
\begin{equation}
\label{5.1}
d\tilde\rho-d\rho=-C_N^2\,\delta(\lambda-\lambda_N),\qquad \lambda\in\mathbb R,\end{equation}
where we recall that $\lambda_N:=-\kappa_N^2.$
We remark that \eqref{5.1} is suggested by \eqref{4.1} and \eqref{4.12}
by assuming that the Gel'fand--Levitan normalization matrices $C_j$
for the remaining bound states at $k=i\kappa_j$ with $1\le j\le N-1$ do not change.
In fact, it turns out that if we use \eqref{5.1} in the construction 
of the input matrix $G(x,y)$ defined in \eqref{4.17}, solve the resulting Gel'fand--Levitan system \eqref{4.16}, and use
the corresponding solution $\mathcal A(x,y)$ to construct the perturbed
quantities, then for $1\le j\le N-1$ each perturbed Gel'fand--Levitan normalization matrix
$\tilde C_j$ coincides with the unperturbed Gel'fand--Levitan normalization matrix
$C_j.$

Using \eqref{5.1} as input to \eqref{4.17}, we obtain
\begin{equation}
\label{5.2}
G(x,y)=- \varphi(i\kappa_N,x) \,C_N^2\, \varphi(i\kappa_N,y)^\dagger,
\end{equation}
where we recall that $\varphi(k,x)$ is the unperturbed regular solution to \eqref{2.1} satisfying the initial conditions \eqref{2.6}.
Using \eqref{3.12} in \eqref{5.2}, we express $G(x,y)$ in terms of the unperturbed Gel'fand--Levitan normalized
bound-state solution $\Phi_N(x)$ as
\begin{equation}
\label{5.3}
G(x,y)=- \Phi_N(x)\,\Phi_N(y)^\dagger.\end{equation}
Since $G(x,y)$ given in \eqref{5.3} constitutes a separable kernel
for \eqref{4.16}, the solution 
$\mathcal A(x,y)$ to \eqref{4.16} is obtained by using the methods of linear algebra,
and it is given by 
\begin{equation}\label{5.4}
\mathcal A(x,y)= \Phi_N(x) \,W_N(x)^+\,\Phi_N(y)^\dagger, \qquad 0 \le y < x, 
\end{equation}
with $W_N(x)$ defined as
\begin{equation}
\label{5.5}
W_N(x):=\int_x^\infty dz\, \Phi_N(z)^\dagger\,\Phi_N(z),
\end{equation}
where we recall that $W_N(x)^+$ denotes the Moore--Penrose inverse of $W_N(x).$
Using \eqref{5.4} in \eqref{4.14} we obtain the perturbed regular solution $\tilde\varphi(k,x)$ as
 \begin{equation}\label{5.6}
\tilde\varphi(k,x)=\varphi(k,x)+ \Phi_{N}(x)\, W_N(x)^+ \int_0^x dy\,\Phi_{N}(y)^\dagger\,\varphi(k,y).
\end{equation}
It is known \cite{AW2025} that \eqref{5.6} can be equivalently expressed as
\begin{equation}
\label{5.7}
\tilde\varphi(k,x)=\varphi(k,x)+\ds\frac{1}{k^2+\kappa_N^2}
\,\Phi_N(x) \,W_N(x)^+\left[\Phi'_N(x)^\dagger\,\varphi(k,x)-\Phi_N(x)^\dagger
\,\varphi'(k,x)\right],
\end{equation}
with the understanding that the term $k^2+\kappa_N^2$ in the denominator causes
only removable singularities.
Using \eqref{5.4} in \eqref{4.15} we obtain the perturbed potential $\tilde Vx)$ as
\begin{equation}
\label{5.8}
\tilde V(x)=V(x)+2\,\ds\frac{d}{dx}\left[
\Phi_N(x) \,W_N(x)^+\,\Phi_N(x)^\dagger
\right].\end{equation}
Next, using \eqref{5.4} in \eqref{4.18}, we express the perturbed boundary matrices
$\tilde A$ and $\tilde B$ as
\begin{equation}\label{5.9}
\tilde A=A, \quad \tilde B=B+A\,C_N^2\, A^\dagger A.
\end{equation}

In the removal of the bound state at $k=i\kappa_N,$ we are interested in the asymptotic behavior of
the potential increment $\tilde V-V.$ For the details we refer the reader to
Theorem~6.3 of \cite{AW2025}.
We recall that the unperturbed potential $V$ satisfies \eqref{2.2} and \eqref{2.4},
the boundary matrices $A$ and $B$ appearing in \eqref{2.7} satisfy \eqref{2.8} and \eqref{2.9},
and the perturbed potential $\tilde V$ is given by \eqref{5.8}.
We have the following:

\begin{enumerate}

\item[\text{\rm(a)}] The perturbed potential $\tilde V$ also satisfies \eqref{2.2}.

\item[\text{\rm(b)}] The potential increment $\tilde V-V$ has the asymptotic behavior
\begin{equation}\label{5.10}
\tilde V(x)-V(x)= O\left(\int_x^\infty dy \,|V(y)| \right), \qquad x \to +\infty.
\end{equation}

\item[\text{\rm(c)}] 
If the unperturbed potential $V$ is further restricted to $L^1_{1+\epsilon}(\mathbb R^+)$ for some fixed $\epsilon \ge 0,$ then 
the perturbed potential $\tilde V$ belongs to $L^1_\epsilon(\mathbb R^+).$ We recall that
the potential class  $L^1_\epsilon(\mathbb R^+)$ consists of potentials $V$ satisfying \eqref{2.5}.

\item[\text{\rm(d)}] If the unperturbed potential $V$ is further restricted to satisfy
\begin{equation}\label{5.11}
|V(x)| \le c\,e^{-\alpha x}, \qquad x \ge x_0,
\end{equation}
for some positive constants $\alpha$ and $x_0$ and with
$c$ denoting a generic constant, then the potential increment $\tilde V(x)-V(x)$ also satisfies
\begin{equation}\label{5.12}
|\tilde V(x)-V(x)| \le c\,e^{-\alpha x}, \qquad x \ge x_0.
\end{equation}
We recall that a generic constant $c$ does not necessarily have the same value in different appearances.
We remark that the estimates in \eqref{5.10}--\eqref{5.12} are not necessarily sharp.

\item[\text{\rm(e)}]
If the support of the unperturbed potential $V$ is contained in the interval $[0,x_0]$
for some positive $x_0,$ then the support of 
the perturbed potential $\tilde V$ is also contained in $[0,x_0].$ 

\end{enumerate}

We refer the reader to \cite{AW2025} for the construction of the remaining relevant quantities
under the further sufficiency assumption that
the unperturbed potential $V$ belongs to $L^1_2(\mathbb R^+).$  
We provide the results from \cite{AW2025} as follows. The perturbed Jost 
matrix $\tilde J(k)$ is related to the unperturbed Jost matrix $J(k)$ as
 \begin{equation}
\label{5.13}
\tilde J(k)=\left[I+\ds\frac{2i\kappa_N}{k-i\kappa_N}\,P_N\right] J(k),\qquad
k\in\overline{\mathbb C^+},\end{equation}
and the determinant of the Jost matrix is changed as
\begin{equation}
\label{5.14}
\det[\tilde J(k)]=\left(\ds\frac{k+i\kappa_N}{k-i\kappa_N}\right)^{m_N} \det[J(k)],
\qquad k\in\overline{\mathbb C^+},
\end{equation}
which indicates that the total number of bound states including the multiplicities 
is decreased by $m_N.$
The perturbed scattering matrix $\tilde S(k)$ is related to the unperturbed
scattering matrix $S(k)$ as
 \begin{equation}
\label{5.15}
\tilde S(k)=\left[I-\ds\frac{2i\kappa_N}{k+i\kappa_N}\,P_N\right] S(k)\left[I-\ds\frac{2i\kappa_N}{k+i\kappa_N}\,P_N\right],
\qquad k\in\mathbb R,\end{equation}
and the determinant of the scattering matrix is changed via
\begin{equation}
\label{5.16}
\det[\tilde S(k)]=\left(\ds\frac{k-i\kappa_N}{k+i\kappa_N}\right)^{2 m_N} \det[S(k)],
\qquad k\in\mathbb R.
\end{equation}
The perturbed Jost solution $\tilde f(k,x)$ is related to the unperturbed Jost solution
$f(k,x)$ as
\begin{equation}
\label{5.17}
\tilde f(k,x)=\left[ f(k,x)-\Phi_N(x) \,W_N(x)^+\ds\int_x^\infty dy\,\Phi_N(y)^\dagger\, f(k,y)\right]
\left[ I+ \ds\frac{2i\kappa_N}{k-i\kappa_N} \,P_N \right],
\end{equation}
which can also be expressed in the alternate form as
\begin{equation}
\label{5.18}
\tilde f(k,x)=\left[ f(k,x)  +\ds\frac{1}{k^2+\kappa_N^2}
\,\Phi_N(x) \,W_N(x)^+\,q_{1}(k,x)\right]
\left[ I+ \ds\frac{2i\kappa_N}{k-i\kappa_N} P_N \right],
\end{equation}
where we have defined
\begin{equation}
\label{5.19}
q_{1}(k,x):=\Phi'_N(x)^\dagger\,f(k,x)-\Phi_N(x)^\dagger
\,f'(k,x).
\end{equation}
Finally, the orthogonal projection matrices $Q_j$ and the Gel'fand--Levitan normalization
matrices $C_j$ for the remaining bound states do not change, i.e. we have
 \begin{equation*}
\tilde Q_j=Q_j,\quad \tilde C_j=C_j,\qquad 1\le j\le N-1.
\end{equation*}

\section{Decreasing the multiplicity of a bound state}
\label{section6}

In this section we assume that the potential $V$ satisfies \eqref{2.2} and at least \eqref{2.4}.
We describe the changes in the relevant quantities related to the Schr\"odinger operator
associated with \eqref{2.1} and \eqref{2.7}, when we perturb the spectrum of our operator by reducing the multiplicity
of a bound state without removing that bound state completely. We do not change the remaining discrete spectrum
and we do not change the continuous spectrum.
We assume that the unperturbed operator has $N$ bound states at $k=i\kappa_j$
with the multiplicity $m_j$ and with the Gel'fand--Levitan normalization matrix $C_j$ for $1\le j\le N.$
Without loss of generality, we assume that we reduce the multiplicity $m_N$ of the bound state
at $k=i\kappa_N$ without completely removing that bound state.
Thus, it is implicitly assumed that $m_N$ satisfies $2\le m_N\le n.$

We assume that the multiplicity of the bound state at $k=i\kappa_N$ is reduced by $m_{N{\text{\rm{r}}}},$ and hence we define the positive
integer representing the reduction in the multiplicity of the bound state at $k=i\kappa_N$ as
\begin{equation}
\label{6.1}
m_{N{\text{\rm{r}}}}:=m_N-\tilde m_N.
\end{equation}
The subscript $N{\text{\rm{r}}}$ indicates that we refer to the $N$th bound state at $k=i\kappa_N$ and that we reduce
its multiplicity.
From \eqref{6.1} we see that $m_{N{\text{\rm{r}}}}$ satisfies the inequality
$1\le m_{N{\text{\rm{r}}}}\le n-1.$
The reduction of the multiplicity of the bound state at $k=i\kappa_N$
by $m_{N{\text{\rm{r}}}}$ is achieved 
by introducing the matrix $Q_{N{\text{\rm{r}}}}$ in such a way that
$Q_{N{\text{\rm{r}}}}$ is an orthogonal projection onto
a proper subspace of 
$Q_N\, \mathbb C^n,$ it has its rank equal to $m_{N{\text{\rm{r}}}},$
and it satisfies $Q_{N{\text{\rm{r}}}}\le Q_N.$
This last matrix inequality is equivalent to 
$Q_N-Q_{N{\text{\rm{r}}}}\ge 0,$ which indicates that
the matrix $Q_N-Q_{N{\text{\rm{r}}}}$ is a nonnegative matrix. In other words,
the eigenvalues of the matrix
 $Q_N-Q_{N{\text{\rm{r}}}}$ are all real and nonnegative.
 For the construction of the matrix $Q_{N{\text{\rm{r}}}}$
 we refer the reader to \cite{AW2025}.

 When the multiplicity of the bound state at $k=i\kappa_N$ is reduced from $m_N$ to $\tilde m_N,$ the transformations
 of the relevant quantities are obtained as follows. For the details, we refer the reader to
 Section~7 of \cite{AW2025}. Since the matrix $Q_{N{\text{\rm{r}}}}$ is an orthogonal projection onto
 a proper subspace of
 $Q_N\,\mathbb C^n,$ it satisfies the matrix equalities
 \begin{equation}
\label{6.2}
Q_{N{\text{\rm{r}}}}\, Q_N=
 Q_N\, Q_{N{\text{\rm{r}}}}= Q_{N{\text{\rm{r}}}}.
 \end{equation}
 Since the columns of the matrix $\varphi(i\kappa_N,x)\, Q_N$ are square integrable in $x\in\mathbb R^+,$
 from \eqref{6.2} it follows that 
 the columns of the matrix $\varphi(i\kappa_N,x)\, Q_{N{\text{\rm{r}}}}$ are also square integrable.
 Analogous to \eqref{3.16} and \eqref{3.17}, we introduce
 the $n\times n$ matrices $\mathbf G_{N{\text{\rm{r}}}}$ and $ \mathbf H_{N{\text{\rm{r}}}}$ as
 \begin{equation}\label{6.3}
 \mathbf G_{N{\text{\rm{r}}}}:= \int_0^\infty dx\,Q_{N{\text{\rm{r}}}}\, \varphi(i\kappa_N,x)^\dagger \,\varphi(i\kappa_N,x) \,Q_{N{\text{\rm{r}}}},
 \end{equation}
 \begin{equation}\label{6.4}
 \mathbf H_{N{\text{\rm{r}}}}:= I-Q_{N{\text{\rm{r}}}}+\mathbf  G_{N{\text{\rm{r}}}}.
 \end{equation}
 One can prove that $\mathbf H_{N{\text{\rm{r}}}}$ is positive, and hence 
 it is invertible. Thus, we uniquely determine
$\mathbf H_{N{\text{\rm{r}}}}^{1/2}$ and its inverse
$\mathbf H_{N{\text{\rm{r}}}}^{-1/2}$ as the positive matrices satisfying the respective
matrix equalities
 \begin{equation*}
\mathbf H_{N{\text{\rm{r}}}}^{1/2}\,\mathbf H_{N{\text{\rm{r}}}}^{1/2}=\mathbf H_{N{\text{\rm{r}}}},\quad
\mathbf H_{N{\text{\rm{r}}}}^{-1/2}\,\mathbf H_{N{\text{\rm{r}}}}^{-1/2}=\mathbf H_{N{\text{\rm{r}}}}^{-1}.
 \end{equation*}
Analogous to \eqref{3.19}, we introduce the $n\times n$ matrix
$C_{N{\text{\rm{r}}}}$ as
\begin{equation}
\label{6.6}
C_{N{\text{\rm{r}}}}:= \mathbf H_{N{\text{\rm{r}}}}^{-1/2} \,Q_{N{\text{\rm{r}}}}.
\end{equation}
Inspired by \eqref{4.17}, \eqref{5.1}, and \eqref{5.2}, with the
help of $C_{N{\text{\rm{r}}}}$ and the unperturbed regular
solution $\varphi(k,x),$ we construct the quantity $G(x,y)$ as
\begin{equation}
\label{6.7}
G(x,y):=- \varphi(i\kappa_N,x) \,C_{N{\text{\rm{r}}}}^2\, \varphi(i\kappa_N,y)^\dagger.
\end{equation}
In analogy with \eqref{3.12}, we introduce the $n\times n$ matrix solution 
$\Phi_{N{\text{\rm{r}}}}(x)$ to the Schr\"odinger equation \eqref{3.3} with $j=N$ as
\begin{equation}\label{6.8} 
\Phi_{N{\text{\rm{r}}}}(x):= \varphi(i\kappa_N, x) \,C_{N{\text{\rm{r}}}}.
\end{equation}
Using \eqref{6.8} in \eqref{6.7}, we write $G(x,y)$ in the equivalent form as
\begin{equation}
\label{6.9}
G(x,y)= -\Phi_{N{\text{\rm{r}}}}(x)\, \Phi_{N{\text{\rm{r}}}}(y)^\dagger.
\end{equation}
The quantity $G(x,y)$ in \eqref{6.9} forms a separable kernel for the Gel'fand--Levitan system \eqref{4.16}.
Thus, we obtain the solution $\mathcal A(x,y)$ to \eqref{4.16} by using the methods of linear algebra and we get
\begin{equation}\label{6.10}
\mathcal A(x,y)= \Phi_{N{\text{\rm{r}}}}(x) \,W_{N{\text{\rm{r}}}}(x)^+ \,\Phi_{N{\text{\rm{r}}}}(y)^\dagger, \qquad  0 \le y <x,
\end{equation}
which is the analog of \eqref{5.4}. We remark that the $n\times n$ matrix-valued quantity $W_{N{\text{\rm{r}}}}(x)$
in \eqref{6.10} is defined as 
\begin{equation}
\label{6.11}
W_{N{\text{\rm{r}}}}(x):=\int_x^\infty dy\,\Phi_{N{\text{\rm{r}}}}(y)^\dagger\, \Phi_{N{\text{\rm{r}}}}(y),
\end{equation}
which is the analog of \eqref{5.5}.
We recall that $W_{N{\text{\rm{r}}}}(x)^+$ denotes the Moore--Penrose inverse of  $W_{N{\text{\rm{r}}}}(x).$
Using \eqref{6.10} in \eqref{4.15}, 
the perturbed potential $\tilde V(x)$ is obtained as
\begin{equation}
\label{6.12}
\tilde V(x)=V(x)+2\,\ds\frac{d}{dx}\left[
\Phi_{N{\text{\rm{r}}}}(x) \,W_{N{\text{\rm{r}}}}(x)^+\,\Phi_{N{\text{\rm{r}}}}(x)^\dagger
\right],
\end{equation} 
which is the analog of \eqref{5.8}.
Using \eqref{6.10} in \eqref{4.14}, we obtain
the perturbed regular solution $\tilde\varphi(k,x)$ 
as
 \begin{equation*}
\tilde\varphi(k,x)=\varphi(k,x)+ \Phi_{N{\text{\rm{r}}}}(x) \,W_{N{\text{\rm{r}}}}(x)^+ \int_0^x dy\,\Phi_{N{\text{\rm{r}}}}(y)^\dagger\,\varphi(k,y),
\end{equation*}
which is equivalently expressed as
\begin{equation}
\label{6.14}
\begin{split}
\tilde\varphi(k,x)=&\varphi(k,x)\\
&+\ds\frac{1}{k^2+\kappa_N^2}
\,\Phi_{N{\text{\rm{r}}}}(x) \,W_{N{\text{\rm{r}}}}(x)^+\left[\Phi'_{N{\text{\rm{r}}}}(x)^\dagger\,\varphi(k,x)-\Phi_{N{\text{\rm{r}}}}(x)^\dagger
\,\varphi'(k,x)\right].
\end{split}
\end{equation}
We note that \eqref{6.14} is the analog of \eqref{5.7}.
Using \eqref{6.10} in \eqref{4.18}, 
the perturbed boundary matrices $\tilde A$ and $\tilde B$ are
expressed in terms of the unperturbed
boundary matrices $A$ and $B$ and
the matrix $C_{N{\text{\rm{r}}}}$ as
\begin{equation}\label{6.15}
\tilde A=A, \quad \tilde B=B+A\,C_{N{\text{\rm{r}}}}^2 A^\dagger A,
\end{equation}
which is the analog of \eqref{5.9}.

In decreasing the multiplicity of the bound state at $k=i\kappa_N,$ it is relevant to estimate the asymptotic behavior of
the potential increment $\tilde V-V.$ For the details we refer the reader to
Section~7 of \cite{AW2025}.
We recall that the unperturbed potential $V$ satisfies \eqref{2.2} and \eqref{2.4},
the boundary matrices $A$ and $B$ appearing in \eqref{2.7} satisfy \eqref{2.8} and \eqref{2.9},
and the perturbed potential $\tilde V$ is given by \eqref{6.12}.
We have the following:

\begin{enumerate}

\item[\text{\rm(a)}] The perturbed potential $\tilde V$ also 
satisfies \eqref{2.2}.

\item[\text{\rm(b)}] 
The asymptotics for the potential increment $\tilde V-V$ given by
\begin{equation}
\label{6.16}
\tilde V(x)-V(x)= O\left(\int_x^\infty dy \,|V(y)| \right), \qquad x \to +\infty,
\end{equation}
which is the analog of \eqref{5.10}.
The asymptotic estimate in \eqref{6.16} is not necessarily sharp.

\item[\text{\rm(c)}] 
If the unperturbed potential
$V$ belongs  to $L^1_{1+\epsilon}(\mathbb R^+)$ for some fixed $\epsilon\ge 0,$ then the perturbed
potential $\tilde V$ belongs to
 $L^1_{\epsilon}(\mathbb R^+).$

\item[\text{\rm(d)}] We have the analogs of \eqref{5.11} and \eqref{5.12}.
In other words, if the unperturbed potential $V$ is further restricted to satisfy
\begin{equation}\label{6.17}
|V(x)| \le c\,e^{-\alpha x}, \qquad x \ge x_0,
\end{equation}
for some positive constants $\alpha$ and $x_0,$ then the potential increment $\tilde V-V$ satisfies
\begin{equation}\label{6.18}
|\tilde V(x)-V(x)| \le c\,e^{-\alpha x}, \qquad x \ge x_0.
\end{equation}
Here, $c$ denotes a generic constant not necessarily taking the same value in different appearances.
We remark that \eqref{6.17} and \eqref{6.18} are the analogs of \eqref{5.11} and \eqref{5.12}, respectively.
We note that the estimates in \eqref{6.16}--\eqref{6.18} are not necessarily sharp.

\item[\text{\rm(e)}]
If the support of $V$ is contained in the interval $[0,x_0],$ then the support of $\tilde V$ is also contained in $[0,x_0].$ 

\end{enumerate}

By further using the sufficiency assumption  that $V$ belongs to $L^1_{2}(\mathbb R^+),$  we determine  \cite{AW2025} 
the transformations of the remaining relevant quantities.
The unperturbed Jost matrix $J(k)$ is transformed into the perturbed Jost matrix $\tilde J(k)$ as
 \begin{equation}
\label{6.19}
\tilde J(k)=\left[I+\ds\frac{2i\kappa_N}{k-i\kappa_N}\,P_{N{\text{\rm{r}}}}\right] J(k),\qquad
k\in\overline{\mathbb C^+},\end{equation}
which is the analog of \eqref{5.13}.
The construction of the orthogonal projection matrix $P_{N{\text{\rm{r}}}}$ is accomplished as follows.
Using the already known orthogonal projection matrix $Q_{N{\text{\rm{r}}}},$ we determine an orthonormal basis
 $\{w^{(l)}_N\}_{l=1}^{m_{N{\text{\rm{r}}}}}$ for the subspace $Q_{N{\text{\rm{r}}}}\, \mathbb C^n.$
For each unit vector $w^{(l)}_N,$ we uniquely determine the vector $\beta^{(l)}_N$ satisfying
\begin{equation}\label{6.20}
\varphi(i\kappa_N,x)\,w^{(l)}_N=f(i\kappa_N,x)\,\beta^{(l)}_N,\qquad 1\le l\le m_{N{\text{\rm{r}}}}.  
\end{equation}
We then obtain $P_{N{\text{\rm{r}}}}$ as the orthogonal projection onto the subspace generated by the set
 $\{\beta^{(l)}_N\}_{l=1}^{m_{N{\text{\rm{r}}}}}.$
We remark that the set
 $\{\beta^{(l)}_N\}_{l=1}^{m_{N{\text{\rm{r}}}}}$
 is in general not an orthogonal set and the vectors
in that set are in general not unit vectors. It is ensured \cite{AW2025}
that the $m_{N{\text{\rm{r}}}}$ vectors in the set are linearly independent.
For the details,
we refer the reader to \cite{AW2025}.

The determinant $\det[J(k)]$ of the unperturbed Jost matrix is transformed as
 \begin{equation*}
\det[\tilde J(k)]=\left(\ds\frac{k+i\kappa_N}{k-i\kappa_N}\right)^{m_{N{\text{\rm{r}}}}} \det[J(k)],\qquad
k\in\overline{\mathbb C^+},\end{equation*}
which is the analog of \eqref{5.14}.
We recall
that $m_{N{\text{\rm{r}}}}$ corresponds to the rank of the orthogonal projection  $Q_{N{\text{\rm{r}}}}.$
The scattering matrix $S(k)$ undergoes the transformation
 \begin{equation*}
\tilde S(k)=\left[I-\ds\frac{2i\kappa_N}{k+i\kappa_N}\,P_{N{\text{\rm{r}}}}\right] S(k)\left[I-\ds\frac{2i\kappa_N}{k+i\kappa_N}\,P_{N{\text{\rm{r}}}}\right],
\qquad k\in\mathbb R,\end{equation*}
which is the analog of \eqref{5.15}.
The determinant $\det[S(k)]$ of the scattering matrix $S(k)$ is transformed as 
\begin{equation*}
\det[\tilde S(k)]=\left(\ds\frac{k-i\kappa_N}{k+i\kappa_N}\right)^{2\, m_{N{\text{\rm{r}}}}} \det[S(k)],
\qquad k\in\mathbb R,\end{equation*}
which is the analog of \eqref{5.16}.
The unperturbed Jost solution $f(k,x)$ is transformed into the perturbed Jost solution $\tilde f(k,x)$ as
\begin{equation}\label{6.24}
\tilde f(k,x)=\left[f(k,x)+\ds\frac{1}{k^2+\kappa_N^2}
\,\Phi_{N{\text{\rm{r}}}}(x) \,W_{N{\text{\rm{r}}}}(x)^+ \,q_{2}(x) \right]  \left[ I+ \ds\frac{2i\kappa_N}{k-i\kappa_N} \,P_{N{\text{\rm{r}}}} \right],
\end{equation}
where we have defined
\begin{equation*}
q_{2}(x):=\Phi'_{N{\text{\rm{r}}}}(x)^\dagger\,f(k,x)-\Phi_{N{\text{\rm{r}}}}(x)^\dagger
\,f'(k,x),
\end{equation*}
which are the analogs of \eqref{5.18} and
\eqref{5.19}, respectively.
We can express \eqref{6.24} equivalently as
\begin{equation*}
\tilde f(k,x)=\left[ f(k,x)-
\Phi_{N{\text{\rm{r}}}}(x) \,W_{N{\text{\rm{r}}}}(x)^+\int_x^\infty dy\,
\Phi_{N{\text{\rm{r}}}}(y)^\dagger \,f(k,y)  \right] \left[ I+ \frac{2i\kappa_N}{k-i\kappa_N} P_{N{\text{\rm{r}}}} \right],
\end{equation*}
which is the analog of \eqref{5.17}.
Under the perturbation specified in \eqref{6.2}, the orthogonal projection matrices $Q_j$ and the Gel'fand--Levitan normalization
matrices $C_j$ for the remaining bound states do not change, i.e. we have
 \begin{equation*}
\tilde Q_j=Q_j,\quad \tilde C_j=C_j,\qquad 1\le j\le N-1.\end{equation*}

 \section{The addition of a bound state}
\label{section7}

In this section we assume that the potential $V$ satisfies \eqref{2.2} and at least \eqref{2.4}.
We provide a description of the changes in the relevant quantities related to
the Schr\"odinger operator associated with \eqref{2.1} and \eqref{2.7}, when we perturb the 
spectrum of our operator by adding a new bound state without changing the
rest of the discrete spectrum and without changing the continuous spectrum. We assume that the unperturbed Schr\"odinger operator has $N$ bound states
occurring at $k=i\kappa_j$ with the multiplicity $m_j$ and with the Gel'fand--Levitan normalization matrix $C_j$ for $1\le j\le N.$ 
We recall that $N$ may be equal to zero, in which case there are no bound states.
We add a 
new bound state at $k=i\tilde\kappa_{N+1}$ with the multiplicity $\tilde m_{N+1}$ and the
Gel'fand--Levitan normalization matrix $\tilde C_{N+1},$ where the positive constant $\tilde\kappa_{N+1}$ is distinct from
$\kappa_j$ for $1\le j\le N.$
We remark that we cannot choose $\tilde C_{N+1}$ in an arbitrary manner. For example, it is necessary that $\tilde C_{N+1}$ is an
$n\times n$ selfadjoint nonnegative matrix. The choice of $\tilde C_{N+1}$
must be compatible with the construction process
given in \eqref{3.16}--\eqref{3.19} starting with the perturbed regular solution 
$\tilde\varphi(k,x)$ and the orthogonal projection matrix
$\tilde Q_{N+1}$ onto the kernel of the perturbed Jost matrix
$\tilde J(i\tilde\kappa_{N+1}).$ We illustrate this compatibility issue
in Example~\ref{example10.4}, where we verify that
the choice of $\tilde C_{N+1}$ used in that example indeed agrees
with the construction described in \eqref{3.16}--\eqref{3.19}.

We construct $\tilde C_{N+1}$ as follows. We first choose an $n\times n$ orthogonal projection matrix $\tilde Q_{N+1}$ having the rank
$\tilde m_{N+1}.$ Next, we construct an $n\times n$ nonnegative hermitian matrix $\tilde{\mathbf G}_{N+1}$ satisfying the equalities
\begin{equation}\label{7.1}
\tilde{\mathbf G}_{N+1} \,\tilde Q_{N+1}=\tilde Q_{N+1}\, \tilde{\mathbf G}_{N+1} = \tilde{\mathbf G}_{N+1},
\end{equation}
in such a way that the  restriction of $\tilde{\mathbf G}_{N+1}$ to $\tilde Q_{N+1}\, \mathbb C^n$  is invertible. 
We then define
the matrix $\tilde{\mathbf H}_{N+1}$ as
\begin{equation}\label{7.2}
\tilde{\mathbf H}_{N+1}:= I- \tilde Q_{N+1}+\tilde{\mathbf G}_{N+1}.
\end{equation}
From \eqref{7.2} we observe that $\tilde{\mathbf H}_{N+1}$ is hermitian and nonnegative.
It turns out \cite{AW2025} that $\tilde{\mathbf H}_{N+1}$ is positive and hence invertible. 
As in \eqref{3.18}, the positivity of
$\tilde{\mathbf H}_{N+1}$
 allows the unique construction of the positive matrix 
$\tilde{\mathbf H}_{N+1}^{1/2}$ and in turn the
construction of 
$\tilde{\mathbf H}_{N+1}^{-1/2}.$
We then use the analog of \eqref{3.19} and obtain
$\tilde C_{N+1}$ as
\begin{equation}\label{7.3} 
\tilde C_{N+1}:=
\tilde{\mathbf H}_{N+1}^{-1/2}\,\tilde Q_{N+1}.
\end{equation}
We determine \cite{AW2025} that each of $\tilde{\mathbf H}_{N+1},$
$\tilde{\mathbf H}_{N+1}^{1/2},$ and 
$\tilde{\mathbf H}_{N+1}^{-1/2}$ commutes with 
$\tilde Q_{N+1}.$
Consequently, from \eqref{7.3} we get
\begin{equation*}
 \tilde C_{N+1}\,\tilde Q_{N+1}=\tilde Q_{N+1}\,\tilde C_{N+1}=
\tilde C_{N+1}.
\end{equation*}

It is convenient to use the Gel'fand--Levitan method to determine the resulting perturbation of the relevant quantities.
We use the Gel'fand--Levitan system \eqref{4.16} and its solution $\mathcal A(x,y)$ to construct
those relevant quantities. Toward our goal, we choose the change in the spectral measure as
\begin{equation}
\label{7.5}
d\tilde\rho-d\rho=\tilde C_{N+1}^2 \,\delta(\lambda-\tilde\lambda_{N+1}),\qquad \lambda\in\mathbb R,\end{equation}
where we have let $\tilde\lambda_{N+1}:=-\tilde\kappa_{N+1}^2.$
Note that \eqref{7.5} is compatible with \eqref{4.1} and \eqref{4.12} 
by assuming that the Gel'fand--Levitan normalization matrices $C_j$ for the existing bound states at $k=i\kappa_j$ with $1\le j\le N$ do not change. It turns out that
if we use \eqref{7.5} in the construction of the input matrix $G(x,y)$ appearing in \eqref{4.17}, solve the resulting 
Gel'fand--Levitan system \eqref{4.16}, and use
the corresponding solution $\mathcal A(x,y)$ to construct the perturbed
quantities, then for $1\le j\le N$ each perturbed Gel'fand--Levitan normalization matrix
$\tilde C_j$ coincides with the unperturbed Gel'fand--Levitan normalization matrix
$C_j,$  and the continuous part of the unperturbed spectral measure $d\rho$ remains unchanged.

Using \eqref{7.5} as input to \eqref{4.17} we get
\begin{equation}\label{7.6}
G(x,y)= \varphi(i\tilde\kappa_{N+1},x) \,\tilde C_{N+1}^2\, \varphi(i \tilde\kappa_{N+1},y)^\dagger,
\end{equation}
where we recall that $\varphi(k,x)$ is the unperturbed regular solution to \eqref{2.1} satisfying the
initial conditions \eqref{2.6}.
With the help of $\varphi(k,x)$ and $\tilde C_{N+1},$
we define the $n\times n$ matrix $\xi_{N+1}$ as
\begin{equation}
\label{7.7}
\xi_{N+1}(x):=\varphi(i\tilde\kappa_{N+1},x)\,\tilde C_{N+1}.\end{equation}
Using \eqref{7.7} in \eqref{7.6}, we express $G(x,y)$ in the equivalent form as
\begin{equation}\label{7.8}
G(x,y)= \xi_{N+1}(x) \, \xi_{N+1}(y)^\dagger.
\end{equation}
We note that $G(x,y)$ forms a separable kernel for the Gel'fand--Levitan system \eqref{4.16}, and hence
the solution $\mathcal A(x,y)$ to \eqref{4.16} can be obtained by using the methods from linear algebra. We get
$\mathcal A(x,y)$ explicitly as
\begin{equation}
\label{7.9}
\mathcal A(x,y)=-\xi_{N+1}(x)\,\Omega_{N+1}(x)^+\,\xi_{N+1}(y)^\dagger,\qquad 0\le y < x,\end{equation}
with the $n\times n$ matrix $\Omega_{N+1}(x)$ defined as
\begin{equation}
\label{7.10}
\Omega_{N+1}(x):=
\tilde Q_{N+1}+\int_0^x dy\,\xi_{N+1}(y)^\dagger\,
\xi_{N+1}(y),
\end{equation}
where we recall that $\Omega_{N+1}(x)^+$ denotes the Moore--Penrose inverse of $\Omega_{N+1}(x).$
Using \eqref{7.9} in \eqref{4.15} we obtain the perturbed potential $\tilde V(x)$ as
\begin{equation}
\label{7.11}
\tilde V(x)=V(x)-2\,\ds\frac{d}{dx}\left[ \xi_{N+1}(x)\,\Omega_{N+1}(x)^+\,\xi_{N+1}(x)^\dagger\right].
\end{equation}
Using \eqref{7.9} in \eqref{4.14} we get the perturbed regular solution $\tilde\varphi(k,x)$ as
 \begin{equation}\label{7.12}
\tilde\varphi(k,x)=\varphi(k,x)- \xi_{N+1}(x)\,\Omega_{N+1}(x)^+ \int_0^x dy\,\xi_{N+1}(y)^\dagger\, \varphi(k,y),
\end{equation}
which can also be written in the equivalent form as
\begin{equation}
\label{7.13}
\tilde\varphi(k,x)=\varphi(k,x)-\ds\frac{1}{k^2+\tilde\kappa_{N+1}^2}
\,\xi_{N+1}(x) \,\Omega_{N+1}(x)^+\left[\xi'_{N+1}(x)^\dagger\,\varphi(k,x)-\xi_{N+1}(x)^\dagger
\,\varphi'(k,x)\right].
\end{equation}
Next, using \eqref{7.9} in \eqref{4.18}, we express the perturbed boundary matrices
$\tilde A$ and $\tilde B$
in terms of the unperturbed boundary matrices $A$ and $B$ as
\begin{equation}\label{7.14}
\tilde A=A, \quad \tilde B=B-A\,\tilde C_{N+1}^2 A^\dagger A.
\end{equation}

In adding a new bound state at $k=i\tilde\kappa_{N+1}$ to the spectrum, we are interested in estimating the asymptotic behavior of
the potential increment $\tilde V-V.$ For the details we refer the reader to
Theorem~8.5 of \cite{AW2025}.
We recall that the unperturbed potential $V$ satisfies \eqref{2.2} and \eqref{2.4},
the boundary matrices $A$ and $B$ appearing in \eqref{2.7} satisfy \eqref{2.8} and \eqref{2.9},
and the perturbed potential $\tilde V$ is given by \eqref{7.11}.
We have the following:

\begin{enumerate}

\item[\text{\rm(a)}] The perturbed potential $\tilde V$ also
satisfies \eqref{2.2}. 

\item[\text{\rm(b)}] 
For every nonnegative constant $a$ and
every parameter $\varepsilon$ satisfying $0<\varepsilon<1,$ the potential increment $\tilde V-V$ has the
asymptotic behavior
\begin{equation}
\label{7.15}
\tilde V(x)- V(x)=O\left(q_{3}(a,\varepsilon,x)\right),\qquad x\to+\infty,
\end{equation}
with $q_{3}(a,\varepsilon,x)$ defined as
\begin{equation}
\label{7.16}
\begin{split}
q_{3}(a,\varepsilon,x):=  &x \, e^{-2\tilde\kappa_{N+1}x}+ \int_x^\infty dy\,|V(y)|
+  e^{-2 \varepsilon \tilde\kappa_{N+1} x}  \int_a^{(1-\varepsilon)x} dy\, |V(y)|
\\
&+\int_{(1-\varepsilon) x}^x dy \, e^{-2 \tilde\kappa_{N+1} (x-y)}\, |V(y)|.
\end{split}
\end{equation}

\item[\text{\rm(c)}] 
If the unperturbed potential $V$ is further restricted to  $L^1_{1+\epsilon}(\mathbb R^+)$ for some fixed $\epsilon \ge 0,$ then 
the perturbed potential $\tilde V$ belongs to $L^1_\epsilon(\mathbb R^+).$

\item[\text{\rm(d)}] If the unperturbed potential $V$ is further restricted to satisfy
\begin{equation}
\label{7.17}
|V(x)| \le c\,e^{-\alpha x}, \qquad x \ge x_0,
\end{equation}
for some positive constants $\alpha$ and $x_0$ and with $c$
denoting a generic constant, then for every $0<\varepsilon<1$  the potential increment
$\tilde V-V$ has the asymptotic behavior as $x\to+\infty$ given by
\begin{equation}\label{7.18}
\tilde V(x) -V(x) =\begin{cases}
O\left(e^{-\alpha x}+ e^{-2 \varepsilon\tilde\kappa_{N+1}x}\right),\qquad  \alpha\le 2 \tilde\kappa_{N+1}, \\
\noalign{\medskip}
O\left( e^{-2\varepsilon\tilde\kappa_{N+1} x}\right), \qquad \alpha > 2  \tilde\kappa_{N+1}.
\end{cases}
\end{equation}

\item[\text{\rm(e)}]
If the unperturbed potential $V$ has compact support, then the perturbed potential $\tilde V$ has the asymptotic behavior
\begin{equation}
\label{7.19}
\tilde V(x)= O\left(x\, e^{-2\tilde\kappa_{N+1} x}\right), \qquad x \to +\infty.
\end{equation}

\end{enumerate}

We remark that the estimates given in \eqref{7.15}--\eqref{7.19} are not necessarily sharp.

We refer the reader to \cite{AW2025} for the construction of the remaining relevant quantities
under the further sufficiency assumption that
the unperturbed potential $V$ belongs to $L^1_2(\mathbb R^+).$  
We provide the results from \cite{AW2025} as follows.
The unperturbed Jost matrix $J(k)$ is transformed into the perturbed Jost matrix $\tilde J(k)$ as
 \begin{equation}
\label{7.20}
\tilde J(k)=\left[I-\ds\frac{2i\tilde\kappa_{N+1}}{k+i\tilde\kappa_{N+1}}\,\tilde P_{N+1}\right] J(k),\qquad
k\in\overline{\mathbb C^+}.
\end{equation}
The orthogonal projection onto the kernel of $\tilde J(i\tilde\kappa_{N+1})$ coincides \cite{AW2025} with the orthogonal projection matrix
$\tilde Q_{N+1}$ 
appearing in \eqref{7.1}--\eqref{7.3} and used in the definition of  the Gel'fand--Levitan normalization matrix $\tilde C_{N+1}.$
We note that the $n\times n$ matrix $\tilde P_{N+1}$ appearing in \eqref{7.20} is the orthogonal projection onto the
 kernel of $\tilde J(i\tilde\kappa_{N+1})^\dagger$ and it has rank $\tilde m_{N+1},$ by recalling that
that $\tilde m_{N+1}$ is the rank of  $\tilde C_{N+1}.$
We have
\begin{equation}\label{7.21}
\tilde P_{N+1}:= L_{N+1}\, \tilde C_{N+1}\, [ \tilde C_{N+1}\, L_{N+1}^\dagger\, L_{N+1}\, \tilde C_{N+1}]^+ \tilde C_{N+1}\, L_{N+1}^\dagger.
\end{equation}
The constant $n\times n$ matrix $L_{N+1}$ appearing in \eqref{7.21} is invertible, and it is obtained 
by using the corresponding coefficient matrix in the expansion
\begin{equation}\label{7.22}
\varphi(i\tilde\kappa_{N+1},x)= f(i\tilde\kappa_{N+1},x)\,K_{N+1}+  g(i\tilde\kappa_{N+1},x)\,L_{N+1},
\end{equation}
with $f(k,x)$ being the Jost solution to \eqref{2.1} appearing in
\eqref{2.10} and $g(k,x)$ being the $n\times n$ matrix-valued solution to \eqref{2.1} constructed 
for $k\in\bCp\setminus\{0\}$ and satisfying
the spacial asymptotics 
\begin{equation}
\label{7.23}
g(k,x)=e^{-ikx}\left[I+o(1)\right],\quad g'(k,x)=-e^{-ikx}\left[ik\,I+o(1)\right],
\qquad x\to+\infty.
\end{equation}
For more information on $g(k,x),$ we refer the reader to \cite{AW2021} and \cite{AW2025}.

Under the perturbation specified in \eqref{7.5}, the determinant of the Jost matrix $J(k)$ is transformed as
 \begin{equation}
\label{7.24}
\det[\tilde J(k)]=\left(\ds\frac{k-i\tilde\kappa_{N+1}}{k+i\tilde\kappa_{N+1}}\right)^{\tilde m_{N+1}} \det[J(k)],\qquad
k\in\overline{\mathbb C^+},\end{equation}
from which we observe that the total
number of bound states including the multiplicities is increased by $\tilde m_{N+1}.$
The scattering matrix $S(k)$ is transformed into $\tilde S(k)$ as 
 \begin{equation}
\label{7.25}
\tilde S(k)=\left[I+\ds\frac{2i\tilde\kappa_{N+1}}{k-i\tilde\kappa_{N+1}}\,\tilde P_{N+1}\right] S(k)\left[I+\ds\frac{2i\tilde\kappa_N}{k-i\tilde\kappa_{N+1}}\,\tilde P_{N+1}\right],
\qquad k\in\mathbb R,
\end{equation}
where we recall that $\tilde P_{N+1}$ is the matrix defined in \eqref{7.21}.
The determinant of the scattering matrix is transformed as 
 \begin{equation}
\label{7.26}
\det[\tilde S(k)]=\left(\ds\frac{k +i\tilde\kappa_{N+1}}{k-i\tilde\kappa_{N+1}}\right)^{2 \tilde m_{N+1}} \det[S(k)],
\qquad k\in\mathbb R.
\end{equation}
The Jost solution $f(k,x)$ is transformed into $\tilde f(k,x)$ as
\begin{equation}
\label{7.27}
\tilde f(k,x)=\left[f(k,x)-\ds\frac{1}{k^2+\tilde\kappa_{N+1}^2}
\,\xi_{N+1}(x) \,\Omega_{N+1}(x)^+
q_{4}(k,x)
\right]
\left[I- \frac{2i \tilde\kappa_{N+1}}{k+i \tilde\kappa_{N+1} } \tilde P_{N+1}\right],
\end{equation}
where we have defined $q_{4}(k,x)$ as
\begin{equation*}
q_{4}(k,x):=\left[\xi'_{N+1}(x)^\dagger\,f(k,x)-\xi_{N+1}(x)^\dagger
\,f'(k,x)\right].\end{equation*}
Furthermore, under the perturbation specified in \eqref{7.5}, the existing orthogonal
projections $Q_j$ and the existing Gel'fand--Levitan normalization matrices $C_j$
remain unchanged, i.e. we have
 \begin{equation*}
\tilde Q_j=Q_j,\quad \tilde C_j=C_j,\qquad 1\le j\le N.
\end{equation*}

\section{Increasing the multiplicity of a bound state}
\label{section8}

In this section we assume that the potential $V$ satisfies \eqref{2.2} and at least \eqref{2.4}.
We describe the changes in the relevant quantities related to the Schr\"odinger operator
associated with \eqref{2.1} and \eqref{2.7}, when we perturb the spectrum of our operator by increasing the multiplicity
of a bound state. We do not change the remaining part of the discrete spectrum
and we do not change the continuous spectrum.
We assume that the unperturbed operator has $N$ bound states at $k=i\kappa_j$
with the multiplicity $m_j$ and with the Gel'fand--Levitan normalization matrix $C_j$ for $1\le j\le N.$
Without loss of generality, 
we increase the multiplicity of the bound state at $k=i\kappa_N$ from $m_N$ to $\tilde m_N$ by
the positive integer $m_{N\text{\rm{i}}},$ which is defined as
\begin{equation}
\label{8.1}
m_{N\text{\rm{i}}}:=\tilde m_N-m_N.
\end{equation}
By using the subscript $N{\text{\rm{i}}}$ in \eqref{8.1}, we indicate that we refer to the $N$th bound state at $k=i\kappa_N$ and that we 
increase its multiplicity.
Since the multiplicity of a bound state cannot exceed $n,$ 
it is implicitly assumed that we have $1\le m_N\le n-1.$

Changing the multiplicity of a bound state does not apply in the scalar case, i.e. when $n=1,$ because 
in the scalar case all bound states have the multiplicity equal to one. Hence, we assume that
we are in the matrix case with $n\ge 2$ and that there exists at least one bound state, i.e. we have $N\ge 1.$
As in the previous sections, we use a tilde to identify the perturbed
quantities obtained after the multiplicity of the bound state is increased. We assume that our
unperturbed potential $V$
satisfies \eqref{2.2} and \eqref{2.4}.
In our notation, our unperturbed regular solution is
$\varphi(k,x),$ unperturbed Jost solution is $f(k,x),$ 
unperturbed Jost matrix is $J(k),$ unperturbed scattering matrix is $S(k),$ unperturbed 
matrix $Q_j$ corresponds to the
orthogonal projection onto the kernel of
$J(i\kappa_j)$ for $1\le j\le N,$
and 
unperturbed
boundary matrices are given by $A$ and $B.$
Our perturbed potential is $\tilde V(x),$ perturbed regular solution is
$\tilde\varphi(k,x),$ perturbed Jost solution is $\tilde f(k,x),$ 
perturbed Jost matrix is $\tilde J(k),$ perturbed scattering matrix is $\tilde S(k),$ 
and perturbed
boundary matrices are given by $\tilde A$ and $\tilde B.$

We recall that $Q_N$ denotes the orthogonal projection onto
the kernel of $J(i\kappa_N),$ the positive integer $m_N$ corresponds to the dimension of
$\text{\rm{Ker}}[J(i\kappa_N)],$ and we use $C_N$ to denote the Gel'fand--Levitan normalization matrix
associated with the bound state at $k=i\kappa_N$ in the unperturbed case.
The increase of the multiplicity of the bound state at $k=i\kappa_N$
by $m_{N{\text{\rm{i}}}}$ is accomplished as follows.
From \eqref{8.1} we see that $1\le m_{N{\text{\rm{i}}}}\le n-m_N.$
When the multiplicity of the bound state at $k=i\kappa_N$ is increased from $m_N$ to $\tilde m_N,$ the transformations of the relevant quantities are obtained as follows. For the details we refer the reader to Section~9 of
\cite{AW2025}.
We introduce the $n\times n$ matrix $\tilde Q_{N{\text{\rm{i}}}}$ with the rank
 $m_{N{\text{\rm{i}}}}$ in order to denote
the orthogonal projection onto a subspace of 
the orthogonal complement in $\mathbb C^n$ of
$Q_N\, \mathbb C^n.$ 
Hence, we have
\begin{equation}
\label{8.2}
\tilde Q_{N\text{\rm{i}}}\, \mathbb C^n \subset \left(Q_N \,\mathbb C^n\right)^\perp.
\end{equation}
We use an $n\times n$ nonnegative hermitian matrix $\tilde{\mathbf G}_{N\text{\rm{i}}}$ satisfying the equalities
\begin{equation}
\label{8.3} 
\tilde{\mathbf G}_{N\text{\rm{i}}}\,\tilde Q_{N\text{\rm{i}}}=\tilde Q_{N\text{\rm{i}}}\,\tilde{\mathbf G}_{N\text{\rm{i}}} = \tilde{\mathbf G}_{N\text{\rm{i}}}.
\end{equation}
We further assume that the restriction of  $\tilde{\mathbf G}_{N\text{\rm{i}}}$ to the subspace $\tilde Q_{N\text{\rm{i}}}\, \mathbb C^n$ is invertible.
With the help of the matrices $\tilde Q_{N\text{\rm{i}}}$ and
$\tilde{\mathbf G}_{N\text{\rm{i}}},$ we define the $n\times n$ matrix $\tilde{\mathbf H}_{N\text{\rm{i}}}$ as
\begin{equation}
\label{8.4}
\tilde{\mathbf H}_{N\text{\rm{i}}}:= I- \tilde Q_{N\text{\rm{i}}}+\tilde{\mathbf G}_{N\text{\rm{i}}},
\end{equation}
which is analogous to \eqref{3.17}. 
Since $\tilde Q_{N\text{\rm{i}}}$ and
$\tilde{\mathbf G}_{N\text{\rm{i}}}$ are hermitian, from \eqref{8.4} it follows that $\tilde{\mathbf H}_{N\text{\rm{i}}}$ is also hermitian.
With the help
of \eqref{8.3} and \eqref{8.4} we observe that
\begin{equation}
\label{8.5}
\tilde{\mathbf H}_{N\text{\rm{i}}}\, \tilde Q_{N\text{\rm{i}}}= \tilde Q_{N\text{\rm{i}}}\,\tilde{\mathbf H}_{N\text{\rm{i}}}.
\end{equation}
It is known \cite{AW2025} that $\tilde{\mathbf H}_{N\text{\rm{i}}}$ is a positive matrix, and hence it
its inverse
$\tilde{\mathbf H}_{N\text{\rm{i}}}^{-1}$ is uniquely determined.
Thus, the positive matrix $\tilde{\mathbf H}_{N\text{\rm{i}}}^{1/2}$ and its
inverse $\tilde{\mathbf H}_{N\text{\rm{i}}}^{-1/2}$ are also uniquely determined, 
and they satisfy
\begin{equation}
\label{8.6}
\tilde{\mathbf H}_{N\text{\rm{i}}}^{1/2}\, \tilde{\mathbf H}_{N\text{\rm{i}}}^{1/2}=
\tilde{\mathbf H}_{N\text{\rm{i}}}, \quad
\tilde{\mathbf H}_{N\text{\rm{i}}}^{-1/2}\, \tilde{\mathbf H}_{N\text{\rm{i}}}^{-1/2}=
\tilde{\mathbf H}_{N\text{\rm{i}}}^{-1}.
\end{equation}
Analogous to \eqref{3.19}, we define the $n\times n$ matrix $\tilde C_{N\text{\rm{i}}}$ as
\begin{equation}\label{8.7}
\tilde C_{N\text{\rm{i}}}:=\tilde{\mathbf H}_{N\text{\rm{i}}}^{-1/2}\, \tilde Q_{N\text{\rm{i}}}.
\end{equation}
It follows that $\tilde C_{N\text{\rm{i}}}$ is hermitian and nonnegative, its rank is equal to $m_{N\text{\rm{i}}},$ and
it satisfies the equalities
\begin{equation}
\label{8.8}
\tilde C_{N\text{\rm{i}}}\,\tilde Q_{N\text{\rm{i}}}= \tilde Q_{N\text{\rm{i}}} \,\tilde C_{N\text{\rm{i}}}=  \tilde C_{N\text{\rm{i}}}.
\end{equation}

With the help of the unperturbed regular solution $\varphi(k,x)$ and the matrix 
$\tilde C_{N\text{\rm{i}}},$ we introduce the $n\times n$ matrix solution 
$\xi_{N\text{\rm{i}}}(x)$ to the Schr\"odinger equation \eqref{3.3} with $j=N$ as
\begin{equation}
\label{8.9}
\xi_{N\text{\rm{i}}}(x):=\varphi(i\kappa_N,x)\,\tilde C_{N\text{\rm{i}}},
\end{equation}
which is the analog of \eqref{7.7}.
Using \eqref{8.9} we construct the kernel $G(x,y)$ to the Gel'fand--Levitan system
\eqref{4.16} as
\begin{equation}
\label{8.10}
G(x,y)=\xi_{N\text{\rm{i}}}(x)\,\xi_{N\text{\rm{i}}}(y)^\dagger,
\end{equation}
which is the analog of \eqref{7.8}.
The quantity $G(x,y)$ in \eqref{8.10} forms a separable kernel for the
Gel'fand--Levitan system \eqref{4.16}. By using the methods of linear algebra, we obtain the solution
$\mathcal A(x,y)$ to \eqref{4.16} as
\begin{equation}
\label{8.11}
\mathcal A(x,y)=-\xi_{N\text{\rm{i}}}(x)\,\Omega_{N\text{\rm{i}}}(x)^+\,\xi_{N\text{\rm{i}}}(y)^\dagger, \qquad 0 \le y < x,
 \end{equation}
where the $n\times n$ matrix $\Omega_{N\text{\rm{i}}}(x)$ is defined as
\begin{equation}
\label{8.12}
\Omega_{N\text{\rm{i}}}(x):=
\tilde Q_{N\text{\rm{i}}}+\int_0^x dy\,\xi_{N\text{\rm{i}}}(y)^\dagger\,
\xi_{N\text{\rm{i}}}(y),
\end{equation}
and we recall that $\Omega_{N\text{\rm{i}}}(x)^+$ denotes the Moore--Penrose inverse of $\Omega_{N\text{\rm{i}}}(x).$
We note that \eqref{8.11} is the analog of \eqref{7.9}.
Using \eqref{8.11} in \eqref{4.15}, 
we obtain the perturbed potential $\tilde V(x)$ as
\begin{equation}\label{8.13}
\tilde V(x):=V(x)-2\,\ds\frac{d}{dx}\left[ \xi_{N\text{\rm{i}}}(x)\,\Omega_{N\text{\rm{i}}}(x)^+\,\xi_{N\text{\rm{i}}}(x)^\dagger\right],\end{equation}
which is the analog of \eqref{7.11}.
Using \eqref{8.11} in \eqref{4.14}, we obtain the perturbed regular solution 
$\tilde\varphi(k,x)$ in terms of the unperturbed regular solution
$\varphi(k,x)$ as
  \begin{equation*}
\tilde\varphi(k,x)=\varphi(k,x)- \xi_{N\text{\rm{i}}}(x)\,\Omega_{N\text{\rm{i}}}(x)^+\, \int_0^x dy\,\xi_{N\text{\rm{i}}}(y)^\dagger\, \varphi(k,y),
\end{equation*}
 which is the analog of \eqref{7.12}.
 In analogy with \eqref{7.13}, we can express the perturbed regular solution also as
 \begin{equation}
\label{8.15}
\tilde\varphi(k,x)=\varphi(k,x)-\ds\frac{1}{k^2+\kappa_N^2}
\,\xi_{N\text{\rm{i}}}(x) \,\Omega_{N\text{\rm{i}}}(x)^+\left[\xi'_{N\text{\rm{i}}}(x)^\dagger\,\varphi(k,x)-\xi_{N\text{\rm{i}}}(x)^\dagger
\,\varphi'(k,x)\right],
\end{equation}
 where the singularities caused by the denominator on the right-hand side of \eqref{8.15} are both removable, and hence  the value of the right-hand side of \eqref{8.15} is understood to hold
 in the limiting sense at $k=\pm i\kappa_N.$
The perturbed quantity $\tilde\varphi(k,x)$ satisfies the initial conditions \eqref{4.13},
where the matrices $\tilde A$ and $\tilde B$
are expressed in terms of the unperturbed
boundary matrices $A$ and $B$ and the nonnegative matrix $\tilde C_{N\text{\rm{i}}}$ appearing in \eqref{8.7} as
\begin{equation}\label{8.16}
\tilde A=A, \quad \tilde B=B-A\,\tilde C_{N\text{\rm{i}}}^2 \,A^\dagger A,
\end{equation}
which is the analog of \eqref{7.14}.
The matrices $\tilde A$ and $\tilde B$ appearing in \eqref{8.16} indeed satisfy \eqref{4.6} and \eqref{4.7},
which are
the counterparts of \eqref{2.8} and \eqref{2.9}, respectively, satisfied
by the unperturbed boundary matrices
 $A$ and $B.$

 Analogous to \eqref{7.22}, we express the matrix solution $\varphi(i\kappa_N, x) \,\tilde Q_{N\text{\rm{i}}}$ to the Schr\"odinger  equation 
\eqref{2.1} with $k= i\kappa_N$ as
\begin{equation}\label{8.17}
\varphi(i\kappa_N,x) \,\tilde Q_{N\text{\rm{i}}}= f(i\kappa_N,x)\,K_{N\text{\rm{i}}} \,
\tilde Q_{N\text{\rm{i}}} +  g(i\kappa_N,x)\, L_{N\text{\rm{i}}}\,\tilde Q_{N\text{\rm{i}}},
\end{equation}
for some $n \times n$ matrices $K_{N\text{\rm{i}}}$ and $L_{N\text{\rm{i}}},$ 
where we recall that $g(k,x)$ is the matrix-valued solution to
\eqref{2.1} appearing in \eqref{7.23}. As shown in \cite{AW2025}, the restriction of  
$L_{N\text{\rm{i}}}$
to the subspace $\tilde Q_{N\text{\rm{i}}} \,\mathbb C^n$ is invertible.

In increasing the multiplicity of the bound state at $k=i\kappa_N,$ it is relevant to estimate the asymptotic behavior of
the potential increment $\tilde V-V.$ For the details we refer the reader to
Section~9 of \cite{AW2025}.
We recall that the unperturbed potential $V$ satisfies \eqref{2.2} and \eqref{2.4},
the boundary matrices $A$ and $B$ appearing in \eqref{2.7} satisfy \eqref{2.8} and \eqref{2.9},
and the perturbed potential $\tilde V$ is given by \eqref{8.13}.
We have the following:

\begin{enumerate}

\item[\text{\rm(a)}] The perturbed potential $\tilde V$ also
satisfies \eqref{2.2}.

\item[\text{\rm(b)}] 
For any $a$ and $\varepsilon$ with
$a\ge 0$ and $0 < \varepsilon <1,$ the corresponding potential increment $\tilde V-V$ has the asymptotic behavior
\begin{equation}
\label{8.18}
\tilde V(x)- V(x)=O\left(q_{5}(x,a,\varepsilon)\right), \qquad x\to+\infty,
\end{equation}
where we have defined
\begin{equation}
\label{8.19}
\begin{split}
q_{5}(x,a,\varepsilon):=&
 x \, e^{-2\kappa_N x} +  \int_x^\infty dy \,|V(y)|
+  
e^{-2 \varepsilon \kappa_N  x} \int_a^{(1-\varepsilon)x} dy\, |V(y)| \\
&+  \int_{(1-\varepsilon) x}^x dy \, e^{-2 \kappa_N (x-y)}\,|V(y)|.
\end{split}
\end{equation}

\item[\text{\rm(c)}]  If the unperturbed potential $V$ belongs  $L^1_{1+\epsilon}(\mathbb R^+)$ for some fixed $\epsilon \ge 0,$ then 
the perturbed potential $\tilde V$ belongs to $L^1_\epsilon(\mathbb R^+).$

\item[\text{\rm(d)}] Assume that the unperturbed potential $V$ is further restricted to satisfy
\begin{equation}
\label{8.20}
|V(x)| \le c\,e^{-\alpha x}, \qquad x \ge x_0,
\end{equation}
for some positive constants $\alpha$ and $x_0,$ where $c$
denotes a generic constant.
Then, for every $0< \varepsilon <1,$ the potential increment
$\tilde V-V$ has the asymptotic behavior as $x\to+\infty$ given by
\begin{equation}
\label{8.21}
\tilde V(x) -V(x) =
\begin{cases}
O\left( e^{-\alpha x}+ e^{-2  \varepsilon \kappa_N x}\right), \qquad \alpha\le 2  \kappa_N, \\
\noalign{\medskip}
O\left( e^{-2 \varepsilon \kappa_N x}\right),  \qquad \alpha > 2 \kappa_N.
\end{cases}
\end{equation}

\item[\text{\rm(e)}]
If the unperturbed potential $V$ has compact support, then the perturbed potential $\tilde V$ has the asymptotic behavior
\begin{equation}
\label{8.22}
 \tilde V(x)= O\left( x\, e^{-2\kappa_N x}\right), \qquad x \to +\infty.
 \end{equation}

\end{enumerate}

We remark that the estimates given in \eqref{8.18}--\eqref{8.22} are not necessarily sharp.

We refer the reader to \cite{AW2025} for the construction of the remaining relevant quantities
under the further sufficiency assumption that
the unperturbed potential $V$ belongs to $L^1_2(\mathbb R^+),$  
We provide the results from \cite{AW2025} as follows. 
In analogy with \eqref{7.20}, the unperturbed Jost matrix $J(k)$ is transformed into the perturbed Jost matrix as
 \begin{equation}
\label{8.23}
\tilde J(k)=\left[I-\ds\frac{2i\kappa_N}{k+i\kappa_N}\,\tilde P_{N\text{\rm{i}}}\right] J(k),\qquad
k\in\overline{\mathbb C^+},
\end{equation}
where we have defined
\begin{equation}
\label{8.24}
\tilde P_{N\text{\rm{i}}}:= L_{N\text{\rm{i}}} \, \tilde C_{N\text{\rm{i}}} \left(\tilde C_{N\text{\rm{i}}}\,L_{N\text{\rm{i}}}^\dagger \,
L_{N\text{\rm{i}}}\,\tilde C_{N\text{\rm{i}}}\right)^+ \tilde C_{N\text{\rm{i}}} \,L_{N\text{\rm{i}}}^\dagger,
\end{equation}
with $L_{N\text{\rm{i}}}$ being the $n\times n$ matrix appearing 
 in \eqref{8.17}.
 The orthogonal projection  $\tilde Q_{N\text{\rm{i}}}$ 
appearing in \eqref{8.2} projects into a subspace of dimension $m_{N\text{\rm{i}}}$ of the kernel of $\tilde J(i\kappa_N).$
  Moreover, $\tilde P_{N\text{\rm{i}}}$ is an orthogonal projection with the rank $m_{N\text{\rm{i}}},$ where we recall
that $m_{N\text{\rm{i}}}$ is also equal to the rank of the orthogonal projection $\tilde Q_{N\text{\rm{i}}}$ appearing in \eqref{8.2}. The quantity
$\tilde P_{N\text{\rm{i}}}$ projects onto a subspace of the kernel of $\tilde J(i\kappa_N)^\dagger,$ where that
subspace has dimension $m_{N\text{\rm{i}}}.$
With the help of \eqref{8.8}, we can write \eqref{8.24} as
\begin{equation}
\label{8.25}
\tilde P_{N\text{\rm{i}}}=\left( L_{N\text{\rm{i}}}\, \tilde Q_{N\text{\rm{i}}}\, \tilde C_{N\text{\rm{i}}}\right) \left[
\left( L_{N\text{\rm{i}}}\, \tilde Q_{N\text{\rm{i}}}\, \tilde C_{N\text{\rm{i}}}\right)^\dagger
\left(
L_{N\text{\rm{i}}}\, \tilde Q_{N\text{\rm{i}}}  \,\tilde C_{N\text{\rm{i}}}\right)\right]^+ \left( L_{N\text{\rm{i}}}\, \tilde Q_{N\text{\rm{i}}}\, \tilde C_{N\text{\rm{i}}}\right)^\dagger.
\end{equation}
Analogous to \eqref{7.24}, the perturbation changes the determinant of the Jost matrix as
 \begin{equation*}
\det[\tilde J(k)]=\left(\ds\frac{k-i\kappa_N}{k+i\kappa_N}\right)^{m_{N\text{\rm{i}}}} \det[J(k)],\qquad
k\in\overline{\mathbb C^+},\end{equation*}
which indicates that the total number of bound states including the multiplicities
is increased by $m_{N\text{\rm{i}}}.$
 Under the perturbation, the scattering matrix $S(k)$ undergoes the transformation
 \begin{equation*}
\tilde S(k)=\left[I+\ds\frac{2i \kappa_N}{k-i\kappa_N}\,\tilde P_{N\text{\rm{i}}}\right] S(k)\left[I+\ds\frac{2i\kappa_N}{k-i\kappa_N}\,\tilde P_{N\text{\rm{i}}}\right],
\qquad k\in\mathbb R,
\end{equation*}
which is analogous to \eqref{7.25},
and the determinant of the scattering matrix changes as
  \begin{equation*}
\det[\tilde S(k)]=\left(\ds\frac{k+i\kappa_N}{k-i\kappa_N}\right)^{2\, m_{N\text{\rm{i}}}} \det[S(k)],
\qquad k\in\mathbb R,
\end{equation*}
which is the analog of \eqref{7.26}.
The perturbation changes the Jost solution $f(k,x)$ into $\tilde f(k,x)$ given by
\begin{equation}
\label{8.29}
\tilde f(k,x)=\left[f(k,x)-\ds\frac{1}{k^2+\kappa_N^2}
\,\xi_{N\text{\rm{i}}}(x) \,\Omega_{N\text{\rm{i}}}(x)^+ q_{6}(x) \right] \left[I- \ds\frac{2i \kappa_N}{k+i \kappa_N} \tilde P_{N\text{\rm{i}}}\right],
\end{equation}
where we have defined
\begin{equation*}
q_{6}(x):=\xi'_{N\text{\rm{i}}}(x)^\dagger\,f(k,x)-\xi_{N\text{\rm{i}}}(x)^\dagger\,f'(k,x).
\end{equation*}
 and the singularity at $k=i\kappa_N$ appearing
 on the right-hand side of \eqref{8.29} is a removable singularity.
 We note that \eqref{8.29} is the analog of \eqref{7.27}.
Under the perturbation,
the projection matrices $Q_j$ and the Gel'fand--Levitan normalization matrices for $1\le j\le N-1$ remain unchanged, i.e. we have
 \begin{equation*}
\tilde Q_j=Q_j,\quad \tilde C_j=C_j,
\qquad 1\le j\le N-1.
\end{equation*}

Let us consider the unperturbed problem described in Section~\ref{section7} 
with $N$ bound states
by assuming that the unperturbed potential $V$ satisfies \eqref{2.2} and belongs to
$L^1_3(\mathbb R^+).$ Let us add a bound state at $k=i\tilde\kappa_{N+1}$
with multiplicity $\tilde m_{N+1}$ and with the Gel'fand--Levitan bound-state normalization matrix
$\tilde C_{N+1},$
where the spectral measure changes from $d\rho$ to $d\tilde\rho$ as
in \eqref{7.5}.
We would like to analyze the addition of the bound state in two different ways.
The first way is to add the aforementioned bound state in one step, and the second way is to add the bound state
in two steps.
The next theorem shows that, in either way, we end up with the same perturbed
quantities, i.e. each of the two ways yields the same perturbed regular solution,
boundary matrices, potential, Jost solution, Jost matrix, scattering matrix, and physical solution.
We remark that the assumption that the unperturbed potential $V$ belongs to
$L^1_3(\mathbb R^+)$ is a sufficiency assumption and is used to
ensure that the perturbed potential obtained by adding the bound state
in two steps belongs to $L^1_1(\mathbb R^+).$

\begin{theorem}
\label{theorem8.1} Consider the unperturbed problem described in Section~\ref{section7}
with $N$ bound states. Assume that 
the unperturbed potential $V$ satisfies \eqref{2.2} and belongs to
$L^1_3(\mathbb R^+).$ Let us add a bound state at $k=i\tilde\kappa_{N+1}$
with multiplicity $\tilde m_{N+1}$ and with the Gel'fand--Levitan bound-state normalization matrix
$\tilde C_{N+1}$ in such a way that
the spectral measure changes from $d\rho$ to $d\tilde\rho$ as
in \eqref{7.5}.
Let $d\rho$ be the unperturbed spectral measure, $\tilde\varphi(k,x)$ denote the corresponding perturbed
regular solution, $\tilde A$ and $\tilde B$ be the perturbed
boundary matrices, $\tilde V(x)$ be the perturbed potential, 
$\tilde f(k,x)$ be the perturbed Jost solution, $\tilde J(k)$ be the perturbed Jost matrix,
$\tilde S(k)$ be the perturbed scattering matrix, and $\tilde\Psi(k,x)$ be the perturbed physical solution.
Alternatively, consider adding the same bound state with the same multiplicity and the same
Gel'fand--Levitan normalization matrix in two steps, where the unperturbed potential
changes from $V$ to $\hat V$ in the first step and from $\hat V$ to $\hat{\hat V}$ in the second step.
Let us use $d\hat\rho,$
 $\hat\varphi(k,x),$ $\hat A,$ $\hat B,$ $\hat f(k,x),$ $\hat J(k,x),$ $\hat S(k),$ $\hat\Psi(k,x)$ to
 denote the respective perturbed quantities
after the first step and 
use $d\hat{\hat\rho},$
$\hat{\hat\varphi}(k,x),$ $\hat{\hat A},$ $\hat{\hat B},$
$\hat{\hat f}(k,x),$ $\hat{\hat J}(k,x),$ $\hat{\hat S}(k),$ $\hat{\hat\Psi}(k,x)$ to
denote the respective perturbed quantities
after the second step.
We know that $d\tilde\rho=d\hat{\hat\rho}.$
Then, the perturbed quantities $\tilde V(x),$ $\tilde\varphi(k,x),$ $\tilde A,$ $\tilde B,$
$\tilde f(k,x),$ $\tilde J(k),$
$\tilde S(k),$
 $\tilde \Psi(k)$ coincides with their respective counterparts $\hat{\hat V}(x),$
 $\hat{\hat\varphi}(k,x),$ $\hat{\hat A},$ $\hat{\hat B},$
$\hat{\hat f}(k,x),$ $\hat{\hat J}(k,x),$ $\hat{\hat S}(k),$ $\hat{\hat\Psi}(k,x).$

\end{theorem}

\begin{proof}
From Theorem~5.3 of \cite{AW2025} we know that
the perturbed regular solution
$\tilde\varphi(k,x)$ and the unperturbed regular solution
$\varphi(k,x)$ are related to each other as
in \eqref{4.14}, where $\mathcal A(x,y)$ satisfies \eqref{4.16} with
$G(x,y)$ expressed in terms of $d\tilde\rho-d\rho$ as in \eqref{4.17}.
Using the same theorem we get
 \begin{equation}
 \label{8.32}
        \hat{\varphi}(k,x)= \varphi(k,x)+ \int_0^x dy\,\hat{\mathcal A}(x, y)\, \varphi(k,y),
\end{equation}
 \begin{equation}
 \label{8.33}
 \hat{ \hat{\varphi}}(k,x)=
        \hat{ \varphi}(k,x)
        + \int_0^x dy\,
       \hat{\hat{\mathcal A}} (x, y)\,
        \hat{\varphi}(k,y),
\end{equation}
where the quantity
$\hat{\mathcal A}(x,y)$ is related to
$d\hat\rho-d\rho$ and the quantity
$\hat{\hat{\mathcal A}}(x,y)$ is related to
$d\hat{\hat\rho}-d\hat\rho$ in a similar manner as in
\eqref{4.14}, \eqref{4.16}, and \eqref{4.17}.
From \eqref{8.32} and \eqref{8.33}
we obtain
 \begin{equation}
 \label{8.34}
 \hat{ \hat{\varphi}}(k,x)=
        \varphi(k,x)
        + \int_0^x dy\,
      \mathcal B(x, y)\,
        \varphi(k,y),
\end{equation}
where the quantity $\mathcal B(x,y)$ is given by
 \begin{equation*}
\mathcal B(x,y):=  \hat{\mathcal A}(x,y)+\hat{\hat{\mathcal A}}(x,y)+\int_y^x dz\,\hat{\hat{\mathcal A}}(x,z)\,
 \hat{\mathcal A}(z,y).
\end{equation*}
As indicated in Theorem~5.3 of \cite{AW2025}, we can invert the Volterra integral equation \eqref{8.34} and
express $\varphi(k,x)$ in terms of $\hat{ \hat{\varphi}}(k,x)$ as
 \begin{equation}
 \label{8.36}
\varphi(k,x)=\hat{\hat{\varphi}}(k,x)+ \int_0^x dy\,{\mathcal C}(x, y) \, \hat{\hat{\varphi}}(k,y),
\end{equation}
for some appropriate quantity $\mathcal C(x,y)$ determined by
$d\rho-d\hat{\hat\rho}.$
Using \eqref{8.36} on the right-hand side of \eqref{4.14}, we obtain
 \begin{equation}
 \label{8.37}
\tilde{\varphi}(k,x)=\hat{\hat\varphi}(k,x)+ \int_0^x dy\,\mathcal D(x, y) \,\hat{\hat\varphi}(k,y),
\end{equation}
where we have
 \begin{equation*}
\mathcal D(x,y):= {\mathcal C}(x,y)
 +\hat{\mathcal A}(x,y)+\int_y^x {\hat{\mathcal A}}(x,z)\,
  {\mathcal C}(z,y) \,dz.
\end{equation*}
Since $d\hat{\hat\rho}= d\tilde\rho,$ using \eqref{8.37} and Theorem~5.3 of \cite{AW2025} we
obtain $\mathcal D(x,y)=0.$
Consequently, from \eqref{8.37} we get
 \begin{equation}
  \label{8.39}
\tilde\varphi(k,x)=\hat{\hat{\varphi}}(k,x),\qquad k\in\mathbb C,\quad x\in\mathbb R^+.
\end{equation}
From the analog of \eqref{4.15} we know that $\tilde\varphi(k,x)$ is a solution to
the perturbed Schr\"odinger equation with the potential
$\tilde V$ and that $\hat{\hat{\varphi}}(k,x)$ is a solution to the
the perturbed Schr\"odinger equation with the potential $\hat{\hat V}.$ By using \eqref{8.39} in those two equations
and by subtracting one equation from the other, we obtain
 \begin{equation}
  \label{8.40}
  [{\tilde V}(x)- \hat{\hat V}(x)]\,\tilde\varphi(k,x)=0,\qquad k\in\mathbb C,\quad x\in\mathbb R^+.
\end{equation}
The large $k$-asymptotics of the regular solution is known. Using (3.2.201) of \cite{AW2021} we have

 \begin{equation}
   \label{8.41}
e^{ikx}\,\tilde\varphi(k,x) = \frac{1}{2}(1+e^{2ikx}) \tilde A +\frac{1}{2} (1-e^{2ikx}) \tilde  B + O\left(\ds\frac{1}{k}\right),\qquad k\to\pm\infty,
\end{equation}
 \begin{equation}
   \label{8.42}
e^{ikx}\,\hat{\hat\varphi}(k,x) = \frac{1}{2}(1+e^{2ikx}) \hat{\hat A} +\frac{1}{2} (1-e^{2ikx}) \hat{\hat  B} + O\left(\ds\frac{1}{k}\right),\qquad k\to\pm\infty.
\end{equation}
From \eqref{8.39}, \eqref{8.41}, and \eqref{8.42} we see that
the pair of boundary matrices $(\tilde A,\tilde B)$ agrees with the pair of
boundary matrices $(\hat{\hat A},\hat{\hat B}),$ i.e. we have
 \begin{equation*}
\tilde A=\hat{\hat A}, \quad \tilde B=\hat{\hat B}.
\end{equation*}
By multiplying both sides of \eqref{8.40} and letting $k\to\pm$ and using \eqref{8.41} in the resulting equation,
we conclude that either we have 
$\tilde V(x)-\hat{\hat V}(x)\equiv 0$ 
or we have
 \begin{equation}
   \label{8.44}
\tilde A+\tilde B=0,\quad \tilde A-\tilde B=0.
\end{equation}
The latter, i.e. the two equalities in \eqref{8.44}, cannot happen because that would imply $(\tilde A,\tilde B)=(0,0).$
However, that would in turn imply that the perturbed regular solution $\tilde\varphi(k,x),$ as the solution to the Schr\"odinger equation
with the zero initial conditions at $x=0,$ would be identically zero. That cannot happen because of the analog of
\eqref{2.9} satisfied by $\tilde A$ and $\tilde B.$ Thus, we must have $\tilde V(x)\equiv\hat{\hat V}(x).$ 
From the uniqueness of the Jost solution to the Schr\"odinger equation, we then conclude that
$\tilde f(k,x)\equiv \hat{\hat f}(k,x).$
The agreement of the Jost solutions and the boundary matrices, as seen from \eqref{2.11} implies that
the Jost matrices $\tilde J(k)$ and $\hat{\hat J}(k)$ agree. 
Then, using \eqref{2.10} we conclude that
the scattering matrices $\tilde S(k)$ and $\hat{\hat S}(k)$ agree.
Finally, using \eqref{2.14} we conclude that the physical solutions
 $\tilde\Psi(k,x)$ and $\hat{\hat\Psi}(k,x)$ agree.

\end{proof}

\section{Illustrative examples of the relevant quantities}
\label{section9}

In this section we illustrate the construction and usage of various quantities relevant to the
half-line matrix-valued Schr\"odinger operator with a selfadjoint boundary condition.
In order to present the relevant quantities in their simplest form, we at times display only their numerical values
in case their exact expressions are too lengthy.

In the first example we demonstrate a way to construct some explicit examples of $2\times 2$ matrix-valued potentials on the half line and the associated
Jost solutions by starting with some scalar-valued potentials and the corresponding Jost solutions.

\begin{example}
\label{example9.1} 
\normalfont
Assume that $V_1$ and $V_2$ are two real-valued potentials used in \eqref{2.1} in the scalar case, i.e. when $n=1,$
and let us use $f_1(k,x)$ and $f_2(k,x),$ respectively, to denote the associated Jost solutions satisfying \eqref{2.10}. 
It can directly be verified that the $2\times 2$ matrix-valued potential $V$ defined
as
\begin{equation}
\label{9.1}
V (x)= \ds\frac{1}{2}\begin{bmatrix}V_1(x)+V_2(x)&V_1(x)-V_2(x)\\
\noalign{\medskip}
V_1(x)-V_2(x)&V_1(x)+V_2(x)\end{bmatrix},
 \end{equation}
 has the associated Jost solution $f(k,x)$ given by
 \begin{equation}
\label{9.2}
f(k,x)= \ds\frac{1}{2}\begin{bmatrix}f_1(k,x)+f_2(k,x)&f_1(k,x)-f_2(k,x)\\
\noalign{\medskip}
f_1(k,x)-f_2(k,x)&f_1(k,x)+f_2(k,x)\end{bmatrix}.
 \end{equation}
 Using $f(k,x)$ and $f(-k,x)$ as a fundamental set of solutions to \eqref{2.1} for $k\in\mathbb R\setminus\{0\},$ 
 we can then construct any other solution  to \eqref{2.1}.
 In the special case  $V_2(x)\equiv 0.$ 
  we have
 $f_2(k,x)\equiv e^{ikx}.$ In that case,
from \eqref{9.1} and \eqref{9.2} we obtain the corresponding potential
 and its associated Jost solution as
 \begin{equation}
\label{9.3}
V (x)= \ds\frac{1}{2}\begin{bmatrix}V_1(x)&V_1(x)\\
\noalign{\medskip}
V_1(x)&V_1(x)\end{bmatrix},\quad 
f(k,x)= \ds\frac{1}{2}\begin{bmatrix}f_1(k,x)+e^{ikx}&f_1(x)-e^{ikx}\\
\noalign{\medskip}
f_1(k,x)-e^{ikx}&f_1(k,x)+e^{ikx}\end{bmatrix}.
 \end{equation}
 In a similar way, we can construct  $2\times 2$ matrix-valued selfadjoint potentials
 that are not real valued. In that case, instead of using \eqref{9.1} we can let
 \begin{equation*}
V (x)= \ds\frac{1}{2}\begin{bmatrix}V_1(x)+V_2(x)&i\,V_1(x)-i\,V_2(x)\\
\noalign{\medskip}
-i\,V_1(x)+i\,V_2(x)&V_1(x)+V_2(x)\end{bmatrix},
 \end{equation*}
and use the corresponding Jost solution $f(k,x)$ given by
 \begin{equation*}
f(k,x)= \ds\frac{1}{2}\begin{bmatrix}f_1(k,x)+f_2(k,x)&i\,f_1(k,x)-i\,f_2(k,x)\\
\noalign{\medskip}
-i\,f_1(k,x)+i\,f_2(k,x)&f_1(k,x)+f_2(k,x)\end{bmatrix}.
 \end{equation*}
 This procedure also yields $n\times n$ matrix valued potentials and their
 associated Jost solutions 
 by using $n$ scalar-valued potentials $V_j$ and their associated Jost solutions
 $f_j(k,x),$ where $1\le j\le n.$ We refer the reader to Section~6.2 of
 \cite{AW2021} for the details of such constructions as well as
 some other ways of producing matrix-valued potentials and their
 associated Jost solutions.
 
 \end{example}

 In the next example, we present some particular scalar-valued potentials
 and their associated Jost solutions, which can be used as the starting elements
 to construct the corresponding matrix-valued quantities by the method of Example~\ref{example9.1}.

 \begin{example}
\label{example9.2} 
\normalfont
We present the scalar-valued potential $V_1$ and its associated Jost solution $f_1(k,x)$ given by
\begin{equation}
\label{9.6}
V_1 (x)= -\ds\frac{8\alpha\epsilon\beta^2\,e^{2\beta x}}{(\alpha+\epsilon\,e^{2\beta x})^2},\quad 
f_1(k,x)=e^{ikx}\left[1-\ds\frac{2i\alpha\beta}{(k+i\beta)(\alpha+\epsilon\,e^{2\beta x})}\right],\qquad x\in\mathbb R^+,
 \end{equation}
where $\alpha,$ $\epsilon,$ and $\beta$ are positive parameters. Even though we use three
parameters in \eqref{9.6}, there are actually two independent parameters given by
$\alpha/\epsilon$ and $\beta.$
Next, we introduce
the scalar-valued potential $V_2$ and its associated Jost solution $f_2(k,x)$ given by
 \begin{equation*}
V_2(x)= \ds\frac{2}{(x+a)^2},\quad 
f_2(k,x)=e^{ikx}\left[1+\ds\frac{i}{k(x+a)}\right],\qquad x\in\mathbb R^+,
 \end{equation*}
where $a$ is a positive parameter. With the help of \eqref{2.4} and \eqref{2.5} we observe that
$V_2$ is integrable but does not belong to $L_1^1(\mathbb R^+).$

\end{example}

 In the next example, starting with a $2\times 2$ matrix-valued potential $V$ and a $2\times 2$ boundary matrices $A$ and $B,$ 
 we illustrate the derivation of the corresponding relevant particular solutions, Jost matrix, and scattering matrix. In some of the computations we have used 
 the symbolic math software Mathematica to
do symbolic computations. For the convenience of displaying some of
the results in their simplest forms, we use an overbar on a digit
to indicate a roundoff in that digit.
 
 \begin{example}
\label{example9.3} 
\normalfont
As the boundary matrices $A$ and $B$ appearing in \eqref{2.6} and \eqref{2.7}, let us use
\begin{equation}
\label{9.8}
A=\begin{bmatrix} 1&2\\
0&3\end{bmatrix},\quad B=\begin{bmatrix} 4&20\\
4&-2\end{bmatrix},\end{equation}
which satisfy \eqref{2.7} and \eqref{2.8}. 
The positivity of
the matrix given in \eqref{2.8} can be verified by the fact that
we have
\begin{equation*}
A^\dagger A+B^\dagger B=\begin{bmatrix} 33&74\\
74&417\end{bmatrix},\end{equation*}
and the eigenvalues of the matrix on the right-hand side of \eqref{9.8} are
evaluated as $430.76\overline{7}$ and $19.233\overline{1},$
which are both positive. 
We choose $V_2(x)\equiv 0$ and construct $V_1(x)$ by using the parameter values $\alpha=2,$ $\epsilon=1,$ and $\beta=1/3$ in 
the first equality of \eqref{9.6}.
Using the resulting potential $V_1(x)$ in the first equality of \eqref{9.3}, we obtain the $2\times 2$ matrix-valued potential $V$ given by
\begin{equation}
\label{9.10}
V(x)=-\ds\frac{8\,e^{2x/3}}{9(2+e^{2x/3})^2}\begin{bmatrix}  1 & 1 \\
\noalign{\medskip}
 1& 1\end{bmatrix},\qquad x\in\mathbb R^+,\end{equation}
 which satisfies \eqref{2.2} and \eqref{2.4}.
 Using the corresponding parameter
 values 
  $\alpha=2,$ $\epsilon=1,$ and $\beta=1/3$ in 
  the second equality of \eqref{9.6}, we construct
  $f_1(k,x).$
 Then, with the help of the second equality of \eqref{9.3}, we obtain the 
 associated Jost solution satisfying \eqref{2.1} and \eqref{2.10} as
 \begin{equation}
\label{9.11}
f(k,x)=e^{ikx}\begin{bmatrix} 1-\ds\frac{2 i}{(3 k+i)(2+e^{2x/3})}&-\ds\frac{2 i}{(3 k+i)(2+e^{2x/3})}
  \\
\noalign{\medskip}
 -\ds\frac{2 i}{(3 k+i)(2+e^{2x/3})}
 &1-\ds\frac{2 i}{(3 k+i)(2+e^{2x/3})}
 \end{bmatrix},\qquad x\in\mathbb R^+.\end{equation}
 From \eqref{9.11} we get
\begin{equation}
\label{9.12}
f(k,0)=\begin{bmatrix}
 \ds\frac{9k+i}{3(3k+i)}&-\ds\frac{2i}{3(3k+i)}\\
\noalign{\medskip}
 -\ds\frac{2i}{3(3k+i)}&\ds\frac{9k+i}{3(3k+i)}\end{bmatrix},\end{equation}
\begin{equation}
\label{9.13}
f'(k,0)=\begin{bmatrix}
 \ds\frac{81ik^2-9k+4i}{27(3k+i)}&\ds\frac{18k+4i}{27(3k+i)}\\
\noalign{\medskip}
 \ds\frac{18k+4i}{27(3k+i)}&\ds\frac{81ik^2-9k+4i}{27(3k+i)}\end{bmatrix}.\end{equation}
Using \eqref{9.8}, \eqref{9.12}, and \eqref{9.13} in \eqref{2.11},
we construct the corresponding Jost matrix $J(k)$ as
\begin{equation}
\label{9.14}
J(k)=\begin{bmatrix} \ds\frac{-81ik^2+333k-40i}{27(3k+ i)}&\ds\frac{-162 ik^2+1584 k+196i}{27(3k+ i)}\\
\noalign{\medskip}
 \ds\frac{306k-40i}{27(3k+ i)}&\ds\frac{-243ik^2-171k-398i}{27(3k+ i)}\end{bmatrix},\qquad k\in\mathbb R.\end{equation}
From \eqref{9.14} we obtain
\begin{equation*}
J(0)=\ds\frac{1}{27}\begin{bmatrix}
 -40&196\\
 \noalign{\medskip}
 -40& -398\end{bmatrix},\end{equation*}
 with the eigenvalues of $J(0)$ given by
 $-13.872\overline{8}$ and $-2.3493\overline{8}.$
 Note also that the determinant of $J(0)$ is given by
\begin{equation*}
\det[J(0)]=\ds\frac{880}{27},\end{equation*}
 which indicates that $J(k)$ is invertible at $k=0.$ Since $J(0)$ is invertible, this corresponds to
 a generic case \cite{AW2021}.
 The adjoint of the Jost matrix in \eqref{9.14} is evaluated as
\begin{equation*}
J(k)^\dagger=\begin{bmatrix} \ds\frac{81ik^2+333k+40i}{27(3k-i)}&\ds\frac{306k+40i}{27(3k-i)}\\
\noalign{\medskip}
\ds\frac{162 ik^2+1584 k-196i}{27(3k-i)}&\ds\frac{243ik^2-171k+398i}{27(3k-i)}\end{bmatrix},\qquad k\in\mathbb R.\end{equation*}
Next, we evaluate the scattering matrix $S(k).$ From \eqref{9.14} we
have
\begin{equation}
\label{9.18}
J(-k)=
\begin{bmatrix}\ds\frac{81ik^2+333k+40i}{27(3k-i)}&\ds\frac{162 ik^2+1584 k-196i}{27(3k-i)}\\
\noalign{\medskip}
 \ds\frac{306k+40i}{27(3k-i)}&\ds\frac{243ik^2-171k+398i}{27(3k-i)}\end{bmatrix},\qquad k\in\mathbb R.\end{equation}
Hence, using \eqref{9.14} and \eqref{9.18} in \eqref{2.12}
we obtain the scattering matrix
as
\begin{equation}
\label{9.19}
S(k)=\begin{bmatrix}
s_{11}(k)& s_{12}(k)\\
s_{12}(k)& s_{11}(k)\end{bmatrix},\end{equation}
where we have defined
\begin{equation}
\label{9.20}
s_{11}(k):=  \ds\frac{729\, k^4 - 5346\,i\, k^3- 18747 \,k^2
- 594 \,i\, k   +880  }{(3k-i)(243 k^3+ 135 \,i\, k^2 +7170 \,k  -880\,i )},
\end{equation}
\begin{equation}
\label{9.21}
s_{12}(k):=  \ds\frac{-12ik(459\,k^2+794)}{(3k-i)(243 k^3+ 135
\,i\, k^2 +7170 \,k  -880\,i )}.
\end{equation}
Using \eqref{9.19}--\eqref{9.21}, it can be verified directly that the scattering matrix $S(k)$ satisfies
\begin{equation*}
S(-k)=S(k)^\dagger=S(k)^{-1},\qquad k\in\mathbb R,
\end{equation*}
confirming (2.5.4) of \cite{AW2021}.
From \eqref{9.19}--\eqref{9.21}, we obtain the large $k$-asymptotics of the scattering matrix as
\begin{equation*}
S(k)=\begin{bmatrix}1&0\\
0&1\end{bmatrix}
+\ds\frac{1}{9ik}\begin{bmatrix}68&68\\
68&-64\end{bmatrix}
+O\left(\ds\frac{1}{k^2}\right),\qquad k\to\pm\infty,\end{equation*}
and we also get the small $k$-asymptotics of $S(k)$ as
\begin{equation*}
S(k)=\begin{bmatrix}-1&0\\
0&-1\end{bmatrix}+O(k),\qquad k\to 0 \text{\rm{ in }} \mathbb R.\end{equation*}
The physical solution $\Psi(k,x)$ is obtained by using \eqref{9.11} and \eqref{9.19}--\eqref{9.21}
in \eqref{2.13}, but it is too lengthy to display here.
Using $x=0,$ from the explicitly evaluated physical solution $\Psi(k,x)$ and its $x$-derivative
$\Psi'(k,x),$ we have
\begin{equation}
\label{9.25}
\Psi(k,0)=\ds\frac{1}{q_{7}(k)}
\begin{bmatrix}
6k(81\,k^2-261 \,ik+106)& -12 i k(153\,k+46\,i)\\
\noalign{\medskip}
-12 i k(153\,k-20\,i)& 6k(81\,k^2+333 \,ik+40)\end{bmatrix},\end{equation}
\begin{equation}
\label{9.26}
\Psi'(k,0)=\ds\frac{1}{q_{7}(k)}
\begin{bmatrix}
72k(27\,k^2-189 \,ik+22)& 72k(27\,k^2+9\,i\,k+44)\\
\noalign{\medskip}
8\, k(243\,k^2-18\,ik+418)& -4k(405\,k^2+3501 \,ik-352)\end{bmatrix},\end{equation}
where we have defined
\begin{equation}
\label{9.27}
q_{7}(k):=243 k^3+ 135 \,i\, k^2 +7170 \,k  -880\,i.\end{equation}
Using \eqref{9.8} and \eqref{9.25}--\eqref{9.27} on the left-hand side of \eqref{2.7} 
we confirm that the physical solution $\Psi(k,x)$ satisfies the 
boundary condition given in \eqref{2.7}.
Next, we evaluate the regular solution $\varphi(k,x)$ by using \eqref{2.14} with input from the physical solution 
$\Psi(k,x)$ and the Jost matrix $J(k).$ We also alternatively
 evaluate the regular solution by using
\eqref{9.2}, \eqref{9.14}, and \eqref{9.18} in \eqref{2.15}. We confirm that we get
the same expression for $\varphi(k,x)$ from either evaluation. We get

\begin{equation}
\label{9.28}
\varphi(k,x)=\ds\frac{1}{q_{8}(k,x)}
\begin{bmatrix}q_{9}(k,x)-q_{9}(-k,x)&q_{10}(k,x)-q_{10}(-k,x)\\
\noalign{\medskip}
q_{11}(k,x)-q_{11}(-k,x)&q_{12}(k,x)-q_{12}(-k,x)\end{bmatrix},\end{equation}
where we have defined
\begin{equation}
\label{9.29}
q_{8}(k,x):=54\,k(3k+i)(3k-i)(2+e^{2x/3}),\end{equation}
\begin{equation}
\label{9.30}
\begin{split}
q_{9}(k,x)
:=&e^{ikx}\left(-80i-372k-1998ik^2+486k^3\right)\\
&+e^{ikx+2x/3}\left(40i+453k-918ik^2+243k^3\right),
\end{split}
\end{equation}
\begin{equation}
\label{9.31}
\begin{split}
q_{10}(k,x)
:=&e^{ikx}\left(-796i-834k-9990ik^2+972k^3\right)\\
&+e^{ikx+2x/3}\left(-196i+996k-4590ik^2+486k^3\right),
\end{split}
\end{equation}
\begin{equation}
\label{9.32}
q_{11}(k,x)
:=e^{ikx}\left(-80i-426k-1998ik^2\right)\\
+e^{ikx+2x/3}\left(40i+426k-918ik^2\right),
\end{equation}
\begin{equation}
\label{9.33}
\begin{split}
q_{12}(k,x)
:=&e^{ikx}\left(392i-780k+702ik^2+1458k^3\right)\\
&+e^{ikx+2x/3}\left(398i+1023k+756ik^2+729k^3\right).
\end{split}
\end{equation}
By evaluating $\varphi(k,0)$ and $\varphi'(k,0)$ and using
\eqref{9.8}, we observe that the regular solution satisfies the initial conditions \eqref{2.6}
with the boundary matrices $A$ and $B$ appearing in \eqref{9.8}.
As $x\to+\infty$ in \eqref{9.28}, with the help of \eqref{9.29}--\eqref{9.33} we observe that,
unless $k$ is real,
each of the four entries of
$\varphi(k,x)$ blows up exponentially as
$x\to+\infty.$ 
We remark that the right-hand side of \eqref{9.28} is indeed entire in $k.$
The three zeros of $q_{8}(k,x)$ in $k$
all correspond to removable singularities for the right-hand side
of \eqref{9.28}.

\end{example}

Next, we consider the bound states associated with
the potential $V$ in \eqref{9.10} and
the boundary matrices $A$ and $B$ appearing in \eqref{9.8}.
The matrix case is naturally more complicated than the scalar case.
 This is due to the fact that the bound states are simple in the scalar case whereas their multiplicities must be taken into account
 in the matrix case. In the next example, we illustrate the construction of the relevant
 orthogonal projections
 used in the analysis of bound states in the matrix case.

 \begin{example}
\label{example9.4} 
\normalfont
In this example, we continue using the potential $V$ in \eqref{9.10} as well as the
boundary matrices $A$ and $B$ given in \eqref{9.8}.
Hence, this example is the continuation of Example~\ref{example9.3}.
We analyze the determinant of the Jost matrix $J(k).$
 From \eqref{9.14} we have
\begin{equation}
\label{9.34}
\det[J(k)]=\ds\frac{-243 k^3-135 i k^2-7170 k+880 i}{27(3k+i)}.\end{equation}
We observe that $\det[J(k)]$ has three simple zeros that are given by
\begin{equation}
\label{9.35}
k=i\kappa_1,\quad k=i\kappa_2,\quad k=-5.7741\overline{9}i,\end{equation}
 where we have let
\begin{equation}
\label{9.36}
\kappa_1:=5.09554\overline{8},\quad \kappa_2:=0.12308204\overline{1}.\end{equation}
From \eqref{9.34}--\eqref{9.36} we see that, among the three zeros of 
 $\det[J(k)],$ only two of them occur in $\mathbb C^+.$ In fact, those two simple zeros listed in \eqref{9.36} occur
 on the positive imaginary axis in
 $\mathbb C^+,$ as they must.
 Hence,
we have two simple bound states
 occurring at $k=i\kappa_1$ and
 $k=i\kappa_2,$ respectively.
 Let us evaluate $J(i\kappa_1)$ and $J(i\kappa_2)$ by using \eqref{9.36} in \eqref{9.14}. We get
 \begin{equation}
\label{9.37}
J(i\kappa_1)=\begin{bmatrix}
 8.5504\overline{1} & 28.365\overline{9} \\
 \noalign{\medskip}
 3.4548\overline{6} & 11.461\overline{5} \end{bmatrix},\quad
 J(i\kappa_2)=
 \begin{bmatrix}
 0.059870\overline{8} & 10.641\overline{6} \\
 \noalign{\medskip}
 -0.063211\overline{2} & -11.235\overline{3} \end{bmatrix}.\end{equation}
  From \eqref{9.37} we see that
\begin{equation*}
J(i\kappa_1)^\dagger=\begin{bmatrix}
 8.5504\overline{1} &  3.4548\overline{6}  \\
\noalign{\medskip}
28.365\overline{9} & 11.461\overline{5} \end{bmatrix},\quad
 J(i\kappa_2)^\dagger=
\begin{bmatrix}
 0.059870\overline{8} &-0.063211\overline{2}  \\
 \noalign{\medskip}
 10.641\overline{6} & -11.235\overline{3} \end{bmatrix}.\end{equation*}
Next, we determine the kernels of the matrices
  $J(i\kappa_1),$ $J(i\kappa_1)^\dagger,$ $J(i\kappa_2),$ and $J(i\kappa_2)^\dagger.$
   From \eqref{9.36} and \eqref{9.37} we observe that the two rows
   in each of those four matrices are proportional, and hence each
   of those four matrices has the rank equal to one and that
   the nullity of each of those four matrices is $2-1,$ which is equal to $1.$
   We see that each of the following four column vectors forms a basis
   for the kernels
   of $J(i\kappa_1),$ $J(i\kappa_1)^\dagger,$ $J(i\kappa_2),$ and $J(i\kappa_1)^\dagger,$
   respectively:
\begin{equation}
\label{9.39}
\begin{bmatrix} 28.365\overline{9}
\\ 
\noalign{\medskip}
-8.5504\overline{1}
\end{bmatrix}, \quad
\begin{bmatrix} 3.4548\overline{6}
\\ 
\noalign{\medskip}
-8.5504\overline{1}
\end{bmatrix}, \quad
\begin{bmatrix} 10.641\overline{6}
\\ 
\noalign{\medskip}
-0.059870\overline{8}
\end{bmatrix}, \quad
\begin{bmatrix} 0.063211\overline{2}
\\ 
\noalign{\medskip}
0.059870\overline{8}
\end{bmatrix}
.\end{equation}
Hence, normalizing the vectors in \eqref{9.39} we obtain the unit column vectors
$w_1,$ $v_1,$ $w_2,$ $v_2,$ generating $\text{\rm{Ker}}[J(i\kappa_1)],$ $\text{\rm{Ker}}[J(i\kappa_1)^\dagger],$ $\text{\rm{Ker}}[J(i\kappa_2)],$ and $\text{\rm{Ker}}[J(i\kappa_1)^\dagger],$
   respectively,
   where we have let
\begin{equation}
\label{9.40}
w_1:=\begin{bmatrix}
0.95744\overline{8}\\
\noalign{\medskip}
-0.28860\overline{6}\end{bmatrix}
,\quad
v_1:=\begin{bmatrix}
0.37463\overline{2}\\
\noalign{\medskip}
-0.92717\overline{4}\end{bmatrix}
,\end{equation}
\begin{equation}
\label{9.41}
w_2:=\begin{bmatrix}
0.99998\overline{4}\\
\noalign{\medskip}
-0.0056260\overline{3}\end{bmatrix}
,\quad
v_2:=\begin{bmatrix}
0.7260\overline{3}\\
\noalign{\medskip}
0.68766\overline{3}\end{bmatrix}.
\end{equation}
As indicated in (3.8) of \cite{AW2025}, we have
the orthogonal projection matrices $Q_1,$ $P_1, $ $Q_2,$ and $P_2$ onto
the kernels of $J(i\kappa_1),$ $J(i\kappa_1)^\dagger,$ $J(i\kappa_2),$
and $J(i\kappa_2)^\dagger,$ respectively, given by
\begin{equation}
\label{9.42}
Q_1=w_1 w_1^\dagger, \quad P_1=v_1 v_1^\dagger, \quad Q_2=w_2\, w_2^\dagger, \quad P_2=v_2\, v_2^\dagger.
\end{equation}
Using \eqref{9.40} and \eqref{9.41} in
\eqref{9.42}, we obtain the
   corresponding
   orthogonal projections
   $Q_1,$ $P_1,$ $Q_2,$ $P_2$ as
\begin{equation}
\label{9.43}
Q_1=
\begin{bmatrix}
0.91670\overline{7} & -0.27632\overline{5} \\
\noalign{\medskip} -0.27632\overline{5} & 0.083293\overline{3} \end{bmatrix},\quad
P_1=
\begin{bmatrix}
0.14034\overline{9} & -0.34734\overline{9} \\
\noalign{\medskip}
 -0.34734\overline{9} & 0.85965\overline{1}\end{bmatrix},\end{equation}
 \begin{equation}
\label{9.44}
Q_2=
\begin{bmatrix}
0.99996\overline{8} & -0.0056259\overline{4} \\
\noalign{\medskip} -0.0056259\overline{4} & 0.000031652\overline{2} \end{bmatrix},\quad
P_2=
\begin{bmatrix}
0.52711\overline{9} & 0.49926\overline{4} \\
\noalign{\medskip}
0.49926\overline{4} & 0.47288\overline{1}\end{bmatrix}.\end{equation}
 One can directly verify that the matrices appearing in
 \eqref{9.43} and \eqref{9.44} satisfy the appropriate
 equalities stated in
 \eqref{3.2}.
 One can verify that
 $P_1$ and $P_2$ do not commute, as expected in general. In fact, we have
\begin{equation*}
P_1\,P_2-P_2\,P_1=
\begin{bmatrix}
0 & -0.34028\overline{2} \\
\noalign{\medskip}
0.34028\overline{2} & 0\end{bmatrix}.\end{equation*}
Similarly, as expected in general,
$P_1$ and $Q_1$ do not commute because we have
\begin{equation*}
P_1\,Q_1-Q_1\,P_1=
\begin{bmatrix}
0 & 0.48824\overline{6} \\
\noalign{\medskip}
-0.48824\overline{6} & 0\end{bmatrix}.\end{equation*}
Likewise, $P_2$ and $Q_2,$ as expected in general, do not commute and we have
\begin{equation*}
P_2\,Q_2-Q_2\,P_2=
\begin{bmatrix}
0 & -0.49953\overline{8} \\
\noalign{\medskip}
0.49953\overline{8} & 0\end{bmatrix}.\end{equation*}
Similarly, one can verify that
$P_2$ and $Q_1$ do not commute, as expected in general, and in fact we have
\begin{equation*}
P_2\,Q_1-Q_1\,P_2=
\begin{bmatrix}
0 & -0.43108\overline{1} \\
\noalign{\medskip}
0.43108\overline{1} & 0\end{bmatrix}.\end{equation*}
One can appreciate the simplicity of the treatment of the bound states in the scalar case, in which there exists only one orthogonal
projection for each bound state, and that orthogonal projection is simply equal to the real number $1.$
In the scalar case, there is no need to make a distinction between
the kernel of $J(k)$ and the kernel of $J(k)^\dagger$ at a bound-state $k$-value.

\end{example}

In the next example, we illustrate the determination of the
Marchenko normalization matrices at the bound states.
We continue to use the potential and the boundary matrices specified 
in Examples~\ref{example9.3} and \ref{example9.4}.

\begin{example}
\label{example9.5}
\normalfont
In this example, we continue to use as input the potential $V$ specified
in \eqref{9.10} and the boundary matrices $A$ and $B$ specified in
\eqref{9.8}. We recall that the corresponding Jost solution is given in \eqref{9.11}
and we also know that there are two simple bound states at
$k=i\kappa_1$ and $k=i\kappa_2,$ respectively, as indicated in \eqref{9.35}.
Let us now evaluate the Jost solution $f(k,x)$ at $k=i\kappa_1$ and
$k=i\kappa_2,$ where
$\kappa_1$ and $\kappa_2$ are the positive constants given in \eqref{9.36}.
Using \eqref{9.36} in \eqref{9.11} we get
\begin{equation}
\label{9.49}
f(i\kappa_1,x)=e^{ -\kappa_1  x}
\begin{bmatrix} 1 -\ds\frac{0.122\overline{8}}{2 + e^{2 x/3}} &-\ds\frac{0.122\overline{8} }{
2 + e^{2 x/3}}
\\
\noalign{\medskip}
-\ds\frac{0.122\overline{8} }{
2 + e^{2 x/3}} & 1 -\ds\frac{0.122\overline{8} }{2 + e^{2 x/3}}
\end{bmatrix},\end{equation}
\begin{equation}
\label{9.50}
f(i\kappa_2,x)=e^{ -\kappa_2 x}
\begin{bmatrix} 1 -\ds\frac{ 1.4606\overline{6}}{2 + e^{2 x/3}} &-\ds\frac{1.4606\overline{6} }{
2 + e^{2 x/3}}
\\
\noalign{\medskip}
-\ds\frac{1.4606\overline{6} }{
2 + e^{2 x/3}} & 1 -\ds\frac{ 1.4606\overline{6}}{2 + e^{2 x/3}}
\end{bmatrix}.\end{equation}
 From \eqref{9.49} and \eqref{9.50} we observe that
\begin{equation*}
f(i\kappa_1,x)=e^{-\kappa_1 x}\left[I+o(1)\right],
\quad f(i\kappa_2,x)=e^{-\kappa_2 x}\left[I+o(1)\right],
\qquad x\to+\infty,\end{equation*}
which are compatible with \eqref{3.4}.
Using \eqref{9.49} and \eqref{9.50} we get
\begin{equation}
\label{9.52}
\int_0^\infty dx\,f(i\kappa_1,x)^\dagger\,
f(i\kappa_1,x)=
\begin{bmatrix}
0.090584\overline{9} & -0.0075399\overline{2} \\
\noalign{\medskip}
-0.0075399\overline{2} & 0.090584\overline{9}\end{bmatrix},\end{equation}
\begin{equation}
\label{9.53}
\int_0^\infty dx\,f(i\kappa_2,x)^\dagger\,
f(i\kappa_2,x)=
\begin{bmatrix}
2.9955\overline{7} & -1.0667\overline{6} \\
\noalign{\medskip}
-1.0667\overline{6} & 2.9955\overline{7}\end{bmatrix}.\end{equation}
Using \eqref{3.8}, with the help of \eqref{9.52}, \eqref{9.53}, and the second
equalities of \eqref{9.43} and \eqref{9.44},
we evaluate the matrices $\mathbf A_1$ and $\mathbf A_2$ defined as
\begin{equation*}
\mathbf A_1:=\int_0^\infty dx\,P_1\,f(i\kappa_1,x)^\dagger\,
f(i\kappa_1,x)\,P_1,\end{equation*}
\begin{equation*}
\mathbf A_2:=\int_0^\infty dx\,P_2\,f(i\kappa_2,x)^\dagger\,
f(i\kappa_2,x)\,P_2,\end{equation*}
and we get
\begin{equation}
\label{9.56}
\mathbf A_1=\begin{bmatrix}
0.013448\overline{6} &-0.03328\overline{4} \\
\noalign{\medskip}
-0.03328\overline{4}& 0.082374\overline{3}\end{bmatrix},\quad
\mathbf A_2=\begin{bmatrix}
1.0175\overline{4} &0.96376\overline{9} \\
\noalign{\medskip}
0.96376\overline{9}& 0.9128\overline{4}\end{bmatrix}.\end{equation}
Using \eqref{9.56} and the second
equalities of \eqref{9.43} and \eqref{9.44}, from 
\eqref{3.9} we obtain
\begin{equation}
\label{9.57}
\mathbf B_1=
\begin{bmatrix}
0.873\overline{1} &0.31406\overline{5} \\
\noalign{\medskip}
0.31406\overline{5}& 0.22272\overline{3}\end{bmatrix},\quad
\mathbf B_2=
\begin{bmatrix}
1.4904\overline{2} &0.46450\overline{5} \\
\noalign{\medskip}
0.46450\overline{5}& 1.4399\overline{6}\end{bmatrix}.\end{equation}
We observe that
$\mathbf B_1$ and $\mathbf B_2$ are hermitian. They are also positive and this can be checked
by evaluating their eigenvalues and confirming that those eigenvalues
are all positive. We recall that we mean a positive definite matrix when we say that a matrix is positive.
The eigenvalues of $\mathbf B_1$ are given by
$1$ and $0.095822\overline{9},$
and the eigenvalues of $\mathbf B_2$ are given by
$1$ and $1.9303\overline{8}.$
Since $\mathbf B_1$ and $\mathbf B_2$ are positive matrices, we can uniquely determine
$\mathbf B_1^{1/2}$ and $\mathbf B_2^{1/2}$ as positive matrices so that
\eqref{3.10} is satisfied. We get
\begin{equation*}
\mathbf B_1^{1/2}=\begin{bmatrix}
0.90309\overline{6} &0.23982\overline{6} \\
\noalign{\medskip}
0.23982\overline{6}& 0.40645\overline{6}\end{bmatrix},\quad
\mathbf B_2^{1/2}=\begin{bmatrix}
1.2052\overline{5} &0.19440\overline{4} \\
\noalign{\medskip}
0.19440\overline{4}& 1.1841\overline{3}\end{bmatrix},
\end{equation*}
which are compatible with \eqref{3.10} and \eqref{9.57}.
Because $\mathbf B_1^{1/2}$ and $\mathbf B_2^{1/2}$ are positive, they are
invertible, and their inverses are given by
\begin{equation}
\label{9.59}
\mathbf B_1^{-1/2}=\begin{bmatrix}
1.3130\overline{4} &-0.7747\overline{5} \\
\noalign{\medskip}
-0.7747\overline{5}& 2.9174\overline{2}\end{bmatrix},\quad
\mathbf B_2^{-1/2}=\begin{bmatrix}
0.85227\overline{2} &-0.13992\overline{1} \\
\noalign{\medskip}
-0.13992\overline{1}& 0.86747\overline{3}\end{bmatrix}.
\end{equation}
One can directly verify that each of
$\mathbf B_1,$ $\mathbf B_1^{1/2},$
$\mathbf B_1^{-1/2}$ is hermitian and
commutes with $P_1.$ Similarly, each of
 $\mathbf B_2,$ $\mathbf B_2^{1/2},$
$\mathbf B_2^{-1/2}$ is hermitian and
commutes with $P_2.$
One can also verify that
$\mathbf B_1$ and $\mathbf B_2$ do not commute, as expected in general, and in fact we have
\begin{equation*}
\mathbf B_1\,\mathbf B_2-\mathbf B_2\,\mathbf B_1=
 \begin{bmatrix}
0 & 0.28625\overline{5} \\
\noalign{\medskip}
-0.28625\overline{5} & 0\end{bmatrix}.\end{equation*}
Using the second equality in \eqref{9.43}, the second equality in \eqref{9.44}, and \eqref{9.59} in \eqref{3.11},
 we obtain the Marchenko normalization matrices
$M_1$ and $M_2$ as
\begin{equation}
\label{9.61}
M_1= \begin{bmatrix}
0.45339\overline{3} & -1.122\overline{1} \\
\noalign{\medskip}
-1.122\overline{1} & 2.7770\overline{7}\end{bmatrix},\quad
M_2= \begin{bmatrix}
0.37939\overline{1} & 0.35934\overline{3} \\
\noalign{\medskip}
0.35934\overline{3} & 0.34035\overline{4}\end{bmatrix}.\end{equation}
Note that $M_1$ and $M_2$ are nonnegative and have their ranks equal to one as the multiplicity of
each bound state is one. This can be verified by evaluating the eigenvalues of
$M_1$ and $M_2.$ The eigenvalues of $M_1$ are given by
$0$ and $3.2304\overline{7},$
and the eigenvalues of $M_2$ are given by
$0$ and $0.71974\overline{5},$
confirming that $M_1$ and $M_2$ are nonnegative and that each have 
the rank equal to one.

\end{example}

In the next example, we illustrate the determination of the
Marchenko normalized bound-state solutions.
We continue to use the potential and the boundary matrices 
used as input to
Examples~\ref{example9.3}--\ref{example9.5}.

\begin{example}
\label{example9.6}
\normalfont

In this example, we continue to use the potential $V$ specified in \eqref{9.10} and the
boundary matrices $A$ and $B$ appearing in \eqref{9.8}.
Using \eqref{9.49}, \eqref{9.50}, and \eqref{9.61} in \eqref{3.5} we obtain
the Marchenko normalized bound-state matrix solutions
$\Psi_1(x)$ and $\Psi_2(x)$ as
\begin{equation}
\label{9.62}
\Psi_1(x)=\ds\frac{e^{-\kappa_1 x}}{2+ e^{2x/3} }
\begin{bmatrix} 0.98890\overline{3} + 0.45339\overline{3}\, e^{2 x/3}  &
-2.4474\overline{3} - 1.122\overline{1} \, e^{2 x/3}\\
\noalign{\medskip}
-2.1620\overline{8} - 1.122\overline{1} \, e^{2 x/3} &
5.3509\overline{2} + 2.7770\overline{7}\, e^{2 x/3}\end{bmatrix},\end{equation}
\begin{equation}
\label{9.63}
\Psi_2(x)=\ds\frac{e^{-\kappa_2 x}}{2+ e^{2x/3} }
\begin{bmatrix} -0.32025\overline{5} + 0.37939\overline{1}\, e^{2 x/3} &
-0.30333\overline{1} + 0.35934\overline{3}\, e^{2 x/3}\\
\noalign{\medskip}
-0.36035\overline{2} + 0.35934\overline{3}\, e^{2 x/3}&
-0.3413\overline{1} + 0.34035\overline{4}\, e^{2 x/3}\end{bmatrix}.\end{equation}
 From \eqref{9.62} and \eqref{9.63}, we observe that
\begin{equation*}
\Psi_1(x)=
e^{-\kappa_1 x}\left(
\begin{bmatrix}  0.45339\overline{3}  &
 - 1.122\overline{1}\\
\noalign{\medskip}
- 1.122\overline{1}  &
 2.7770\overline{7}\end{bmatrix}+o(1)\right),\qquad x\to+\infty,\end{equation*}
\begin{equation*}
 \Psi_2(x)=e^{-\kappa_2 x}\left(
\begin{bmatrix} 0.37939\overline{1}  &
 0.35934\overline{3}\\
\noalign{\medskip}
 0.35934\overline{3} &
0.34035\overline{4}\end{bmatrix}+o(1)\right),\qquad x\to+\infty.\end{equation*}
From \eqref{9.62} and \eqref{9.63} we obtain
\begin{equation}
\label{9.66}
\Psi_1(0)=
\begin{bmatrix} 0.48076\overline{5}  &
-1.1898\overline{4}\\
\noalign{\medskip}
-1.0947\overline{3} &
2.7093\overline{3}\end{bmatrix},
\quad
\Psi_1'(0)=
\begin{bmatrix} -2.4558\overline{4}  &
6.0779\overline{5}\\
\noalign{\medskip}
5.5721\overline{5}  &
-13.790\overline{5}\end{bmatrix},
\end{equation}
\begin{equation}
\label{9.67}
\Psi_2(0)=
\begin{bmatrix}  0.019712\overline{1}  &
0.018670\overline{4}\\
\noalign{\medskip}-0.00033649\overline{4} &
 -0.00031871\overline{2}\end{bmatrix},
\quad
\Psi_2'(0)=
\begin{bmatrix}  0.077502\overline{5} &
0.073406\overline{9}\\
\noalign{\medskip}
0.079970\overline{1} &
0.075744\overline{2}\end{bmatrix}.
\end{equation}
Using \eqref{9.8} and \eqref{9.66} in \eqref{2.7}, we verify that
$\Psi_1(x)$ satisfies the boundary condition \eqref{2.7}.
Similarly, using \eqref{9.8} and \eqref{9.67} in \eqref{2.7} we confirm that
$\Psi_2(x)$ satisfies the boundary condition \eqref{2.7}.
Using $P_1$ given in the second equality in \eqref{9.43}, $P_2$ given in the second equality of \eqref{9.44}, 
$\Psi_1(x)$ in \eqref{9.62}, and $\Psi_2(x)$ in \eqref{9.63}, we confirm the orthonormality relations
stated in \eqref{3.6} and \eqref{3.7}, i.e. we verify that
\begin{equation*}
\ds\int_0^\infty dx\,\Psi_1(x)^\dagger\,\Psi_1(x)=P_1,\quad
\ds\int_0^\infty dx\,\Psi_2(x)^\dagger\,\Psi_2(x)=P_2,
\end{equation*}
\begin{equation}
\label{9.69}
\ds\int_0^\infty dx\,\Psi_1(x)^\dagger\,\Psi_2(x)=0,\quad
\ds\int_0^\infty dx\,\Psi_2(x)^\dagger\,\Psi_1(x)=0,
\end{equation}
where we use $0$ to denote the zero matrix on the right-hand sides in \eqref{9.69}.

\end{example}

In the next example, we illustrate the determination of the
Gel'fand--Levitan normalization matrices at the bound states.
We continue to use as input the potential $V$ and
the boundary matrices $A$ and $B$ from Examples~\ref{example9.3}--\ref{example9.6}.

\begin{example}
\label{example9.7}
\normalfont
In this example, we continue to use the potential $V$ specified in \eqref{9.10}
and the boundary matrices $A$ and $B$ given in \eqref{9.8}.
Using \eqref{2.10} and \eqref{2.15}, we see that
\begin{equation}
\label{9.70}
\varphi(k,x)=\ds\frac{e^{ikx}}{2ik}\left[I+o(1)\right]J(-k)-
\ds\frac{e^{-ikx}}{2ik}\left[I+o(1)\right] J(k),\qquad x\to+\infty,\quad k\in\mathbb C.
\end{equation}
We remark that \eqref{9.70} would normally hold only for $k\in\mathbb R$
because in general the common domain of
$J(k)$ and $J(-k)$ would be $k\in\mathbb R.$
In this example, \eqref{9.70} 
holds for $k\in\mathbb C$ because
the corresponding Jost solution $f(k,x)$ appearing in \eqref{9.11} and Jost matrix $J(k)$
appearing in \eqref{9.14} each have an extension to
$k\in\mathbb C.$
Using \eqref{9.70} we evaluate the large $x$-asymptotics of the regular solution at the bound states
with $k=i\kappa_1$ and $k=i\kappa_2$ as
\begin{equation*}
\varphi(i\kappa_1,x)=\ds\frac{e^{\kappa_1 x}}{2\kappa_1}\left[I+o(1)\right]J(i\kappa_1)
,\qquad x\to+\infty.
\end{equation*}
\begin{equation*}
\varphi(i\kappa_2,x)=
\ds\frac{e^{\kappa_2 x}}{2\kappa_2}\left[I+o(1)\right]J(i\kappa_2)
,\qquad x\to+\infty,\end{equation*}
and hence we observe that $\varphi(i\kappa_1,x)$ blows up exponentially as $x\to+\infty$ unless $J(i\kappa_1)=0$ and that
$\varphi(i\kappa_2,x)$ blows up exponentially as $x\to+\infty$ unless $J(i\kappa_2)=0.$
We note that we would have $J(i\kappa_1)=0$ in this example only when the kernel of $J(i\kappa_1)$ had the dimension $2.$
However, in this example, the bound state at $k=i\kappa_1$ is simple and hence the kernel of $J(i\kappa_1)$ has the dimension $1.$
Similarly, because the bound state at $k=i\kappa_2$ is simple, it follows that the kernel of $J(i\kappa_2)$ has the dimension $1$
and hence $J(i\kappa_2)\ne 0$ in this example.
Thus, each of $\varphi(i\kappa_1,x)$ and $\varphi(i\kappa_2,x)$ 
is exponentially growing as $x\to+\infty.$
Consequently, the 
integrals $I_1(b)$ and $I_2(b)$ defined as
\begin{equation*}
I_1(b):=\int_0^b dx\,\varphi(i\kappa_1,x)^\dagger\,\varphi(i\kappa_1,x),\quad
I_2(b):=\int_0^b dx\,\varphi(i\kappa_2,x)^\dagger\,\varphi(i\kappa_2,x),
\end{equation*}
do not converge as $b\to+\infty.$
We can illustrate this numerically by observing that we have
\begin{equation*}
I_1(1)=\begin{bmatrix} 2417.8\overline{6}& 8021.1\overline{2}\\
\noalign{\medskip}
8021.1\overline{2} & 26610.\overline{7}\end{bmatrix},\quad
I_2(1)=\begin{bmatrix}
14.975\overline{2} & 44.034\overline{3} \\
44.034\overline{3} & 176.1\overline{8}\end{bmatrix},\end{equation*}
\begin{equation*}
I_1(50)=\begin{bmatrix}
1.5916\overline{4}\times 10^{220} & 4.5391\overline{9}\times 10^{220}\\
\noalign{\medskip} 4.5391\overline{9}\times 10^{220} & 1.7517\overline{3}\times 10^{221}\end{bmatrix},\end{equation*}
\begin{equation*}
I_2(50)=\begin{bmatrix}
11391\overline{1} & 2.0007\overline{9}\times 10^{7}\\
\noalign{\medskip}  2.0007\overline{9}\times 10^{7} & 3.5556\overline{9}\times 10^{9}\end{bmatrix}.\end{equation*}
We note that we can prevent the exponential growths in
$\varphi(i\kappa_1,x)$ and $\varphi(i\kappa_2,x)$ as $x\to+\infty,$
and this can be achieved by using the matrices
$Q_1$ and $Q_2,$ respectively, appearing in the first equalities of (9.43) and (9.44).
Since $Q_1$ and $Q_2$ are the orthogonal projections onto
the kernels of $J(i\kappa_1)$ and $J(i\kappa_2),$ respectively, we have
\begin{equation}
\label{9.77}
J(i\kappa_1)\,Q_1=0,\quad
J(i\kappa_2)\,Q_2=0,\end{equation}
and hence using \eqref{9.70} with $k=i\kappa_1$ and $k=i\kappa_2,$ respectively, with the help of \eqref{9.77}
we obtain
\begin{equation*}
\varphi(i\kappa_1,x)\,Q_1=O(e^{-\kappa_1 x}),\quad
\varphi(i\kappa_2,x)\,Q_2=O(e^{-\kappa_2 x}),
\qquad x\to+\infty.\end{equation*}
Using the orthogonal projection matrices $Q_1$ and $Q_2$ from Example~\ref{example9.4}
and the regular solution $\varphi(k,x)$ constructed explicitly in Example~\ref{example9.3},
from
\eqref{3.16} we obtain
\begin{equation}
\label{9.79}
\mathbf G_1:=\int_0^\infty dx\,Q_1\,\varphi(i\kappa_1,x)^\dagger\,
\varphi(i\kappa_1,x)\,Q_1=
\begin{bmatrix}0.080478\overline{7}& -0.024258\overline{9}\\
\noalign{\medskip}
-0.024258\overline{9}&0.0073124\overline{2}\end{bmatrix},\end{equation}
\begin{equation}
\label{9.80}
\mathbf G_2:=\int_0^\infty dx\,Q_2\,\varphi(i\kappa_2,x)^\dagger\,
\varphi(i\kappa_2,x)\,Q_2=\begin{bmatrix}1326.1\overline{3}& -7.4609\overline{5} \\
\noalign{\medskip}
-7.4609\overline{5}& 0.041976\overline{2}\end{bmatrix}.\end{equation}
Using \eqref{9.79}, \eqref{9.80}, and
the first equalities of \eqref{9.43} and \eqref{9.44}, from
\eqref{3.17} we obtain
the constant positive matrices $\mathbf H_1$ and
$\mathbf H_2$ as
\begin{equation*}
\mathbf H_1=\begin{bmatrix}
0.16377\overline{2}& 0.25206\overline{6}\\
\noalign{\medskip} 
0.25206\overline{6}& 0.92401\overline{9}\end{bmatrix},
\quad
\mathbf H_2=\begin{bmatrix}
1326.1\overline{3}& -7.4553\overline{3}\\
\noalign{\medskip} 
-7.4553\overline{3}& 1.0419\overline{4}\end{bmatrix},\end{equation*}
from which we get the positive matrices $\mathbf H_1^{1/2}$ and $\mathbf H_2^{1/2}$
appearing in \eqref{3.18} as
\begin{equation*}
\mathbf H_1^{1/2}=\begin{bmatrix}
0.3549\overline{1}& 0.19445\overline{1}\\
\noalign{\medskip} 
0.19445\overline{1}& 0.94138\overline{6}\end{bmatrix},
\quad
\mathbf H_2^{1/2}=\begin{bmatrix}
36.415\overline{5}& -0.19925\overline{2}\\
\noalign{\medskip} 
-0.19925\overline{2}& 1.0011\overline{2}\end{bmatrix}.\end{equation*}
We note that $\mathbf H_1$ and $\mathbf H_2$ are hermitian and that they are positive.
Their positivity is seen from the fact that $\mathbf H_1$ has the positive eigenvalues
$1$ and $0.087791\overline{2}$ and $\mathbf H_2$ has the positive eigenvalues
$1$ and $1326.1\overline{7}.$ 
We also observe that  $\mathbf H_1^{1/2}$ and $\mathbf H_2^{1/2}$ are hermitian, and they are positive.
Their positivity is also confirmed by the fact that $\mathbf H_1^{1/2}$ has the positive eigenvalues
$1$ and $0.29629\overline{6}$ and $\mathbf H_2^{1/2}$ has the positive eigenvalues
$1$ and $36.416\overline{6}.$ 
Since $\mathbf H_1^{1/2}$ and $\mathbf H_2^{1/2}$ are positive, they are invertible and their inverses are given by
\begin{equation}
\label{9.83}
\mathbf H_1^{-1/2}=\begin{bmatrix}
3.1771\overline{8}& -0.65627\overline{4}\\
\noalign{\medskip} 
-0.65627\overline{4}& 1.1978\overline{2}\end{bmatrix},
\quad
\mathbf H_2^{-1/2}=\begin{bmatrix}
0.027490\overline{8}&0.0054714\overline{5}\\
\noalign{\medskip} 
0.0054714\overline{5}& 0.99996\overline{9}\end{bmatrix}.\end{equation}
We directly verify that each of $\mathbf H_1,$ $\mathbf H_1^{1/2},$ and $\mathbf H_1^{-1/2}$ is hermitian and commutes with 
the orthogonal projection matrix $Q_1.$
Similarly, we directly verify that each of $\mathbf H_2,$ $\mathbf H_2^{1/2},$ and $\mathbf H_2^{-1/2}$ is hermitian and commutes with $Q_2.$
With the help of $Q_1$ given in the first equality of \eqref{9.43}, $Q_2$ appearing in the first equality of \eqref{9.44}, and the
two matrices $\mathbf H_1^{-1/2}$ and $\mathbf H_2^{-1/2}$ appearing in \eqref{9.83}, from \eqref{3.19} we obtain the
Gel'fand--Levitan normalization matrices $C_1$ and $C_2$ as
\begin{equation}
\label{9.84}
C_1=\begin{bmatrix}
3.0938\overline{9}& -0.93259\overline{9}\\
\noalign{\medskip} 
-0.93259\overline{9}& 0.28111\overline{5}\end{bmatrix},
\quad
C_2=\begin{bmatrix}
0.027459\overline{1}&-0.00015448\overline{8}\\
\noalign{\medskip} 
-0.00015448\overline{8}& 8.6916\overline{8}\times10^{-7}\end{bmatrix}.\end{equation}
From \eqref{9.84} we observe that $C_1$ and $C_2$ are hermitian, 
the eigenvalues of $C_1$ are $0$ and $3.3750\overline{1},$ and the
eigenvalues of $C_2$ are
$0$ and $0.0274\overline{6}.$
Thus, we confirm that $C_1$ and $C_2$ each are nonnegative and have the rank equal to one.

\end{example}

In the next example, we illustrate the determination of the
Gel'fand--Levitan normalized bound-state solutions.
We continue to use as input the same potential $V$ and the same boundary matrices
$A$ and $B$ we have used in Examples~\ref{example9.3}--\ref{example9.7}.

\begin{example}
\label{example9.8}
\normalfont

In this example, we still use the potential $V$ given in \eqref{9.10} and
the boundary matrices $A$ and $B$ given in \eqref{9.8}.
Using the explicit expression for the regular solution
obtained in Example~\ref{example9.3},
the bound-state $k$-values listed in \eqref{9.35} and \eqref{9.36}, and
the Gel'fand--Levitan normalization matrices $C_1$ and $C_2$ listed in
\eqref{9.84}, with the help of 
\eqref{3.12} we obtain the
Gel'fand--Levitan normalized bound-state solutions
$\Phi_1(x)$ and $\Phi_2(x)$ as
\begin{equation}
\label{9.85}
\Phi_1(x)=\ds\frac{e^{-\kappa_1 x}}{2+ e^{2x/3} }
\begin{bmatrix} 2.5273\overline{4} + 1.1587\overline{4}\, e^{2 x/3}  &
-0.76182\overline{3}-0.34928\overline{1}\, e^{2 x/3}\\
\noalign{\medskip}
-5.5256\overline{4}-2.8677\overline{5}\, e^{2 x/3} &
1.6656\overline{1}+0.86443\overline{3}\, e^{2 x/3}\end{bmatrix},\end{equation}
\begin{equation}
\label{9.86}
\Phi_2(x)=\ds\frac{e^{-\kappa_2 x}}{2+ e^{2x/3} }
\begin{bmatrix} -0.44109\overline{7} + 0.52254\overline{8}\, e^{2 x/3} &
0.0024816\overline{6} - 0.0029399\overline{1}\, e^{2 x/3}\\
\noalign{\medskip}
-0.49632\overline{5} + 0.49493\overline{4}\, e^{2 x/3}&
0.0027923\overline{8} - 0.0027845\overline{6}\, e^{2 x/3}\end{bmatrix}.\end{equation}
From \eqref{9.85} and \eqref{9.86} we obtain the large $x$-asymptotics 
\begin{equation*}
\Phi_1(x)=
e^{-\kappa_1 x}\left(
\begin{bmatrix} 1.1587\overline{4}  &
-0.34928\overline{1}\\
\noalign{\medskip}
-2.8677\overline{5}  &
0.86443\overline{3}\end{bmatrix}+o(1)\right),\qquad x\to+\infty,\end{equation*}
\begin{equation*}
 \Phi_2(x)=e^{-\kappa_2 x}\left(
\begin{bmatrix} 0.52254\overline{8} &
- 0.0029399\overline{1}\\
\noalign{\medskip}
 0.49493\overline{4} &
 - 0.0027845\overline{6}\end{bmatrix}+o(1)\right),\qquad x\to+\infty.\end{equation*}
 Using $\Phi_1(x)$ given in \eqref{9.85}, $\Phi_2(x)$ given in \eqref{9.86}, and 
 the orthogonal projection matrices $Q_1$ and $Q_2$ given in 
 the first equalities of
 \eqref{9.43} and \eqref{9.44}, respectively, we confirm that the
 orthonormalizations
 stated in \eqref{3.14} and \eqref{3.15} indeed hold
 for the bound states at $k=i\kappa_1$ and $k=i\kappa_2.$
From \eqref{9.85} and \eqref{9.86} we obtain
\begin{equation}
\label{9.89}
\Phi_1(0)=
\begin{bmatrix}1.2286\overline{9}  &
-0.37036\overline{8}\\
\noalign{\medskip}
-2.797\overline{8}&
0.84334\overline{6}\end{bmatrix},
\quad
\Phi_1'(0)=
\begin{bmatrix}-6.2764\overline{1}  &
1.8919\overline{1}\\
\noalign{\medskip}
14.240\overline{8}  &
-4.2926\overline{3}\end{bmatrix},
\end{equation}
\begin{equation}
\label{9.90}
\Phi_2(0)=
\begin{bmatrix} 0.027150\overline{1} &
-0.0001527\overline{5}\\
\noalign{\medskip}-0.00046346\overline{4} &
2.607\overline{5}\times 10^{-6}\end{bmatrix},
\quad
\Phi_2'(0)=
\begin{bmatrix}  0.10674\overline{7} &
-0.00060056\overline{9}\\
\noalign{\medskip}
0.11014\overline{5} &
-0.00061969\overline{1}\end{bmatrix}.
\end{equation}
Using the boundary matrices $A$ and $B$ appearing in \eqref{9.8} and also using the expressions
\eqref{9.89} and \eqref{9.90}, we verify that each of the Gel'fand--Levitan
normalized bound-state solutions $\Phi_1(x)$ and $\Phi_2(x)$ satisfies the boundary condition
\eqref{2.7}.
Having obtained the Marchenko normalized bound-state solutions
$\Psi_1(x)$ and $\Psi_2(x)$ and
the Gel'fand--Levitan normalized bound-state solutions
$\Phi_1(x)$ and $\Phi_2(x),$ we next determine the dependency matrices $D_1$ and $D_2.$ 
Using $\Psi_1(0)$ and $\Psi_2(0)$ listed in the first equalities
of \eqref{9.66} and \eqref{9.67},
respectively, we obtain
the Moore--Penrose inverses
$\Psi_1(0)^+$ and $\Psi_2(0)^+$
as
\begin{equation}
\label{9.91}
[\Psi_1(0)]^+=
\begin{bmatrix}0.0.047199\overline{7}  &
-0.10747\overline{6}\\
\noalign{\medskip}
-0.11681\overline{4} &
0.26599\overline{2}\end{bmatrix},\quad
[\Psi_2(0)]^+=
\begin{bmatrix}26.733\overline{1}  &
-0.45634\overline{5}\\
\noalign{\medskip}
25.320\overline{4} &
-0.4322\overline{3}\end{bmatrix}.
\end{equation}
Using $x=0$ in \eqref{3.24}, we get
\begin{equation}
\label{9.92}
D_1=[\Psi_1(0)]^+\Phi_1(0),\quad D_2=
[\Psi_2(0)]^+\Phi_2(0).
\end{equation}
Using $\Phi_1(0)$ from the first equality of \eqref{9.89}, $\Phi_2(0)$ from the first equality of
\eqref{9.90}, and
$\Psi_1(0)^+$ and $\Psi_2(0)^+$ from 
 \eqref{9.91}, with the help of \eqref{9.92}
we obtain the dependency matrix $D_1$ at the bound state $k=i\kappa_1$ and
the dependency matrix $D_2$ at the bound state $k=i\kappa_2$ as
\begin{equation}
\label{9.93}
D_1=\begin{bmatrix}0.3586\overline{9} &
-0.10812\overline{1}\\
\noalign{\medskip}
-0.88772\overline{1}&
0.26758\overline{8}\end{bmatrix},
\quad
D_2=
\begin{bmatrix}0.72601\overline{8}  &
-0.0040846\overline{6}\\
\noalign{\medskip}
0.68765\overline{2} &
-0.0038688\overline{1}\end{bmatrix}.
\end{equation}
Using $D_1$ and $D_2$ appearing in \eqref{9.93}, we verify the properties listed in \eqref{3.21} and
\eqref{3.27}--\eqref{3.30}.

\end{example}

Consider the orthogonal projection matrix $P$ satisfying the first set of equalities in \eqref{3.2}
in the special case when $n=2.$ The rank of $P$ can only be $2,$ $1,$ or $0.$ If the rank of $P$ is equal
to $2,$ then $P$ must be the $2\times 2$ identity matrix. This can be seen as follows. Since $P$ is invertible,
by multiplying $P(P - I) = 0$ on the left by $P^{-1}$ we obtain $P-I=0$ or equivalently $P = I.$ If the rank of $P$ is zero, then $P$ must be the $2\times 2$ zero matrix.
 In the next theorem we characterize all $2\times 2$ orthogonal projection matrices with rank $1.$ The result is useful in the determination of
 the projection matrices $P$ used in \eqref{7.20} and \eqref{8.23}
 in the transformation of the Jost matrix in adding a new bound state or in increasing the multiplicity of an existing bound state.
 Even though this result is not necessarily new, the elementary proof we provide yields some insight
 on the computation of such projection matrices not only in the $2\times 2$ matrix case
 but also
 when the matrix size is larger.

\begin{theorem}
\label{theorem9.9} Any $2\times 2$ rank-one projection matrix $P$ is given by
\begin{equation}
\label{9.94}
P=
\begin{bmatrix} \ds\frac{1}{2}\pm\ds\sqrt{\ds\frac{1}{4}-\beta^2-\gamma^2}&\beta+i\,\gamma\\
\beta-i\,\gamma& \ds\frac{1}{2}\mp\ds\sqrt{\ds\frac{1}{4}-\beta^2-\gamma^2}\end{bmatrix},
 \end{equation}
 where $\beta$ and $\gamma$ are two real constants satisfying $0\le \beta^2+\gamma^2\le 1/4.$
\end{theorem}

\begin{proof}
Consider the orthogonal projection matrix $P$ satisfying the first set of equalities in \eqref{3.2} in the special case when $n=2.$
The rank of $P$ can be either $1$ or $2.$ If the rank is equal to $2,$ then $P$ must be the $2\times 2$ identity matrix. 
This is seen as follows. Since $P$ is invertible, by multiplying
$P(P-I)=0$ on the left by $P^{-1},$ we obtain $P-I=0$ or equivalently $P=I.$
Let us now concentrate on the case where the rank of $P$ is equal to $1.$
Since
$P$ is selfadjoint, it must have the form
\begin{equation}
\label{9.95}
P=
\begin{bmatrix} \alpha&\beta+i\gamma\\
\beta-i\gamma&\epsilon\end{bmatrix},
 \end{equation}
for some real constants $\alpha,$ $\beta,$ $\gamma,$ and $\epsilon.$ The idempotent property
$P^2=P$ yields
\begin{equation*}
\begin{bmatrix} \alpha^2+\beta^2+\gamma^2&(\alpha+\epsilon)(\beta+i\gamma)\\
(\alpha+\epsilon)(\beta-i\gamma)&\epsilon^2+\beta^2+\gamma^2\end{bmatrix}=
\begin{bmatrix} \alpha&\beta+i\gamma\\
\beta-i\gamma&\epsilon\end{bmatrix},
 \end{equation*}
 which is equivalent to
 \begin{equation}
\label{9.97}
\begin{bmatrix} \alpha^2-\alpha+\beta^2+\gamma^2&(\alpha+\epsilon-1)(\beta+i\gamma)\\
(\alpha+\epsilon-1)(\beta-i\gamma)&\epsilon^2-\epsilon+\beta^2+\gamma^2\end{bmatrix}=
\begin{bmatrix} 0&0\\
0&0\end{bmatrix}.
 \end{equation}
 The inspection of the $(1,2)$-entries in \eqref{9.97} shows that
 we must have either $\beta+i\gamma=0$ or $\alpha+\epsilon=1.$
In the former case, we must have $\beta=0,$ $\gamma=0,$ and the real parameters $\alpha$ and $\epsilon$ can only take the values of $0$ or $1.$
This case yields the two possibilities for $P$ given by
\begin{equation}
\label{9.98}
P=
\begin{bmatrix}1&0\\
0&0\end{bmatrix},\quad P=
\begin{bmatrix}0&0\\
0&1\end{bmatrix}.
 \end{equation}
In the latter case, i.e. when we have
 $\alpha+\epsilon=1,$ the four constraints in \eqref{9.97} can equivalently be expressed as the single
 constraint given by
  \begin{equation}
\label{9.99}
\alpha^2-\alpha+|p|^2=0,
 \end{equation}
 where we have let $p:=\beta+i\,\gamma.$
 Solving \eqref{9.99}, we obtain the two possible values for $\alpha$ given by
  \begin{equation}
\label{9.100}
\alpha=\ds\frac{1}{2}\pm\ds\sqrt{\ds\frac{1}{4}-|p|^2},
 \end{equation}
 with the constraint
$0< 4|p|^2\le 1.$
Using \eqref{9.100} and  $\alpha+\epsilon=1$ in \eqref{9.95}, we observe that $P$ has the two possible
forms given by
\begin{equation}
\label{9.101}
P=
\begin{bmatrix} \ds\frac{1}{2}\pm\ds\sqrt{\ds\frac{1}{4}-|p|^2}&p\\
p^\ast& \ds\frac{1}{2}\mp\ds\sqrt{\ds\frac{1}{4}-|p|^2}\end{bmatrix},
 \end{equation}
The determinant of either of the two matrices appearing on the right-hand side of \eqref{9.101} is zero, and hence
those two matrices each have rank equal to one.
The four possibilities for $P$ given in
\eqref{9.98} and \eqref{9.101} cover all the possibilities
for any $2\times 2$ rank-one projection matrix $P.$
Thus, the proof is complete. 
\end{proof}

In the next theorem, in the $2\times 2$ matrix case we show how the
projection matrices $\hat P$ and $\hat{\hat P}$ 
used to add a new bound state and to increase the multiplicity of that bound state, respectively, are related to each other.
This result is useful in practical computations, and it provides some insight
into the process of increasing the multiplicity of a bound state in general.
The notation used in the theorem is similar to the notation used in Theorem~\ref{theorem8.1}.
 
\begin{theorem}
\label{theorem9.10} Let $\hat P$ be an orthogonal projection matrix used
in \eqref{7.20} in the $2\times 2$ case in order to add a new bound-state
with the multiplicity $1$ at $k=i\tilde\kappa.$ Let 
$\hat{\hat P}$ be an orthogonal projection matrix used
in \eqref{8.23} in the $2\times 2$ case in order to increase the multiplicity of
the bound-state at $k=i\tilde\kappa$ from $1$ to $2.$
Then, $\hat P$ must be of the form given on the right-hand side of \eqref{9.94} and we must have
$\hat{\hat P}=I-\hat P.$
Consequently, we also have 
$\hat{\hat P}\,\hat P=0$ and $\hat P\,\hat{\hat P}=0.$
\end{theorem}

\begin{proof}

Let us suppose that the unperturbed operator does not have a bound state at $k=i\tilde\kappa$ and let
$J(k)$ be the corresponding unperturbed $2\times 2$ Jost matrix. Let us add a bound state
at $k=i\tilde\kappa$ with multiplicity one, and let us use $\hat J(k)$ to denote the corresponding
Jost matrix for the perturbed operator. From \eqref{7.20} we have
\begin{equation}
\label{9.102}
\hat J(k)=\left[I-\ds\frac{2i\tilde\kappa}{k+i\tilde\kappa}\,\hat P\right]J(k),
 \end{equation}
 where $I$ is the $2\times 2$ identity matrix.
We remark that the matrix $\hat P$ appearing in \eqref{9.102} must be equal to one of the two matrices given in \eqref{9.94}
for some appropriate $\beta$ and $\gamma.$ Let us also 
increase the multiplicity of the bound state at 
$k=i\tilde\kappa$ from $1$ to $2$
by using the projection matrix $\hat{\hat P}.$ As seen from \eqref{8.23}, we then end up with the
further perturbed Jost matrix $\hat{\hat J}(k),$ which is related to $\hat J(k)$ as
\begin{equation}
\label{9.103}
\hat{\hat J}(k)=\left[I-\ds\frac{2i\tilde\kappa}{k+i\tilde\kappa}\,\hat{\hat P}\right]\hat J(k),
 \end{equation}
 where the matrix $\hat{\hat P}$ should have one of the two forms appearing on the right-hand side of \eqref{9.101}
 with some appropriate $\beta$ and $\gamma.$
From \eqref{9.102} and \eqref{9.103} it follows that
\begin{equation}
\label{9.104}
\hat{\hat J}(k)=\left[I-\ds\frac{2i\tilde\kappa}{k+i\tilde\kappa}\,\hat{\hat P}\right]\left[I-\ds\frac{2i\tilde\kappa}{k+i\tilde\kappa}\,\hat P\right]J(k),
 \end{equation}
which can be written as
\begin{equation}
\label{9.105}
\hat{\hat J}(k)=\left[I-\ds\frac{2i\tilde\kappa}{k+i\tilde\kappa}\left(\hat{\hat P}+\hat P\right)+\left(\ds\frac{2i\tilde\kappa}{k+i\tilde\kappa}\right)^2\hat{\hat P}\,\hat P\right]J(k).
 \end{equation}
If the bound state at $k=i\tilde\kappa$ with multiplicity $2$ is added to the unperturbed problem, from \eqref{7.20} with
the projection matrix equal to the $2\times 2$ identity matrix, we get
\begin{equation}
\label{9.106}
\hat{\hat J}(k)=\left[I-\ds\frac{2i\tilde\kappa}{k+i\tilde\kappa}\,I\right]J(k),
 \end{equation}
 where we have also used Theorem~\ref{theorem8.1} indicating that the addition of the bound state in one step or in two steps yields
 the same result.
Subtracting \eqref{9.106} from \eqref{9.105} we obtain
\begin{equation}
\label{9.107}
\left[\ds\frac{2i\tilde\kappa}{k+i\tilde\kappa}\,I-\ds\frac{2i\tilde\kappa}{k+i\tilde\kappa}\left(\hat{\hat P}+\hat P\right)+\left(\ds\frac{2i\tilde\kappa}{k+i\tilde\kappa}\right)^2\hat{\hat P}\,\hat P\right]J(k)=\begin{bmatrix}0&0\\
0&0\end{bmatrix}.
 \end{equation}
The Jost matrix $J(k)$ is invertible for $k\in\bCpb\setminus\{0\}$ because it corresponds to the case with no bound states. 
Multiplying \eqref{9.107} on the right by the inverse of $J(k),$ we obtain
\begin{equation}
\label{9.108}
\ds\frac{2i\tilde\kappa}{k+i\tilde\kappa}\,I-\ds\frac{2i\tilde\kappa}{k+i\tilde\kappa}\left(\hat{\hat P}+\hat P\right)+\left(\ds\frac{2i\tilde\kappa}{k+i\tilde\kappa}\right)^2\hat{\hat P}\,\hat P=\begin{bmatrix}0&0\\
0&0\end{bmatrix},\qquad
k\in\bCpb\setminus\{0\},
 \end{equation}
 from which we get
 $\hat{\hat P}+\hat P=I$ and 
  $\hat{\hat P}\,\hat P=0.$
 Thus, we have $\hat{\hat P}=I-\hat P.$ 
 Multiplying the matrix equality $\hat{\hat P}=I-\hat P$ with the matrix $\hat P$ on the right and using $\hat P^2=\hat P,$ we also get
 $\hat P\,\hat{\hat P}=0.$
\end{proof}

The matrix equalities stated in Theorem~\ref{theorem9.10}, i.e. 
\begin{equation}
\label{9.109}
\hat{\hat P}_j+\hat P_j=\begin{bmatrix}1&0\\
0&1\end{bmatrix},\quad 
\hat{\hat P}_j\,\hat P_j=\begin{bmatrix}0&0\\
0&0\end{bmatrix},
 \end{equation}
 are useful in the addition of a bound state at $k=i\tilde\kappa_j$ and increasing the multiplicity of a bound state in the 
 $2\times 2$ matrix case as follows.
It is enough to construct $\hat J(k),$ and this amounts to the determination of
 the two real constants $\beta$ and $\gamma$ and picking the correct sign in the diagonal entry on the right-hand side of
 \eqref{9.94}. That task is accomplished by requiring that the orthogonal projection matrix $\hat Q_j$ appearing
 in the second set of equalities in \eqref{3.2} is the projection onto the kernel of $\hat J(i\tilde\kappa_j).$ 
 Note that $\hat Q_j$ is the orthogonal projection matrix onto 
 the orthogonal complement of the kernel of the Gel'fand--Levitan normalization matrix
$\hat C_j$ appearing in \eqref{3.19}. Hence, 
$\hat Q_j$ can be determined from $\hat C_j.$
 from the Gel'fand--Levitan normalization matrix $\hat C_j$ appearing in \eqref{3.19}.
 In order to determine the exact form of $\hat P_j$ in
 \eqref{9.102} we simply impose the condition that $\hat J(i\tilde\kappa_j) \,\hat Q_j=0,$
 which allows us to pick the correct complex constant and the correct sign on the right-hand side in \eqref{9.101}.
The reasoning and conclusion presented in \eqref{9.102}--\eqref{9.109} in the $2\times 2$ matrix case can readily be generalized 
to the case of adding or removing a bound state with any multiplicities.
For example, we have the following results. If $\hat P$ is the $n\times n$ orthogonal projection matrix
used in \eqref{9.102} to add a new bound state at $k=i\tilde\kappa$ with multiplicity $\tilde m$
and if $\hat{\hat P}$ is the $n\times n$ orthogonal projection matrix
used in \eqref{9.103} to increase the multiplicity of the bound state at $k=i\tilde\kappa$ from $\tilde m$ to $n,$ then we get
\begin{equation*}
\hat{\hat P}=I-\hat P,\quad
\hat{\hat P}\,\hat P=0,\quad 
\hat P\,\hat{\hat P}=0,
 \end{equation*}
where we assume that $1\le \tilde m<n.$

In the next example we deal with a selfadjoint potential that belongs to $L^1(\mathbb R^+)$ but not $L_1^1(\mathbb R^+).$
For such a potential we show that the Jost solution $f(k,x),$ the physical solution $\Psi(k,x),$ the Jost matrix $J(k),$
and the scattering matrix $S(k)$ may not have continuous extensions to $k=0.$

\begin{example}
\label{example9.11} 
\normalfont
The Whittaker equation \cite{AS1977} 
\begin{equation}
\label{9.111}\ds\frac{d^2 w}{dz^2}+\left(-\ds\frac{1}{4}+\ds\frac{\kappa}{z}+\ds\frac{1/4-\mu^2}{z^2}\right)w=0,
 \end{equation}
has the two linearly independent solutions
$M_{\kappa,\mu}(z)$ and $W_{\kappa,\mu}(z),$ which are known as the Whittaker functions.
Using $\kappa=0$ and $z=-2ik(1+x)$ in \eqref{9.111}, we obtain the 
Schr\"odinger equation
\begin{equation}
\label{9.112}-\ds\frac{d^2 w}{dx^2}+V_\mu(x)\,w=k^2\,w,\qquad x\in\mathbb R^+,
\end{equation}
 where the potential $V_\mu(x)$ is given by
 \begin{equation}
\label{9.113}
V_\mu(x):=\ds\frac{\mu^2-1/4}{(1+x)^2}.
\end{equation}
 We can express $M_{0,\mu}(-2ik(1+x))$ and $W_{0,\mu}(-2ik(1+z))$
 in terms of Kummer's confluent hypergeometric
 function $M(a,b,z)$ and Tricomi's confluent hypergeometric
 function $U(a,b,z),$ respectively, as
 \begin{equation}
\label{9.114}
M_{0,\mu}(-2ik(1+x))=e^{ik(1+x)}\left[-2ik(1+x)\right]^{\mu+1/2} M(\mu+1/2,1+2\mu,-2ik(1+x)),
\end{equation}
 \begin{equation}
\label{9.115}
W_{0,\mu}(-2ik(1+x))=e^{ik(1+x)}\left[-2ik(1+x)\right]^{\mu+1/2} U(\mu+1/2,1+2\mu,-2ik(1+x)).
\end{equation}
The general solution to \eqref{9.112} is a linear combination of the two
functions appearing on the right-hand sides of \eqref{9.114} and
\eqref{9.115}, respectively. The Jost solution $f_\mu(k,x)$ satisfying \eqref{2.10} and
\eqref{9.112} is then obtained as
 \begin{equation*}
f_\mu(k,x)=e^{-ik}\,W_{0,\mu}(-2ik(1+x)),\qquad \overline{\mathbb C^+}\setminus\{0\},
\end{equation*}
or equivalently as
\begin{equation}
\label{9.117}
f_\mu(k,x)=e^{ikx}\left[-2ik(1+x)\right]^{\mu+1/2} U(\mu+1/2,1+2\mu,-2ik(1+x)),\qquad \overline{\mathbb C^+}\setminus\{0\}.
\end{equation}
The right-hand side of \eqref{9.117} can be expressed in terms of the modified Bessel function of the second kind
$K_\mu(z),$ and we have
\begin{equation}
\label{9.118}
f_\mu(k,x)=\ds\sqrt{\ds\frac{2}{\pi}}\,e^{-ik}\,\sqrt{-ik(1+x)}\,K_\mu(-ik(1+x)),\qquad \overline{\mathbb C^+}\setminus\{0\}.
\end{equation}
Let us use the Dirichlet boundary condition by choosing
the matrices $A$ and $B$ appearing in \eqref{2.7} as
$A=0$ and $B=1,$ where in the scalar case we use
a scalar to represent a $1\times 1$ matrix.
In this case, from \eqref{2.11} and \eqref{9.118} we see that
the Jost matrix $J_\mu(k)$ is equal to $f_\mu(k,0),$ i.e. we have
\begin{equation}
\label{9.119}
J_\mu(k)=\left(-2ik\right)^{\mu+1/2} U(\mu+1/2,1+2\mu,-2ik),
\end{equation}
which is equivalent to
\begin{equation}
\label{9.120}
J_\mu(k)=\ds\sqrt{\ds\frac{2}{\pi}}\,e^{-ik}\,\sqrt{-ik}\,K_\mu(-ik).\qquad \overline{\mathbb C^+}\setminus\{0\}.
\end{equation}
Using \eqref{9.120} in \eqref{2.12} and \eqref{2.13}, we obtain the corresponding
scattering matrix $S_\mu(k)$ and the physical solution $\Psi_\mu(k,x)$ as
\begin{equation}
\label{9.121}
S_\mu(k)=-\ds\frac{i\,k\,e^{2ik}\,K_\mu(ik)}{\sqrt{k^2}\,K_\mu(-ik)},\qquad \mathbb R\setminus\{0\},
\end{equation}
\begin{equation}
\label{9.122}
\Psi_\mu(k,x)=f_\mu(-k,x)+f_\mu(k,x)\, S_\mu(k),\qquad \mathbb R\setminus\{0\}.
\end{equation}
The small $k$-asymptotics of the right-hand sides of \eqref{9.118}--\eqref{9.122} are obtained
by using the small $z$-asymptotics of $K_\mu(z)$ For example, when we have
\begin{equation*}
\mu=\sqrt{\ell(\ell+1)+\ds\frac{1}{4}},
\end{equation*}
where $\ell$ is a positive integer, the corresponding potential $V_\mu(x)$ appearing in \eqref{9.113} is given by
\begin{equation}
\label{9.124}
V_\mu(x)=\ds\frac{\ell(\ell+1)}{(1+x)^2}.
\end{equation}
The Jost solution $f_\mu(k,x)$ corresponding to the potential $V_\mu(x)$ in \eqref{9.124} can be expressed in terms of the Hankel function
of the first kind $H^{(1)}_\nu(z)$ of order $\nu,$ and we have
\begin{equation}
\label{9.125}
f_\mu(k,x)=i^{\ell+1}\,\ds\sqrt{\ds\frac{\pi k}{2}(1+x)}\,H^{(1)}_{\ell+1/2}(k(1+x)),\qquad \overline{\mathbb C^+}\setminus\{0\}.
\end{equation}
The corresponding Jost matrix $J_\mu(k)$ is obtained from \eqref{9.125} as $f_\mu(k,0),$ and hence we get
\begin{equation}
\label{9.126}
J_\mu(k)=i^{\ell+1}\,\ds\sqrt{\ds\frac{\pi k}{2}}\,H^{(1)}_{\ell+1/2}(k),\qquad \overline{\mathbb C^+}\setminus\{0\}.
\end{equation}
Using \eqref{9.126} in \eqref{2.12}, we obtain the corresponding scattering matrix $S_\mu(k)$ as
\begin{equation}
\label{9.127}
S_\mu(k)=-i\, \ds\frac{H^{(1)}_{\ell+1/2}(-k)}
{H^{(1)}_{\ell+1/2}(k)},\qquad
\mathbb R\setminus\{0\}.
\end{equation}
The corresponding physical solution $\Psi_\mu(k,x)$ is obtained by using \eqref{9.125} and \eqref{9.127} on the right-hand side of \eqref{9.122}.
Using the small $k$-asymptotics of the Hankel function of the first kind on the right-hand sides of \eqref{9.125} and \eqref{9.126}, we obtain
\begin{equation}
\label{9.128}
f_\mu(k,x)=\ds\frac{1}{[k(1+x)]^\ell}
\left[ \left(\ds\frac{i}{2}\right)^\ell\,\ds\frac{(2\ell)!}{\ell!}+O(k^2)\right],\qquad k\to 0 \text{\rm{ in }}\overline{\mathbb C^+}\setminus\{0\},
\end{equation}
\begin{equation}
\label{9.129}
J_\mu(k)=\ds\frac{1}{k^\ell}
\left[ \left(\ds\frac{i}{2}\right)^\ell\,\ds\frac{(2\ell)!}{\ell!}+O(k^2)\right],\qquad k\to 0 \text{\rm{ in }}\overline{\mathbb C^+}\setminus\{0\}.
\end{equation}
From \eqref{9.128} and \eqref{9.129} we observe that the right-hand sides of \eqref{9.128} and \eqref{9.129} each have a pole
of order $\ell$ at $k=0.$ Thus, neither the zero-energy Jost solution $f_\mu(0,x)$ 
nor the zero-energy Jost matrix $J_\mu(0)$ exist.
On the other hand, using \eqref{9.129} on the right-hand side of \eqref{2.12}, we observe that 
the scattering matrix $S_\mu(k)$
given in \eqref{9.127} has a continuous extension to $k=0,$ and we
have
$S_\mu(0)=(-1)^{\ell+1}.$
The zero-energy Schr\"odinger equation with the potential $V_\mu$ in \eqref{9.124} can be written as
\begin{equation}
\label{9.130}
-\psi''+\ds\frac{\ell(\ell+1)}{(1+x)^2}\,\psi=0,\qquad x\in\mathbb R^+,
\end{equation}
and it has two linearly independent solutions given by
\begin{equation}
\label{9.131}
\psi(x)=\left(1+x\right)^{1/2+\sqrt{\ell(\ell+1)+1/4}},\quad \psi(x)=\left(1+x\right)^{1/2-\sqrt{\ell(\ell+1)+1/4}},
\qquad x\in\mathbb R^+.
\end{equation}
Any other solution to \eqref{9.130} must be a linear combination of the two solutions listed in \eqref{9.131}.
In this case, the zero-energy Jost solution $f_0(0,x)$ does not exist because we cannot construct a solution to \eqref{9.130}
with the asymptotic behavior $1+o(1)$ as $x\to+\infty.$ This is because the first solution
in \eqref{9.131} blows up and the second solution there decays as $x\to+\infty$ if $\ell$ is a positive integer.
The zero-energy regular solution $\varphi_\mu(0,x)$ can be expressed as a linear combination of the 
two solutions in \eqref{9.131} as
\begin{equation*}
\varphi_\mu(0,x)=\ds\frac{\left(1+x\right)^{1/2+\sqrt{\ell(\ell+1)+1/4}}-\left(1+x\right)^{1/2-\sqrt{\ell(\ell+1)+1/4}}}{2\,\sqrt{\ell(\ell+1)+\ds\frac{1}{4}}},
\qquad x\in\mathbb R^+.
\end{equation*}
The corresponding zero-energy physical solution $\Psi_\mu(0,x)$ does not exist in this case as the following argument shows.
Using \eqref{2.14}, we express the physical solution $\Psi_\mu(k,x)$ in terms of the regular solution $\varphi_\mu(k,x)$ and the
inverse of the Jost matrix $J_\mu(k)$ as
\begin{equation}
\label{9.133}
\Psi_\mu(k,x)=\varphi_\mu(k,x)\left[ 2ik\,J_\mu(k)^{-1}\right],
\qquad k\in\overline{\mathbb C^+}\setminus\{0\}.
\end{equation}
From \eqref{9.129} we see that
$J_\mu(k)^{-1}=O(k^\ell)$ as $k\to 0.$ Since $\varphi_\mu(k,x)$ exists for $k\in\mathbb C,$ 
we observe that the right-hand side of \eqref{9.133} has a limit as $k\to 0,$ and that limit is equal to the function
that is identically zero for all $x\ge 0.$ However, we cannot use that limit as the zero-energy physical solution
$\Psi_\mu(0,x)$ because the latter must be a nontrivial bounded solution to \eqref{9.130}.
Since we cannot construct a nontrivial bounded solution to \eqref{9.130} as a linear combination
of the two solutions listed in \eqref{9.131}, we conclude that the zero-energy physical solution $\Psi_\mu(0,x)$ in this case
does not exist. This is also compatible with the limit as $k\to 0$ for \eqref{9.122} and the fact that 
the zero-energy Jost solution $f_\mu(0,x)$ does not exist.
In general, for other $\mu$-values, the quantity
$S_\mu(k)$ does not have a continuous extension to $k=0.$ For example, when $\mu=0$ the potential $V_\mu(x)$ appearing in 
\eqref{9.113} is given by
\begin{equation}
\label{9.134}
V_0(x)=-\ds\frac{1}{4(1+x)^2},\qquad x\in\mathbb R^+.
\end{equation}
In that case, using $\mu=0$ in \eqref{9.118}--\eqref{9.121}, we obtain
the Jost solution $f_0(k,x),$ the Jost matrix $J_0(k),$ and the scattering matrix $S_0(k),$ respectively, as
\begin{equation}
\label{9.135}
f_0(k,x)=\ds\sqrt{\ds\frac{2}{\pi}}\,e^{-ik}\,\sqrt{-ik(1+x)}\,K_0(-ik(1+x)),
\qquad k\in\overline{\mathbb C^+}\setminus\{0\},
\end{equation}
\begin{equation}
\label{9.136}
J_0(k)=\ds\sqrt{\ds\frac{2}{\pi}}\,e^{-ik}\,\sqrt{-ik}\,K_0(-ik),\qquad k\in\overline{\mathbb C^+}\setminus\{0\},
\end{equation}
\begin{equation}
\label{9.137}
S_0(k)=-\ds\frac{i\,k\,e^{2ik}\,K_0(ik)}{\sqrt{k^2}\,K_0(-ik)},\qquad k\in\mathbb R\setminus\{0\}.
\end{equation}
The physical solution $\Psi_0(k,x)$ is obtained by using \eqref{9.135} and \eqref{9.137} in \eqref{2.13}, and we get
\begin{equation}
\label{9.138}
\Psi_0(k,x)=f_0(-k,x)+f_0(k,x)\, S_0(k).\qquad k\in\overline{\mathbb C^+}\setminus\{0\}.
\end{equation}
Similarly, we obtain the regular solution $\varphi_0(k,x)$  by using \eqref{9.135} and \eqref{9.136} in \eqref{2.15} and we have
\begin{equation}
\label{9.139}
\varphi_0(k,x)=\ds\frac{1}{2ik}\left[f_0(k,x)\,J_0(-k)-f_0(-k,x)\,J_0(k)\right],\qquad k\in\mathbb R\setminus\{0\}.
\end{equation}
The small $k$-asymptotics of $f_0(k,x)$ is obtained by letting $k\to 0$ on the right-hand side of \eqref{9.135}. As
$k\to 0$ in $\overline{\mathbb C^+},$ we get
\begin{equation}
\label{9.140}
f_0(k,x)=-\ds\sqrt{\ds\frac{2}{\pi}}\,\sqrt{-ik(1+x)}\left(\gamma+\log(-ik(1+x)/2)\right)+O(|k|^{3/2}),
\end{equation}
where we use $\log$ to denote the complex-valued principal logarithm and 
use $\gamma$ for the Euler--Mascheroni constant with the approximate value of $0.577.$
Similarly, from \eqref{9.136} we obtain the small $k$-asymptotics of $J_0(k)$ as
\begin{equation}
\label{9.141}
J_0(k)=-\ds\sqrt{\ds\frac{2}{\pi}}\,\sqrt{-ik}\left(\gamma+\log(-ik/2)\right)+O(|k|^{3/2}),\qquad k\to 0\text{\rm{ in }}\overline{\mathbb C^+}.
\end{equation}
From \eqref{9.141} we conclude that the Jost matrix $J_0(k),$ which is defined 
when $k\in\overline{\mathbb C^+}\setminus\{0\},$ extends continuously to $k=0$ and satisfies
$J_0(0)=0.$ On the other hand, using \eqref{9.141} in \eqref{2.12}, we see that the scattering matrix $S_0(k),$ which is defined when $k\in\mathbb R\setminus\{0\},$
has the right and left limits, obtained as $k\to 0^+$ and $k\to 0^-,$ respectively, and we have
\begin{equation}
\label{9.142}
S_0(0^+)=-i,\quad
S_0(0^-)=i.
\end{equation}
From \eqref{9.142} we conclude that $S_0(k)$ cannot be extended to $k=0$ in a continuous manner. Hence, $S_0(0)$ does not exist.
From \eqref{9.140} we observe that the limit as $k\to 0$ for the Jost solution $f_0(k,x)$ is identically zero.
Then, since $S_0(k)$ has modulus $1,$ from \eqref{9.138} we conclude that the limit as $k\to 0$ for the physical solution $\Psi_0(k,x)$
is also identically zero. In this case, we conclude that neither the Jost solution $f_0(k,x)$ nor the physical solution $\Psi_0(k,x)$ can be defined at $k=0.$
If $f_0(0,x)$ existed at $k=0,$ it would have to satisfy the asymptotics $f_0(0,x)=1+o(1)$ as $x\to +\infty.$
Similarly, if $\Psi_0(0,x)$ existed, it would have to be a nontrivial bounded solution to the zero-energy
Schr\"odinger equation.
 The zero-energy
Schr\"odinger equation with the potential $V_0(x)$ appearing in \eqref{9.134} is given by
\begin{equation}
\label{9.143}
-\psi''-\ds\frac{1}{4(1+x)^2}\,\psi=0,\qquad x\in\mathbb R^+,
\end{equation}
and it has two linearly independent solutions given by
\begin{equation}
\label{9.144}
\psi(x)=\sqrt{1+x},\quad \psi(x)=\sqrt{1+x}\,\ln(1+x),
\qquad x\in\mathbb R^+.
\end{equation}
where $\ln$ denotes the real-valued natural logarithm function.
By using a linear combination of the two solutions listed in \eqref{9.144}, it is impossible to construct a solution to \eqref{9.143} with the required asymptotic behavior as $x\to+\infty.$ 
Thus, the zero-energy Jost solution $f_0(0,x)$ does not exist and 
the question regarding the continuity at $k=0$ for the Jost solution $f_0(k,x)$ becomes irrelevant.
Similarly,
the zero-energy physical solution $\Psi_0(0,x)$ does not exist and 
the question regarding the continuity at $k=0$ for the physical solution $\Psi_0(k,x)$ becomes irrelevant.
The fact that limit as $k\to 0$ for the right-hand side of \eqref{9.140} is identically zero does not mean that
the zero-energy Jost solution $f_0(0,x)$ is given by the function that is identically zero. 
Similarly, the fact that limit as $k\to 0$ for the right-hand side of \eqref{9.138} is identically zero does not mean that
the zero-energy physical solution $\Psi_0(0,x)$ is given by the function that is identically zero.
As we have already observed, the Jost matrix $J_0(k),$ defined when $k\in\overline{\mathbb C^+}\setminus\{0\},$ 
extends continuously to $k=0$ and satisfies $J_0(0)=0.$ Since the Jost solution $f_0(k,x)$ does not exist at $k=0,$ the zero-energy Jost matrix $J_0(0)$ 
cannot be defined by using \eqref{2.11} at $k=0.$
The regular solution given in \eqref{9.139} has a limit as $k\to 0.$ That limit can be evaluated by using \eqref{9.140} and \eqref{9.141} on the right-hand side of \eqref{9.139},
and it is equal to $\sqrt{1+x}\,\ln(1+x).$
In fact, with the help of \eqref{9.144}, 
we can construct the zero-energy regular solution satisfying the initial conditions \eqref{2.6} with $A=0$ and $B=1.$ We then get
\begin{equation}
\label{9.145}
\varphi_0(0,x)=\sqrt{1+x}\,\ln(1+x),
\qquad x\in\mathbb R^+.
\end{equation}
This confirms that $\varphi_0(k,x)$ is continuous at $k=0.$
We note that a regular solution to the Schr\"odinger equation exists for
$k\in\mathbb C$ whenever the potential $V$ is locally integrable for $x\in[0,+\infty).$

\end{example}

\section{Examples of changing the discrete spectrum}
\label{section10}

In this section we provide some illustrative examples to demonstrate 
some modifications in the spectrum of the
matrix-valued Schr\"odinger
operator.  In the explicit examples provided, we start with an unperturbed  Schr\"odinger
operator for which we know the potential $V$ and the boundary matrices $A$ and $B.$
We then evaluate the corresponding perturbed  Schr\"odinger
operator and the relevant quantities when a discrete eigenvalue is added to or removed from the
spectrum. We supplement our examples with various remarks for better understanding of the methods used.
Such remarks include a correction to the result presented in
(IV.1.30) of \cite{CS1989} concerning asymptotic behavior of the potential increment when a discrete eigenvalue
is added to the spectrum of the scalar-valued Schr\"odinger operator with the Dirichlet boundary condition.

In our first example below, we illustrate the addition of a discrete eigenvalue in the scalar case 
with a non-Dirichlet boundary condition. This example illustrates that the addition of a bound state in the non-Dirichlet case 
does not necessarily change the zero unperturbed potential but it changes the unperturbed boundary condition so that
the perturbed problem has a bound state as a result of the perturbed boundary condition. This example was
treated in \cite{AU2022} by a different method.

\begin{example}
\label{example10.1} 
\normalfont
In the scalar case, i.e. when $n=1,$ let us choose the unperturbed potential $V$ and the boundary matrices $A$ and $B$ as
\begin{equation}
\label{10.1}
V (x)\equiv 0, \quad (A,B)=(1,\kappa_1),
 \end{equation}
 where $\kappa_1$ is a positive parameter and 
 we omit the matrix brackets in the $1\times 1$ matrix case.
 Using the input data set in \eqref{10.1}, 
 we evaluate the Jost solution $f(k,x),$
 the regular solution $\varphi(k,x),$ the Jost matrix $J(k),$ and the scattering matrix
 $S(k)$ as
 \begin{equation}
\label{10.2}
f(k,x)=e^{ikx},\quad \varphi(k,x)=
\cos(kx)+\ds\frac{\kappa_1}{k}\,\sin(kx),
 \end{equation}
  \begin{equation}
\label{10.3}
J(k)=-i(k+i\kappa_1),\quad
S(k)=\ds\frac{k-i\kappa_1}{k+i\kappa_1}.
 \end{equation}
 From the first equality of \eqref{10.3}, we see that $J(k)$ 
 has no zeros on the positive imaginary axis in the complex $k$-plane, and hence the 
 unperturbed problem has no bound states. To obtain our perturbed problem, 
 we add one bound state at $k=i\tilde \kappa_1$ with the Gel'fand--Levitan
 normalization matrix $\tilde C_1$ by using
 \begin{equation}
\label{10.4}
\tilde \kappa_1=\kappa_1,\quad 
\tilde C_1=\tilde c_1
 \end{equation}
 where $\tilde c_1$ is a positive constant.
 To evaluate the relevant quantities for the perturbed problem, we use the method of Section~\ref{section7}.
 From the second equality of \eqref{10.2} we get
  \begin{equation}
\label{10.5}
\varphi(i\tilde\kappa_1,x)=e^{\kappa_1 x}.
 \end{equation}
 Using \eqref{10.5} and the second equality of \eqref{10.2} in \eqref{7.7}, we get
  \begin{equation}
\label{10.6}
\xi_1(x)=\tilde c_1\,e^{\kappa_1 x}.
 \end{equation}
 Since we are in the scalar case, the orthogonal projection matrices $\tilde P_1$ and $\tilde Q_1$ 
 are both equal to $1,$
i.e. we have
 \begin{equation}
\label{10.7}
\tilde P_1=1,\quad \tilde Q=1.
 \end{equation}
Using \eqref{10.6} and \eqref{10.7} in \eqref{7.10}, we obtain
 \begin{equation}
\label{10.8}
\Omega_1(x)=1+\ds\frac{\tilde c_1^2}{2\kappa_1}\left[e^{2\kappa_1 x}-1\right].
 \end{equation}
 With the help of \eqref{10.6}, \eqref{10.8}, and the first equality of \eqref{10.1}, from \eqref{7.11} we obtain the perturbed potential $\tilde V$ as
 \begin{equation}
\label{10.9}
\tilde V(x)= \ds\frac{8 \,\tilde c_1^2 \, \kappa_1^2 \left(\tilde c_1^2-2\kappa_1\right) e^{2\kappa_1 x}}
{\left(2\kappa_1-\tilde c_1^2+\tilde c_1^2 \,e^{2\kappa_1  x}\right)^2}.
 \end{equation}
 Using the second equalities of \eqref{10.1} and \eqref{10.4} in \eqref{7.14}, we obtain
 the perturbed boundary matrices $\tilde A$ and $\tilde B$ as
 \begin{equation*}
\tilde A=1,\quad \tilde B=\kappa_1-\tilde c_1^2.
 \end{equation*}
 With the help of \eqref{10.6}, \eqref{10.8}, and the second equality of \eqref{10.2}, from \eqref{7.13} we obtain the perturbed regular solution
 $\tilde\varphi(k,x)$ as
   \begin{equation*}
\tilde\varphi(k,x)=\cos(kx)+\ds\frac{\kappa_1\,\sin(kx)}{k}\,
\ds\frac{2\kappa_1-\tilde c_1^2 -\tilde c_1^2 \,e^{2\kappa_1 x}}
{   2\kappa_1-\tilde c_1^2+\tilde c_1^2\,e^{2\kappa_1 x} }.
 \end{equation*}
 The perturbed Jost solution
 $\tilde f(k,x)$ is obtained from \eqref{7.27}, by using \eqref{10.6}, \eqref{10.8}, and the first
 equalities of \eqref{10.2}, \eqref{10.4}, and \eqref{10.7} as
 \begin{equation}
\label{10.12}
\tilde f(k,x)= e^{ikx}\left[1+\ds\frac{2\,i\,\kappa_1   \left(\tilde c_1^2-2\kappa_1\right) }{(k+i\kappa_1)
\left(2\kappa_1-\tilde c_1^2+\tilde c_1^2 \,e^{2\kappa_1  x}\right)}\right].
 \end{equation}
 The perturbed Jost matrix $\tilde J(k)$ is obtained by using the first equalities of
 \eqref{10.3}, \eqref{10.4}, and \eqref{10.7} in \eqref{7.20} as
  \begin{equation}
\label{10.13}
\tilde J(k)=-i\left(k-i\kappa_1\right),
 \end{equation}
which has a zero on the positive imaginary axis at $k=i\kappa_1,$ confirming that
the perturbed problem has a bound state at $k=i\kappa_1.$ 
 Using \eqref{10.13} in \eqref{7.25} we obtain the perturbed scattering matrix
 $\tilde S(k)$ as
   \begin{equation}
\label{10.14}
\tilde S(k)=\ds\frac{k+i\kappa_1}{k-i\kappa_1}.
 \end{equation}
From \eqref{10.9} we see that the perturbed potential becomes $\tilde V(x)\equiv 0$ if the positive parameter
$\tilde c_1$ in \eqref{10.4} is chosen as
$\tilde c_1=\sqrt{2\kappa_1}.$ In that case \eqref{10.12} yields $\tilde f(k,x)=e^{ikx}.$ For any other choices for $\tilde c_1,$ the support of
$\tilde V$ is the interval $[0,+\infty)$ and from \eqref{10.9} we get
 \begin{equation}
\label{10.15}
\tilde V(x)=O(e^{-2\kappa_1 x}),\qquad x\to+\infty.
 \end{equation}
 The asymptotic decay estimate in \eqref{10.15} is consistent with our general estimate \eqref{7.19}.

\end{example}

We remark that any example where the perturbed problem is obtained by adding a bound state to the unperturbed problem
also provides an example where a bound state is removed, and this is done by interchanging the roles of the perturbed and unperturbed problems.
In the next example we provide such an example by using the quantities in Example~\ref{example10.1} to illustrate the removal of a bound state.

\begin{example}
\label{example10.2} 
\normalfont

In the scalar case, i.e. when $n=1,$ as the unperturbed potential $V$ let us use
\begin{equation}
\label{10.16}
V (x)= \ds\frac{8 \,c_1^2 \, \kappa_1^2 \left(c_1^2-2\kappa_1\right) e^{2\kappa_1 x}}
{\left(2\kappa_1-c_1^2+c_1^2 \,e^{2\kappa_1  x}\right)^2},
 \end{equation}
 and as the unperturbed  boundary matrices $A$ and $B$ let us use
\begin{equation}
\label{10.17}
A=1,\quad B=\kappa_1-c_1^2,
 \end{equation}
 where $\kappa_1$ and $c_1$ are two positive parameters.
Corresponding to the input data set specified in \eqref{10.16} and \eqref{10.17},
with the help of \eqref{9.6}
we obtain 
 the Jost solution $f(k,x).$ 
 We then use \eqref{2.11}, \eqref{2.12}, and \eqref{2.15} to construct
 the Jost matrix $J(k),$ the scattering matrix $S(k),$ and the regular solution
$ \varphi(k,x).$ We get
 \begin{equation}
\label{10.18}
f(k,x)= e^{ikx}\left[1+\ds\frac{2\,i\,\kappa_1   \left(c_1^2-2\kappa_1\right) }{(k+i\kappa_1)
\left(2\kappa_1-c_1^2+c_1^2 \,e^{2\kappa_1  x}\right)}\right],
 \end{equation}
  \begin{equation}
\label{10.19}
J(k)=-i\left(k-i\kappa_1\right), \quad S(k)=\ds\frac{k+i\kappa_1}{k-i\kappa_1},
 \end{equation}
   \begin{equation}
\label{10.20}
\varphi(k,x)=\cos(kx)+\ds\frac{\kappa_1\,\sin(kx)}{k}\,
\ds\frac{2\kappa_1-c_1^2 -c_1^2 \,e^{2\kappa_1 x}}
{   2\kappa_1-c_1^2+c_1^2\,e^{2\kappa_1 x} }.
 \end{equation}
 From the first equality of \eqref{10.19} we observe that
 $J(k)$ has a zero on the positive imaginary axis, i.e. at $k=i\kappa_1.$
 Thus, the unperturbed problem has a bound state at 
 $k=i\kappa_1.$  In the scalar case, the orthogonal projections $P_1$ and $Q_1$
 associated with the bound state at $k=i\kappa_1$ are given by
  \begin{equation}
\label{10.21}
P_1=1,\quad Q_1=1.
 \end{equation}
 Using \eqref{10.20} we obtain
 \begin{equation}
\label{10.22}
\varphi(i\kappa_1,x)=\ds\frac{2\kappa_1\,e^{\kappa_1 x}}
{ 2\kappa_1-c_1^2+c_1^2\,e^{2\kappa_1 x} }.
 \end{equation}
Using \eqref{10.22} and the second equality of \eqref{10.21}, with the help of
\eqref{3.16}, \eqref{3.17}, and \eqref{3.19} we obtain the 
 the Gel'fand--Levitan normalization matrix,
 which is a scalar in this case, as
 \begin{equation}
\label{10.23}
C_1=c_1,
 \end{equation}
 where $c_1$ is the positive constant appearing in \eqref{10.16}--\eqref{10.20} and \eqref{10.22}.
 Using \eqref{10.22} and \eqref{10.23} in \eqref{3.12}, we obtain
 \begin{equation}
\label{10.24}
\Phi_1(x)=\ds\frac{2\,c_1\,\kappa_1\,e^{\kappa_1 x}}
{ 2\kappa_1-c_1^2+c_1^2\,e^{2\kappa_1 x} }.
 \end{equation}
 Using \eqref{10.24} in \eqref{5.5} we obtain
  \begin{equation}
\label{10.25}
W_1(x)=\ds\frac{2\,\kappa_1}
{ 2\kappa_1-c_1^2+c_1^2\,e^{2\kappa_1 x} }.
 \end{equation}
With the help of \eqref{10.16}, \eqref{10.24}, and \eqref{10.25}, from \eqref{5.8} we obtain the perturbed potential
$\tilde V(x)$ as
  \begin{equation*}
\tilde V (x)\equiv 0.
\end{equation*}
 Next, using \eqref{10.20}, \eqref{10.24}, and \eqref{10.25} in \eqref{5.7} we obtain the perturbed regular solution
 $\tilde\varphi(k,x)$ as
   \begin{equation*}
\tilde\varphi(k,x)=\cos(kx)+\ds\frac{\kappa_1}{k}\,\sin(kx).
\end{equation*}
 Similarly, using the first equalities of \eqref{10.19} and \eqref{10.21} in \eqref{5.13} we obtain the perturbed Jost matrix $\tilde J(k)$ as
 \begin{equation*}
\tilde J(k)=-i(k+i\kappa_1),
\end{equation*}
 from the zero of which at $k=-i\kappa_1$ we conclude that the perturbed operator does not have any bound states.
 The perturbed scattering matrix $\tilde S(k)$ is obtained by using 
 the second equality of \eqref{10.19} and the first equality of \eqref{10.21} in \eqref{5.15} as
  \begin{equation*}
\tilde S(k)=\ds\frac{k-i\kappa_1}{k+i\kappa_1}.
\end{equation*}
 The perturbed boundary matrices $\tilde A$ and $\tilde B$ are obtained from \eqref{5.9} by using \eqref{10.17} and \eqref{10.23}, and we have
   \begin{equation*}
\tilde A=1,\quad \tilde B=\kappa_1. 
\end{equation*}
 Finally, using \eqref{10.18}, \eqref{10.24}, \eqref{10.25}, and the first equality of \eqref{10.21}
 in \eqref{5.18} we obtain the perturbed Jost solution $\tilde f(k,x)$ as
   \begin{equation*}
\tilde f(k,x)\equiv e^{ikx}.
\end{equation*}
Thus, we confirm that the method of Section~\ref{section5} used in this example
to remove a bound state and the method of Section~\ref{section7} used in Example~\ref{example10.2} to add a bound state are complementary
in the sense that the removal undoes the addition of the bound state.

\end{example}

 In (IV.1.30) on p.~63
 of \cite{CS1989} it is claimed that, in the scalar case with the Dirichlet
 boundary condition, when we add a bound state at $k=i\kappa,$ the potential increment
 $\tilde V(x)-V(x)$ decays as $F e^{-2\kappa x}$ as $x\to+\infty,$ where $F$ is a nonzero constant.
 Our example below shows that the result (IV.1.30) 
 of \cite{CS1989} is incorrect. In order to demonstrate that that
 incorrectness is independent of the asymptotic decay of the unperturbed potential
 $V,$ we include here the result in Example~8.9 of \cite{AW2025} but by also providing
 a full complement to that result.

\begin{example}
\label{example10.3} 
\normalfont
Consider the scalar case, i.e. when $n=1,$ where the unperturbed potential is given by
\begin{equation*}
V (x)=- \ds\frac{8 \,e^{2 x}}
{\left(1+e^{2  x}\right)^2},
 \end{equation*}
 and the boundary matrices $A$ and $B,$ which are scalars in this case, are given by
\begin{equation}
\label{10.33}
A=0,\quad B=-1.
 \end{equation}
Using the boundary matrices of \eqref{10.33} in \eqref{2.7}, we observe that the
corresponding boundary condition is the Dirichlet boundary condition $\psi(0)=0.$
We apply the method of Section~\ref{section7} as detailed in Example~\ref{example10.1}
in order to add a bound state to the unperturbed operator.
The corresponding Jost solution $f(k,x),$ regular solution $\varphi(k,x),$ Jost matrix $J(k),$ and the scattering matrix
 $S(k)$ are obtained as
 \begin{equation*}
f(k,x)= e^{ikx}\left[1-\ds\frac{2\,i }{(k+i)
\left(1+e^{2  x}\right)}\right],
 \end{equation*}
  \begin{equation*}
\varphi(k,x)=-\ds\frac{k\,\sin(kx)+\cos(kx)\, \tanh x}
{k^2+1},
 \end{equation*}
  \begin{equation}
\label{10.36}
J(k)=-\ds\frac{k}{k+i},
 \end{equation}
 \begin{equation*}
S(k)=-\ds\frac{k+i}{k-i}.
 \end{equation*}
From the zero of the Jost matrix in \eqref{10.36} we see that the unperturbed operator does not have any bound states.
Let us add a bound state at $k=i\kappa_1$ with the Gel'fand--Levitan constant $C_1,$ where we have
 \begin{equation}
\label{10.38}
\kappa_1=1,\quad 
C_1=4.
 \end{equation}
Although this particular case is already covered in Example~8.9 of \cite{AW2025}, we include it here to consider all
the possibilities in order to demonstrate that our result is independent of the
asymptotic decay of the unperturbed potential $V(x)$ as $x\to+\infty.$
The corresponding perturbed Jost matrix $\tilde J(k),$ perturbed boundary matrices $\tilde A$ and $\tilde B,$
and perturbed potential $\tilde V(x)$
 are given by
 \begin{equation*}
\tilde J(k)=-\ds\frac{k(k-i)}{(k+i)^2},
 \end{equation*}
 \begin{equation*}
\tilde A=0,\quad \tilde B=-1.
 \end{equation*}
 \begin{equation*}
\tilde V(x)=\ds\frac{q_{13}(x)+q_{14}(x)}{\left[ (1-2x)\,\cosh x+(1+4 x^2+\cosh(2x) )\,\sinh x\right]^2},
\end{equation*}
where we have defined
 \begin{equation*}
q_{13}(x):=7-24 x+32 x^4+64 x^2 \cosh(2x)-(16+32x)\,\sinh(2x),
\end{equation*}
 \begin{equation*}
q_{14}(x):=-(9+8x^2)\,\cosh(4x)+(-2+20x)\,\sinh(4x).
\end{equation*}
With the choices in \eqref{10.38}, the unperturbed potential $V(x)$ behaves
as $O(e^{-2 x})$ as $x\to+\infty.$ The perturbed potential $\tilde V(x)$
behaves as  $O(x^2\,e^{-2 x})$ as $x\to+\infty,$ and hence the potential difference
$\tilde V(x)-V(x)$ also behaves as
 $O(x^2\,e^{-2 x})$ as $x\to+\infty.$ This is in contradiction to
 the result stated in (IV.1.30) of \cite{CS1989}, which incorrectly predicts
 $\tilde V(x)-V(x)=O(e^{-2x})$ as $x\to+\infty.$
 We remark that our decay estimate of $O(x^2\,e^{-2 x})$ for $\tilde V(x)-V(x)$ is consistent with our general estimate given in \eqref{7.15},
 In this case, our estimate in \eqref{7.15} is not sharp.
 Instead of using the values in \eqref{10.38}, let us add  a bound state at $k=i\kappa_1$ with the Gel'fand--Levitan constant $C_1,$ where we have
 \begin{equation}
\label{10.44}
\kappa_1=2,\quad 
C_1=3.
 \end{equation}
 The corresponding perturbed Jost matrix $\tilde J(k),$ the perturbed boundary matrices $\tilde A$ and $\tilde B,$
and the perturbed potential $\tilde V(x)$
 are then given by
 \begin{equation*}
\tilde J(k)=-\ds\frac{k(k-2i)}{(k+i)(k+2i)},
 \end{equation*}
 \begin{equation*}
\tilde A=0,\quad \tilde B=-1,
 \end{equation*}
 \begin{equation*}
\tilde V(x)=\ds\frac{q_{15}(x)+q_{16}(x)+q_{17}(x)}{\left[(16-24x)\,\cosh x-8\,\sinh x+9\,\sinh(3x)+\sinh(5x)\right]^2},
\end{equation*}
where we have let
 \begin{equation*}
q_{15}(x):=2468+1536 x-1152 x^2-1728\,\cosh(2x)-1152\,\cosh(4x)-64\,\cosh(6x),
\end{equation*}
 \begin{equation*}
q_{16}(x):=-36 \,\cosh(8x)
-2304\,\sinh(2x)+3456\,x\,\sinh(2x)-1152\,\sinh(4x),
\end{equation*}
\begin{equation*}
q_{17}(x):=1728\,x\,\sinh(4x)-256\,\sinh(6x)+384 x\,\sinh(6x).
\end{equation*}
 With the choices in \eqref{10.44}, the unperturbed potential $V(x),$ the perturbed potential $\tilde V(x),$
and the potential difference
$\tilde V(x)-V(x)$ each behave as
 $O(e^{-2 x})$ as $x\to+\infty.$ 
The result stated in (IV.1.30) of \cite{CS1989} incorrectly predicts
 $\tilde V(x)-V(x)=O(e^{-4x})$ as $x\to+\infty.$
 Our decay estimate of $O(e^{-2 x})$ for $\tilde V(x)-V(x)$ in this case agrees with our general asymptotic decay estimate stated in \eqref{7.15}.
In this case, our estimate in \eqref{7.15} is sharp.
  Instead of using the values in \eqref{10.38} or in \eqref{10.44}, let us add  a bound state at $k=i\kappa_1$ with the Gel'fand--Levitan constant $C_1,$ where we have
 \begin{equation}
\label{10.51}
\kappa_1=3,\quad 
C_1=1.
 \end{equation}
  The corresponding perturbed Jost matrix $\tilde J(k),$ perturbed boundary matrices $\tilde A$ and $\tilde B,$
and perturbed potential $\tilde V(x)$
 are given by 
  \begin{equation*}
\tilde J(k)=-\ds\frac{k(k-3i)}{(k+i)(k+3i)},
\end{equation*}
 \begin{equation*}
\tilde A=0,\quad \tilde B=-1,
 \end{equation*}
  \begin{equation*}
\tilde V(x)=\ds\frac{q_{18}(x)+q_{19}(x)+q_{20}(x)+q_{21}(x)}{\left[192-12x-3\,\sinh(2x)+3\,\sinh(4x)+\sinh(6x)\right]^2},
\end{equation*}
where we have defined
 \begin{equation*}
q_{18}(x):=
-2 \, \text{\rm{sech}}^2 x+ 73728\,\cosh^3 x\,\sinh x - 4608 x \,\cosh^3 x\,\sinh x,
 \end{equation*}
  \begin{equation*}
q_{19}(x):=- 
221184 \,\cosh^3 x\,\cosh(2 x)\, \sinh x+ 
13824 x \,\cosh^3 x\,\cosh(2x)\,\sinh x,
\end{equation*}
  \begin{equation*}
q_{20}(x):=- 
6336 \,\cosh^3 x\,\sinh x\,\sinh(2x)
+ 1152\,\cosh^3 x\,\sinh x\,\,\sinh(4x),
\end{equation*}
  \begin{equation*}
q_{21}(x):=- 192 \,\cosh^3 x \,\sinh x\,\sinh(6x).
\end{equation*}
  With the choices in \eqref{10.51}, the unperturbed potential $V(x),$ the perturbed potential $\tilde V(x),$
and the potential difference
$\tilde V(x)-V(x)$ each behave as
 $O(e^{-2 x})$ as $x\to+\infty.$
 This is also in contradiction to
 the result stated in (IV.1.30) of \cite{CS1989}, which incorrectly predicts the asymptotic decay
 $\tilde V(x)-V(x)=O(e^{-6x})$ as $x\to+\infty.$
 Our asymptotic decay estimate
 of $O(e^{-2 x})$  
for $\tilde V(x)-V(x)$ agrees with the asymptotic estimate stated in \eqref{7.15}.
In this case, our estimate in \eqref{7.15} is sharp.

\end{example}

 In the next example we consider the addition of a bound state
 of multiplicity $1$ in the $2\times 2$ matrix case with a non-Dirichlet boundary condition.
 In the example, we illustrate the construction of
 the orthogonal projection matrix $\tilde P_{N+1}$ appearing in \eqref{7.20} by using two
 methods. The first method uses the construction specified \eqref{7.21} and the second method uses
 the procedure described in Theorems~\ref{theorem9.9}.

\begin{example}
\label{example10.4}
\normalfont
In this example we use the unperturbed potential $V(x)$ and the boundary matrices $A$ and $B$
given by 
\begin{equation}
\label{10.59}
V (x)=- \ds\frac{8 \, e^{2x/3}}{9\left(1+2\,e^{2x/3}\right)^2}
\begin{bmatrix} 1&1\\
1&1\end{bmatrix},
 \end{equation}
  \begin{equation}
\label{10.60}
A=\begin{bmatrix} 1&0\\
0&1\end{bmatrix},\quad B=\ds\frac{1}{18} \begin{bmatrix} 17&-1\\
-1&17\end{bmatrix}.
 \end{equation}
As described in Example~\ref{example9.3},  we construct the corresponding Jost solution $f(k,x),$ Jost matrix $J(k),$ and regular solution as
 \begin{equation}
\label{10.61}
f(k,x)= e^{ikx}\left(\begin{bmatrix} 1&0\\
0&1\end{bmatrix}-\ds\frac{i}{(3k+i)
\left(1+2\,e^{2x/3}\right)}\begin{bmatrix} 1&1\\
1&1\end{bmatrix}\right),
 \end{equation}
  \begin{equation}
\label{10.62}
J(k)=\ds\frac{-i(k+i)}{2(3k+i)}
\begin{bmatrix} 6k+i &-i\\
-i&6k+i\end{bmatrix},
 \end{equation}
  \begin{equation}
\label{10.63}
\varphi(k,x)= \ds\frac{1}{q_{22}(k,x)}\,
\begin{bmatrix} q_{23}(k,x)-q_{23}(-k,x)&q_{24}(k,x)-q_{24}(-k,x)\\
\noalign{\medskip}
q_{24}(k,x)-q_{24}(-k,x)&q_{23}(k,x)-q_{23}(-k,x)\end{bmatrix},
 \end{equation}
 where we have defined
  \begin{equation}
\label{10.64}
 q_{22}(k,x):=4k(3k+i)(3k-i)(1+2 e^{2x/3}),
 \end{equation}
   \begin{equation}
\label{10.65}
 q_{23}(k,x):=e^{ikx}(k-i)\left[18k^2-3ik+1+e^{2x/3}(36k^2+6ik+2)\right],
 \end{equation}
   \begin{equation}
\label{10.66}
 q_{24}(k,x):=e^{ikx}(k-i)\left[-3ik-1+e^{2x/3}(6ik-2)\right].
 \end{equation}
Even though $q_{22}(k,x)$ appearing in the denominator on the right-hand side of
\eqref{10.63} has simple zeros at $k=0,$ $k=i/3,$ and $k=-i/3,$ that right-hand side
in \eqref{10.63} is entire in $k.$ Hence,
those three singularities of the right-hand side in \eqref{10.63} are removable singularities.
 From \eqref{10.62} we evaluate the determinant of $J(k)$ as
  \begin{equation*}
\det[J(k)]=-\ds\frac{3k(k+i)^2}{3k+i}.
 \end{equation*}
 From \eqref{10.62} we obtain
   \begin{equation}
\label{10.68}
J(i)=\ds\frac{1}{4}
\begin{bmatrix} 7 &-1\\
-1&7\end{bmatrix}.
 \end{equation}
 Since $\det[J(k)]$ does not have a zero on the positive imaginary axis, the unperturbed operator does not have any bound states.
Using the method of Section~\ref{section7}, let us add a simple bound state at $k=i\tilde\kappa_1$ with the Gel'fand--Levitan normalization matrix $\tilde C_1,$ where we use
 \begin{equation}
\label{10.69}
\tilde\kappa_1=1,\quad 
\tilde C_1=
\begin{bmatrix}2&0\\
0&0\end{bmatrix}.
 \end{equation}
 From \eqref{3.19} and the second equality of \eqref{10.69} we see that the 
 the corresponding orthogonal projection $\tilde Q_1$ associated with the bound-state data
 in \eqref{10.69} is given by
  \begin{equation}
\label{10.70}
\tilde Q_1=
\begin{bmatrix}1&0\\
0&0\end{bmatrix}.
  \end{equation}
  Let us remark that $\tilde C_1$ and $\tilde Q_1$ are compatible. In fact, if we start with the matrix $\tilde Q_1$ given in
  \eqref{10.70} and try to construct all nonzero nonnegative matrices $\tilde C_1$ satisfying
   \begin{equation}
\label{10.71}
\tilde C_1\,\tilde Q_1=\tilde Q_1\,\tilde C_1=\tilde C_1,
  \end{equation}
  we observe that $\tilde C_1$ must have the form
    \begin{equation*}
\tilde C_1=
\begin{bmatrix}a&0\\
0&0\end{bmatrix},
  \end{equation*}
  where $a$ is a positive constant. In our example, we simply let $a=2.$
  Since the rank of $\tilde Q_1$ is equal to one, 
the corresponding orthogonal projection $\tilde P_1$ must also have rank one. Hence, $\tilde P_1$ must be equal to one of the two
matrices appearing on the right-hand side of \eqref{9.101}.
In order to determine the exact matrix $\tilde P_1,$ we proceed as in Theorem~\ref{theorem9.9} and impose
the condition that $\tilde J(i)\,\tilde Q_1=0.$ With the help of \eqref{7.20} we get the restriction
  \begin{equation}
\label{10.73}
\left(\begin{bmatrix} 1 &0\\
0&1\end{bmatrix}-\ds\frac{2i}{2i}
\,\begin{bmatrix} \ds\frac{1}{2}\pm\ds\sqrt{\ds\frac{1}{4}-|\tilde p_1|^2}&\tilde p_1\\
\tilde p_1^\ast& \ds\frac{1}{2}\mp\ds\sqrt{\ds\frac{1}{4}-|\tilde p_1^2|}\end{bmatrix}\right)
\ds\frac{1}{4}
\begin{bmatrix} 7 &-1\\
-1&7\end{bmatrix}\begin{bmatrix}1 &0\\
0&0\end{bmatrix}=
\begin{bmatrix}0 &0\\
0&0\end{bmatrix},
 \end{equation}
where we have used \eqref{10.68} and \eqref{10.70}.
Note that \eqref{10.73} is equivalent to a system of equations containing the unknown $\tilde p_1,$ and we have
  \begin{equation}
\label{10.74}
\begin{cases}
7\left(\ds\frac{1}{2}\mp\ds\sqrt{\ds\frac{1}{4}-|\tilde p_1|^2}   \right)+\tilde p_1=0,\\
-7\tilde p_1^\ast-\left(\ds\frac{1}{2}\pm\ds\sqrt{\ds\frac{1}{4}-|\tilde p_1|^2}\right)=0,
\end{cases}
 \end{equation}
 where we recall that the asterisk is used to denote complex conjugation.
 We write the system in \eqref{10.74} as
 \begin{equation}
\label{10.75}
\begin{cases}
7\left(\ds\frac{1}{2}\mp\ds\sqrt{\ds\frac{1}{4}-|\tilde p_1|^2}   \right)+\tilde p_1=0,\\
49\,\tilde p_1^\ast+7\left(\ds\frac{1}{2}\pm\ds\sqrt{\ds\frac{1}{4}-|\tilde p_1|^2}\right)=0,
\end{cases}
 \end{equation}
 which is obtained from \eqref{10.74} by multiplying the second equality by $-7$ on both sides.
 By adding the two equalities in \eqref{10.75} we obtain
 \begin{equation*}
7+ \tilde p_1+49\,\tilde p_1^\ast=0,
 \end{equation*}
and hence $\tilde p_1$ must be real and we must have
$\tilde p_1=-7/50.$
 One can directly verify that any one of the three equivalent systems in \eqref{10.73}--\eqref{10.75}
 is satisfied when $\tilde p_1=-7/50$ provided we use only the upper signs there. Thus, the orthogonal projection matrix
 $\tilde P_1$ is given by
  \begin{equation}
\label{10.77}
\tilde P_1=\ds\frac{1}{50}
\begin{bmatrix}49&-7\\
-7&1\end{bmatrix}.
 \end{equation}
 Let us construct the same matrix
 $\tilde P_1$ appearing in \eqref{10.77} by using \eqref{7.21}. For this, we proceed as follows.
 We first need to construct the invertible matrix $L_1$ appearing in \eqref{7.21} when $N=1.$ From \eqref{7.22} we observe that
 $L_1$ is the constant $2\times 2$ matrix obtained as the coefficient of the term $e^x$ in the expression for $\varphi(i,x).$
 With the help of \eqref{10.63}--\eqref{10.66} we see that $L_1$ is given by
   \begin{equation*}
L_1=\lim_{x\to+\infty} \, \ds\frac{-e^x}{q_{22}(i,x)}\,
\begin{bmatrix} q_{23}(-i,x)&q_{24}(-i,x)\\
\noalign{\medskip}
q_{24}(-i,x)&q_{23}(-i,x)\end{bmatrix},
 \end{equation*}
 which yields
 \begin{equation}
\label{10.79}
L_1=\frac{1}{8}\,
\begin{bmatrix} 7&-1\\
-1&7\end{bmatrix},
 \end{equation}
 Using \eqref{10.69} and \eqref{10.79} we obtain
  \begin{equation}
\label{10.80}
L_1\,\tilde C_1=\frac{1}{4}\,
\begin{bmatrix} 7&0\\
-1&0\end{bmatrix}.
 \end{equation}
 Finally, using \eqref{10.80} in \eqref{7.21} we get $\tilde P_1$ and observe that
 it is identical to the matrix appearing on the right-hand side of \eqref{10.77}. 
 The perturbed boundary matrices $\tilde A$ and $\tilde B$ 
 are obtained by using \eqref{10.60} and the second equality of \eqref{10.69}, and we have
  \begin{equation*}
\tilde A=\begin{bmatrix} 1&0\\
0&1\end{bmatrix},\quad \tilde B=\ds\frac{1}{18} \begin{bmatrix} -55&-1\\
-1&17\end{bmatrix}.
 \end{equation*}
 We obtain the perturbed Jost matrix $\tilde J(k)$ by using \eqref{10.62}, the first equality of \eqref{10.69}, and \eqref{10.77}, and we get
  \begin{equation}
\label{10.82}
\tilde J(k)=\ds\frac{1}{50(3k+i)}
\begin{bmatrix} (k-i)(-150 i k+31)&17k+31i\\
\noalign{\medskip}
17(k-i)&-150 ik^2+169k+17i\end{bmatrix}.
 \end{equation}
 Using \eqref{10.82} in \eqref{2.12}, we obtain the corresponding scattering matrix $\tilde S(k)$ as
   \begin{equation}
 \label{10.83}
\tilde S(k)=\ds\frac{1}{625(k^2+1)(3k-i)}
\begin{bmatrix} q_{25}(k)&17 i(31-25k^2)\\
\noalign{\medskip}
17 i(31-25k^2)&q_{25}(k)^\ast\end{bmatrix},
 \end{equation}
 where we have defined 
  \begin{equation*}
q_{25}(x):=336i-1525 k+3600 i k^2+1875 k^3.
  \end{equation*}
 From \eqref{10.82} we see that
the determinant of the perturbed Jost matrix is given by
  \begin{equation*}
\det[\tilde J(k)]=\ds\frac{-3k(k+i)(k-i)}{3k+i},
 \end{equation*}
 from which we see that there is a simple bound state at $k=i.$
 The corresponding perturbed potential $\tilde V(x)$ is obtained as in \eqref{7.11} with the help of \eqref{7.7} and \eqref{7.10}.
 We obtain
\begin{equation*}
\tilde V (x)=- \ds\frac{8 \, e^{2x/3}}{9\left(-9-18\,e^{2x/3}+26\,e^{2x}+25\,e^{8x/3}\right)^2}
\begin{bmatrix} q_{26}(x)&q_{27}(x)\\
\noalign{\medskip}
q_{27}(x)&q_{28}(x)\end{bmatrix},
 \end{equation*}
 where we have defined
 \begin{equation*}
q_{26}(x):=81 - 2025 \,e^{4 x/3} - 5328 \,e^{2 x} - 3969 \,e^{8 x/3)}+ 361 \,e^{4 x},
  \end{equation*}
  \begin{equation*}
q_{27}(x):=-81 + 405 \,e^{4 x/3} + 144 \,e^{2 x} - 567 \,e^{8 x/3} + 323 \,e^{4 x},
  \end{equation*}
 \begin{equation*}
q_{28}(x):=81 - 81 \,e^{4 x/3} - 144 \,e^{2 x} - 81 \,e^{8 x/3} + 289 \,e^{4 x}.
  \end{equation*}
 The corresponding perturbed Jost solution $\tilde f(k,x)$ is obtained from \eqref{7.27} with the help of \eqref{7.7}, \eqref{7.10}, 
 \eqref{10.61},  the first equality of \eqref{10.69}, and \eqref{10.77}. We get
  \begin{equation}
\label{10.90}
\tilde f(k,x)=e^{ikx}\left(\begin{bmatrix} 1&0\\
0&1\end{bmatrix}+\ds\frac{1}{q_{29}(k,x)}\begin{bmatrix} q_{30}(k,x)&q_{31}(k,x)\\
\noalign{\medskip}
q_{32}(k,x)&q_{33}(k,x)\end{bmatrix}\right),
 \end{equation}
 where we have let
 \begin{equation*}
q_{29}(k,x):=25(3k+i)(k+i)
\left(-9-18\,e^{2x/3}+26\,e^{2x}+25\,e^{8x/3}\right),
  \end{equation*}
  \begin{equation*}
q_{30}(k,x):=36(-8+43\,ik)+882\,e^{2x/3}\,(-1+3ik)+722\,e^{2x}(1-ik),
  \end{equation*}
   \begin{equation*}
q_{31}(k,x):=36(-6+ik)+126\,e^{2x/3}\,(1-3ik)+646\,e^{2x}(-1+ik),
  \end{equation*}
  \begin{equation*}
q_{32}(k,x):=36(6+ik)+126\,e^{2x/3}\,(1-3ik)+646\,e^{2x}(-1+ik),
  \end{equation*}
 \begin{equation*}
q_{33}(k,x):=36(-8+7ik)+18\,e^{2x/3}\,(-1+3ik)+578\,e^{2x}(1-ik).
  \end{equation*}
 In the construction of $\tilde f(k,x)$ given in \eqref{10.90}, we have constructed the Moore--Penrose inverse of the matrix
 $\Omega_1(x)$ defined in \eqref{7.10} when $N=0.$
 For the evaluation of $\Omega_1(x)^+,$ we have used the symbolic software Mathematica, but we have also 
 independently verified that the matrices $\Omega_1(x)$ and $\Omega_1(x)^+$ satisfy the four matrix equalities
 given in (3.2) of \cite{AW2025} used in the definition \cite{BG2003} of the
 Moore--Penrose inverse of a matrix.
With the help of the perturbed scattering matrix $\tilde S(k)$ given in \eqref{10.83} and the perturbed Jost solution
  $\tilde f(k,x)$ given in \eqref{10.90}, from \eqref{2.15} we construct the perturbed regular solution
  $\tilde\varphi(k,x)$ explicitly. Since the expression is too lengthy, we do not display it here.
  We obtain $\tilde\varphi(i,x)$ as
    \begin{equation}
\label{10.96}
\tilde \varphi(i,x)=\ds\frac{1}{q_{34}(k,x)}\begin{bmatrix} q_{35}(k,x)&q_{36}(k,x)\\
\noalign{\medskip}
q_{37}(k,x)&q_{38}(k,x)\end{bmatrix},
 \end{equation}
 where we have let
 \begin{equation*}
q_{34}(x):=-9-18 e^{2x/3}+26 e^{2x}+25 e^{8x/3},
  \end{equation*}
   \begin{equation*}
q_{35}(x):=10 e^x+14 e^{5x/3},\quad
q_{37}(x):=2 e^x-2 e^{5x/3},\quad
  \end{equation*}
   \begin{equation*}
q_{36}(x):=-e^x+4 e^{5x/3}-6 e^{3x}+3 e^{11x/3},
  \end{equation*}
   \begin{equation*}
q_{38}(x):=-11 e^x-16 e^{5x/3}+30 e^{3x}+21 e^{11x/3}.
  \end{equation*}
  Using $\tilde Q_1$ given in \eqref{10.70}
  and $\tilde\varphi(i,x)$ given in \eqref{10.96}, with the help of
  \eqref{3.16}--\eqref{3.19} we construct the corresponding Gel'fand--Levitan normalization matrix
  $\tilde C_1,$ and we observe that the constructed $2\times 2$ matrix
  $\tilde C_1$ coincides with the matrix appearing in the second equality of \eqref{10.69}.
  In this example, the unperturbed potential $V(x),$ the perturbed potential $\tilde V(x),$ and the 
potential difference $\tilde V(x)-V(x)$ each behave as
$O(e^{-2x/3})$ as $x\to+\infty.$ The asymptotic decay
for $\tilde V(x)-V(x)$ agrees with the asymptotic estimate stated in \eqref{7.15}.
In this case, our estimate in \eqref{7.15} is sharp.

\end{example}

In the next example we demonstrate the addition of a bound state of multiplicity $2$ by using the method of Section~\ref{section7}.
The unperturbed  operator in this example coincides with the 
unperturbed operator of Example~\ref{example10.4}.

\begin{example}
\label{example10.5}
\normalfont
In this example, let us assume that our unperturbed operator
corresponds to the potential $V(x)$ given in \eqref{10.59} and to the boundary matrices $A$ and $B$ 
given in \eqref{10.60}. We know from Example~\ref{example10.4} that the unperturbed operator
does not have any bound states.
Instead of adding a simple bound state at $k=i$ as in
Example~\ref{example10.4}, let us add a bound state at $k=i$ with multiplicity two. We choose our 
Gel'fand--Levitan normalization matrix $\tilde C_1$ as
\begin{equation}
\label{10.101}
\tilde C_1=\ds\frac{1}{6}\,
\begin{bmatrix}4+3\,\sqrt{2}&4-3\,\sqrt{2}\\
\noalign{\medskip}
4-3\,\sqrt{2}&4+3\,\sqrt{2}\end{bmatrix},
 \end{equation}
 The matrix $\tilde C_1$ is a positive matrix because
 its eigenvalues are $4/3$ and $\sqrt{2}.$
Since the multiplicity of the bound state to be added is $2$ in this
$2\times 2$ matrix case, we know that
 the orthogonal projection matrices $\tilde Q_1$ and $\tilde P_1$ associated 
 with the bound state at $k=i$ are both equal to the $2\times 2$ identity matrix.
 Using $\tilde\kappa_1=1$ and the unperturbed Jost matrix $J(k)$ given in \eqref{10.62}, from
 \eqref{7.20} we obtain the perturbed Jost matrix
 $\tilde J(k)$ as
 \begin{equation}
\label{10.102}
\tilde J(k)=\ds\frac{-i(k-i)}{2(3k+i)}
\begin{bmatrix} 6k+i &-i\\
-i&6k+i\end{bmatrix}.
 \end{equation}
 In fact, in this case we have
 \begin{equation*}
\tilde J(k)=\ds\frac{k-i}{k+i}
\,J(k).
 \end{equation*}
 The determinant of the perturbed Jost matrix $\tilde J(k)$ is given by
  \begin{equation*}
\det[\tilde J(k)]=-\ds\frac{3k(k-i)^2}{3k+i},
 \end{equation*}
and hence there is a bound state at $k=i$ with the multiplicity $2.$
 The corresponding perturbed boundary matrices
 $\tilde A$ and $\tilde B$ are obtained from \eqref{7.14} with the help of
 \eqref{10.60} and \eqref{10.101}. We obtain
 \begin{equation}
\label{10.105}
\tilde A=\begin{bmatrix} 1&0\\
0&1\end{bmatrix},\quad \tilde B=\ds\frac{1}{18} \begin{bmatrix} -17&1\\
1&-17\end{bmatrix}.
 \end{equation}
 The perturbed potential  $\tilde V(x)$ is obtained from \eqref{7.11} with the help of \eqref{7.7}, \eqref{7.10}, and \eqref{10.101}.
 We get
 \begin{equation}
\label{10.106}
\tilde V(x)=- \ds\frac{8 \, e^{2x/3}}{9\left(2+e^{2x/3}\right)^2}
\begin{bmatrix} 1&1\\
1&1\end{bmatrix}.
 \end{equation}
 The corresponding perturbed Jost solution
 $\tilde f(k,x)$ is obtained from \eqref{7.27} with the help of \eqref{7.7}, \eqref{7.10}, and \eqref{10.61}. We have
 \begin{equation}
\label{10.107}
\tilde f(k,x)= e^{ikx}\left(\begin{bmatrix} 1&0\\
0&1\end{bmatrix}-\ds\frac{2\,i}{(3k+i)
\left(2+e^{2x/3}\right)}\begin{bmatrix} 1&1\\
1&1\end{bmatrix}\right).
 \end{equation}
 As done in Example~\ref{example10.4}, we verify the compatibility of the
  Gel'fand--Levitan normalization matrix
  $\tilde C_1$ and the orthogonal projection matrix
  $\tilde Q_1=I.$
  We observe that the two equalities in \eqref{10.71} are trivially satisfied.
  As mentioned already the matrix $\tilde C_1$ is chosen as a positive matrix.
  We construct the corresponding perturbed regular solution $\tilde\varphi(k,x)$ as
     \begin{equation}
\label{10.108}
\tilde\varphi(k,x)= \ds\frac{1}{4k(9k^2+1)(2+e^{2x/3})}\,
\begin{bmatrix} q_{5}(k,x)-q_{5}(-k,x)&q_{39}(k,x)-q_{39}(-k,x)\\
\noalign{\medskip}
q_{39}(k,x)-q_{39}(-k,x)&q_{5}(k,x)-q_{5}(-k,x)\end{bmatrix},
 \end{equation}
 where we have defined
   \begin{equation}
\label{10.109}
 q_{5}(k,x):=e^{ikx}(k+i)\left[36k^2-6ik+2+e^{2x/3}(18k^2+3ik+1)\right],
 \end{equation}
   \begin{equation}
\label{10.110}
 q_{39}(k,x):=e^{ikx}(k+i)\left[-6ik-2+e^{2x/3}(3ik-1)\right].
 \end{equation}
 Using $\tilde Q_1=I$ and $\tilde\varphi(i,x),$ with the help of
  \eqref{3.16}--\eqref{3.19} we construct the corresponding Gel'fand--Levitan normalization matrix
  $\tilde C_1,$ and we observe that the constructed $2\times 2$ matrix
  $\tilde C_1$ coincides with the matrix appearing in \eqref{10.101}.
In this example, the unperturbed potential $V(x),$ the perturbed potential $\tilde V(x),$ and the 
potential difference $\tilde V(x)-V(x)$ each behave as
$O(e^{-2x/3})$ as $x\to+\infty.$ The asymptotic decay of
$\tilde V(x)-V(x)$ is the same as the decay estimate predicted in \eqref{7.15}.
In this case, our estimate given in \eqref{7.15} is sharp.

 \end{example}

In the next example we illustrate the method of Section~\ref{section6} used to reduce
the multiplicity of a bound state. As our unperturbed problem we use the perturbed problem of
Example~\ref{example10.5}.

\begin{example}
\label{example10.6}
\normalfont
We choose our unperturbed operator in such a way that
the unperturbed potential $V(x)$ is
\begin{equation}
\label{10.111}
V(x)=- \ds\frac{8 \, e^{2x/3}}{9\left(2+e^{2x/3}\right)^2}
\begin{bmatrix} 1&1\\
1&1\end{bmatrix},
 \end{equation}
 which coincides with the right-hand side of \eqref{10.106}, and we choose
 the unperturbed boundary matrices $A$ and $B$ as
  \begin{equation}
\label{10.112}
A=\begin{bmatrix} 1&0\\
0&1\end{bmatrix},\quad B=\ds\frac{1}{18} \begin{bmatrix} -17&1\\
1&-17\end{bmatrix},
 \end{equation}
 which coincide with the boundary matrices listed in \eqref{10.105}.
 The unperturbed Jost solution $f(k,x),$ the unperturbed Jost matrix $J(k),$ and the unperturbed regular solutions $\varphi(k,x)$ 
 are respectively given by
  \begin{equation}
\label{10.113}
f(k,x)= e^{ikx}\left(\begin{bmatrix} 1&0\\
0&1\end{bmatrix}-\ds\frac{2\,i}{(3k+i)
\left(2+e^{2x/3}\right)}\begin{bmatrix} 1&1\\
1&1\end{bmatrix}\right).
 \end{equation}
  \begin{equation}
 \label{10.114}
J(k)=\ds\frac{-i(k-i)}{2(3k+i)}
\begin{bmatrix} 6k+i &-i\\
-i&6k+i\end{bmatrix},
 \end{equation}
  \begin{equation}
\label{10.115}
\varphi(k,x)= \ds\frac{1}{4k(9k^2+1)(2+e^{2x/3})}\,
\begin{bmatrix} q_{5}(k,x)-q_{5}(-k,x)&q_{39}(k,x)-q_{39}(-k,x)\\
\noalign{\medskip}
q_{39}(k,x)-q_{39}(-k,x)&q_{5}(k,x)-q_{5}(-k,x)\end{bmatrix},
 \end{equation}
 where 
 $q_{5}(k,x)$ and 
 $q_{39}(k,x)$ are the quantities defined in \eqref{10.109} and \eqref{10.110}, respectively.
 We remark that the Jost solution $f(k,x)$
 given in \eqref{10.113} coincides with $\tilde f(k,x)$ appearing in \eqref{10.107},
 the Jost matrix given in \eqref{10.114} coincides with $\tilde J(k)$ appearing in \eqref{10.102}, and
 the regular solution $\varphi(k,x)$ given in \eqref{10.115} coincides with
 $\tilde\varphi(k,x)$ appearing in \eqref{10.108}.
Even though the denominator in \eqref{10.115} appears to have simple poles at $k=0,$ $k=i/3,$ and $k=-i/3,$ the left-hand side
in \eqref{10.115} is entire in $k.$ Hence,
those three singularities of the right-hand side in \eqref{10.115} are removable singularities.
 From \eqref{10.114} we get the determinant of the Jost matrix $J(k)$ as
   \begin{equation}
\label{10.116}
\det[J(k)]= \ds\frac{-3k(k-i)^2}{3k+i}.
 \end{equation}
 Since the right-hand side of \eqref{10.116} has a double zero at $k=i$ occurring on the positive imaginary axis
 in the complex $k$-plane, we see that
 the unperturbed operator has a bound state at $k=i$ with the multiplicity $2.$
 We are interested in reducing that multiplicity from $2$ to $1$ by using the
 method of Section~\ref{section6}. To simplify our notation, we use
$Q_{\text{\rm{r}}},$
$\mathbf G_{\text{\rm{r}}},$
$\mathbf H_{\text{\rm{r}}},$
$C_{\text{\rm{r}}},$
$\Phi_{\text{\rm{r}}},$
$W_{\text{\rm{r}}},$
$P_{\text{\rm{r}}},$
to denote the quantities
$Q_{N\text{\rm{r}}},$
$\mathbf G_{N\text{\rm{r}}},$
$\mathbf H_{N\text{\rm{r}}},$
$C_{N\text{\rm{r}}},$
$\Phi_{N\text{\rm{r}}},$
$W_{N\text{\rm{r}}},$
$P_{N\text{\rm{r}}},$
respectively, appearing in Section~\ref{section6}.
From \eqref{10.114} we observe that
 $J(i)$ is the $2\times 2$ zero matrix. and hence the orthogonal projection
 $Q$ onto the kernel of $J(i)$ is equal to the $2\times 2$ identity matrix.
 This is not surprising because the multiplicity of the bound state is $2$ and
 we are in the $2\times 2$ matrix case.
 In our example, let us choose the matrix
 $Q_{\text{\rm{r}}}$ as
 \begin{equation}
\label{10.117}
Q_{\text{\rm{r}}}=\begin{bmatrix} 1&0\\
0&0\end{bmatrix},
 \end{equation}
and we observe that \eqref{6.2} is satisfied.
Using \eqref{10.115} at $k=i$ we get
  \begin{equation}
\label{10.118}
\varphi(i,x)= \ds\frac{e^{-x}}{4(2+e^{2x/3})}\,
\begin{bmatrix} 7+5 \,e^{2x/3}&-1+e^{2x/3}\\
\noalign{\medskip}
-1+e^{2x/3}&7+5 \,e^{2x/3}\end{bmatrix}.
 \end{equation}
Using \eqref{10.117} and \eqref{10.118} on the right-hand side of \eqref{6.3} we have
\begin{equation}
\label{10.119}
\mathbf G_{\text{\rm{r}}}=\begin{bmatrix} \ds\frac{17}{32}&0\\
\noalign{\medskip}
0&0\end{bmatrix},
 \end{equation}
Next, using \eqref{10.117} and \eqref{10.119} in \eqref{6.4}, we obtain
\begin{equation}
\label{10.120}
\mathbf H_{\text{\rm{r}}}=\begin{bmatrix} \ds\frac{17}{32}&0\\
\noalign{\medskip}
0&1\end{bmatrix},
 \end{equation}
which is a positive matrix.
Then, with the help of \eqref{6.6}, \eqref{10.117}, \eqref{10.120}
we construct the nonnegative matrix
$C_{\text{\rm{r}}}$ as
\begin{equation}
\label{10.121}
C_{\text{\rm{r}}}=\begin{bmatrix} \ds\frac{4\sqrt{2}}{\sqrt{17}}&0\\
\noalign{\medskip}
0&0\end{bmatrix}.
 \end{equation}
Using \eqref{10.112} and \eqref{10.121} in \eqref{6.15}, we get
\begin{equation}
\label{10.122}
\tilde A
=\begin{bmatrix} 1&0\\
0&1\end{bmatrix},\quad \tilde B=\begin{bmatrix} \ds\frac{287}{306}& \ds\frac{1}{18}
\\
\noalign{\medskip}
\ds\frac{1}{18}&-\ds\frac{17}{18}\end{bmatrix}.
 \end{equation}
Using \eqref{10.118} and \eqref{10.121} in \eqref{6.8} we have
\begin{equation}
\label{10.123}
\Phi_{\text{\rm{r}}}(x)= \ds\frac{\sqrt{2}\,e^{-x}}{\sqrt{17}\,(2+e^{2x/3})}\,
\begin{bmatrix}7+5\,e^{2x/3}&0\\
\noalign{\medskip}
-1+e^{2x/3}&0\end{bmatrix}.
 \end{equation}
Then, using \eqref{10.123} in \eqref{6.11} we get
\begin{equation}
\label{10.124}
W_{\text{\rm{r}}}(x)= \ds\frac{e^{-2x}\,(25+26\,e^{2x/3})}{17\,(2+e^{2x/3})}\,
\begin{bmatrix} 1&0\\
0&0\end{bmatrix}.
 \end{equation}
 The Moore--Penrose inverse $W_{\text{\rm{r}}}(x)^+$
of the matrix $W_{\text{\rm{r}}}(x)$ is obtained as
\begin{equation}
\label{10.125}
W_{\text{\rm{r}}}(x)^+= \ds\frac{17\,(2+e^{2x/3})}{e^{-2x}\,(25+26\,e^{2x/3})}\,
\begin{bmatrix} 1&0\\
0&0\end{bmatrix},
 \end{equation}
where we observe that the $(1,1)$-entries in
\eqref{10.124} and \eqref{10.125} are reciprocals of each other.
Using \eqref{10.111}, \eqref{10.123}, and \eqref{10.125} in \eqref{6.12} we obtain the perturbed
potential $\tilde V(x)$ as
\begin{equation}
\label{10.126}
\tilde V(x)= \ds\frac{e^{2x/3}}{9\,(25+26\,e^{2x/3})^2}\,
\begin{bmatrix} -2888&2584\\
2584&-2312\end{bmatrix}.
 \end{equation}
Using \eqref{10.115}, \eqref{10.123}, and \eqref{10.125} in \eqref{6.14}, we obtain the perturbed regular solution $\tilde\varphi(k,x)$ as
\begin{equation}
\label{10.127}
\tilde \varphi(k,x)=
 \ds\frac{1}{q_{40}(k,x)}\,
\begin{bmatrix} q_{41}(k,x)-q_{41}(-k,x)&q_{42}(k,x)-q_{42}(-k,x)\\
\noalign{\medskip}
q_{43}(k,x)-q_{43}(-k,x)&q_{44}(k,x)-q_{44}(-k,x)\end{bmatrix},
 \end{equation}
where we have defined
\begin{equation}
\label{10.128}
q_{40}(k,x):=
 4k(9k^2+1)\,(25+26\,e^{2x/3}),
  \end{equation}
 \begin{equation}
\label{10.129}
\begin{split}
 q_{41}(k,x):=&e^{ikx}\left(-31i-26k-507ik^2+450k^3\right)\\
&+2\, e^{ikx+2 x/3}\left(-7i+64k-177ik^2+234k^3\right),
\end{split}
 \end{equation}
 \begin{equation}
\label{10.130}
 q_{42}(k,x):=e^{ikx}\left(31i-76k+51ik^2\right)
+2\, e^{ikx+2 x/3}\left(7i+38k-51ik^2\right),
 \end{equation}
 \begin{equation}
\label{10.131}
 q_{43}(k,x):=17e^{ikx}\left(-i+4k+3ik^2\right)
+34\, e^{ikx+2 x/3}\left(-i-2k-3ik^2\right),
 \end{equation}
\begin{equation}
\label{10.132}
\begin{split}
 q_{44}(k,x):=&e^{ikx}\left(17i+118k+357ik^2+450k^3\right)\\
&+2\, e^{ikx+2 x/3}\left(17i-8k+255ik^2+234k^3\right).
\end{split}
 \end{equation}
From \eqref{10.117} we see that a unit vector generating $Q_{\text{\rm{r}}}$
is given by $w=\begin{bmatrix}1\\ 0\end{bmatrix}.$ In our example,
the equality \eqref{6.20} is given by
\begin{equation}
\label{10.133}
\varphi(i,x)\,w=f(i,x)\,\beta.
 \end{equation}
Using \eqref{10.113} and \eqref{10.118} in \eqref{10.133}, we obtain
$\beta=\begin{bmatrix}5/4\\ 1/4\end{bmatrix}.$
Hence, a unit vector generating the orthogonal projection matrix
$P_{\text{\rm{r}}}$ is given by $\begin{bmatrix}5\\ 1\end{bmatrix}/\sqrt{26}.$
Thus, we obtain $P_{\text{\rm{r}}}$ as
\begin{equation*}
P_{\text{\rm{r}}}=\ds\frac{1}{\sqrt{26}}\,\begin{bmatrix}5\\ 1\end{bmatrix}\,
\ds\frac{1}{\sqrt{26}}\,\begin{bmatrix}5& 1\end{bmatrix},
 \end{equation*}
which yields
\begin{equation}
\label{10.135}
P_{\text{\rm{r}}}=\ds\frac{1}{26}\,\begin{bmatrix}25&5
\\ 5&1\end{bmatrix}.
\end{equation}
Using \eqref{10.114} and \eqref{10.135} in \eqref{6.19}, we obtain the perturbed Jost matrix $\tilde J(k)$ as
\begin{equation}
\label{10.136}
\tilde J(k)= \ds\frac{1}{26\,(3k+i)}\,
\begin{bmatrix} (k+i)(-78ik+7)&17k-7i\\
\noalign{\medskip}
17(k+i)&78k^2-59ik+17\end{bmatrix},
 \end{equation}
which yields the determinant $\det[\tilde J(k)]$ as
\begin{equation}
\label{10.137}
\det[\tilde J(k)]= \ds\frac{-3k(k-i)(k+i)}{3k+i}.
 \end{equation}
The only zero on the positive imaginary axis of the right-hand side of \eqref{10.137} is a simple zero
occurring at $k=i.$ Thus, the transformation has reduced the multiplicity of the bound state at $k=i$ from $2$ to $1.$
Using \eqref{10.113}, \eqref{10.123}, \eqref{10.125}, and \eqref{10.135} in \eqref{6.24}, we obtain the perturbed Jost solution
$\tilde f(k,x)$ as
\begin{equation}
\label{10.138}
\tilde f(k,x)= \ds\frac{e^{ikx}}{13(3k+i)\,(25+26\,e^{2x/3})}\,
\begin{bmatrix} q_{45}(k,x)&323i\\
\noalign{\medskip}
323i&q_{46}(k,x)\end{bmatrix},
 \end{equation}
 where we have let
   \begin{equation}
\label{10.139}
 q_{45}(k,x):=-36 i + 975 k + 338\, e^{2 x/3} (3 k+i),
 \end{equation}
   \begin{equation}
\label{10.140}
 q_{46}(k,x):=36 i + 975 k + 338\, e^{2 x/3} (3 k+i).
 \end{equation}
In this example, we have $V(x),$ $\tilde V(x),$ and $\tilde V(x)-V(x)$ all have the behavior
$O(e^{-2x/3})$ as $x\to+\infty.$ This asymptotic behavior agrees with the predicted asymptotic behavior 
stated in \eqref{6.16}.
In this case, our estimate given in \eqref{6.16} is sharp.

\end{example}

In the next example, we illustrate the method of Section~\ref{section8} used to increase
the multiplicity of an existing bound state. As our unperturbed problem we use the perturbed problem of
Example~\ref{example10.6}.

\begin{example}
\label{example10.7}
\normalfont
In this example we choose our unperturbed half-line matrix Schr\"odinger operator so that the unperturbed potential $V(x)$ is given by
\begin{equation}
\label{10.141}
V(x)= \ds\frac{e^{2x/3}}{9\,(25+26\,e^{2x/3})^2}\,
\begin{bmatrix} -2888&2584\\
2584&-2312\end{bmatrix},
 \end{equation}
which coincides with the right-hand side of \eqref{10.126}, and as our unperturbed boundary matrices we use the $2\times 2$
constant matrices $A$ and $B$ given by
\begin{equation}
\label{10.142}
A
=\begin{bmatrix} 1&0\\
0&1\end{bmatrix},\quad B=\begin{bmatrix} \ds\frac{287}{306}& \ds\frac{1}{18}
\\
\noalign{\medskip}
\ds\frac{1}{18}&-\ds\frac{17}{18}\end{bmatrix},
 \end{equation}
which coincide with the matrices on the right-hand side of \eqref{10.122}. 
The unperturbed Jost solution $f(k,x)$ is obtained as the $2\times 2$ matrix appearing on the right-hand side of
\eqref{10.138}, i.e. we have
\begin{equation}
\label{10.143}
f(k,x)= \ds\frac{e^{ikx}}{13(3k+i)\,(25+26\,e^{2x/3})}\,
\begin{bmatrix} q_{45}(k,x)&323i\\
\noalign{\medskip}
323i&q_{46}(k,x)\end{bmatrix},
 \end{equation}
with $q_{45}(k,x)$ and $q_{46}(k,x)$ defined as in
\eqref{10.139} and \eqref{10.140}, respectively.
From Example~\ref{example10.6} we know that the
unperturbed operator corresponding to \eqref{10.141} and \eqref{10.142} has a simple bound state at $k=i.$ We are interested in increasing the
multiplicity of that bound state from $1$ to $2.$ For this, we proceed as follows.
 To simplify our notation, we use
$\tilde Q_{\text{\rm{i}}},$
$\tilde{\mathbf G}_{\text{\rm{i}}},$
$\mathbf H_{\text{\rm{i}}},$
$\tilde C_{\text{\rm{i}}},$
$\xi_{\text{\rm{i}}}(x),$
$\Omega_{\text{\rm{i}}}(x),$
$L_{\text{\rm{i}}},$
$P_{\text{\rm{i}}},$
to denote the quantities
$\tilde Q_{N\text{\rm{i}}},$
$\tilde{\mathbf G}_{N\text{\rm{i}}},$
$\mathbf H_{N\text{\rm{i}}},$
$\tilde C_{N\text{\rm{i}}},$
$\xi_{N\text{\rm{i}}}(x),$
$\Omega_{N\text{\rm{i}}}(x),$
$L_{N\text{\rm{i}}},$
$P_{N\text{\rm{i}}},$
respectively, appearing in Section~\ref{section8}.
From \eqref{10.136} we know that the unperturbed Jost matrix $J(k)$  is given by 
\begin{equation}
\label{10.144}
J(k)= \ds\frac{1}{26\,(3k+i)}\,
\begin{bmatrix} (k+i)(-78ik+7)&17k-7i\\
\noalign{\medskip}
17(k+i)&-78ik^2-59k-17i\end{bmatrix}.
 \end{equation}
Using \eqref{10.144} we obtain
\begin{equation*}
J(i)= \ds\frac{1}{52}\,
\begin{bmatrix} 85&5\\
17&1\end{bmatrix},
 \end{equation*}
and hence we see that the kernel of $J(i)$ is spanned by a unit vector along $\begin{bmatrix} 1\\
-17\end{bmatrix}.$ Thus, the orthogonal projection matrix $Q$ onto $\text{\rm{Ker}}[J(i)]$ is given by
\begin{equation}
\label{10.146}
Q= \ds\frac{1}{1+17^2}\,
\begin{bmatrix} 1\\
-17\end{bmatrix} \begin{bmatrix} 1&-17
\end{bmatrix},
 \end{equation}
or equivalently we have
\begin{equation*}
Q= \ds\frac{1}{290}\,
\begin{bmatrix} 1&-17\\
-17&289\end{bmatrix}.
 \end{equation*}
In order to increase the multiplicity of the bound state at $k=i$ from $1$ to $2,$ we need to specify the 
matrix $\tilde C_{\text{\rm{i}}}$ appearing in \eqref{8.7}. We remark that we cannot specify
$\tilde C_{\text{\rm{i}}}$ arbitrarily because it has to be compatible with the construction
described in \eqref{8.2}--\eqref{8.7}. 
Using \eqref{8.2} we first need determine
the rank-one orthogonal projection matrix $\tilde Q_{\text{\rm{i}}}$ onto a subspace of
 the orthogonal complement in $\mathbb C^n$ of $Q\,\mathbb C^n,$ where in this case we have $n=2.$
 From \eqref{10.146} we already know that
 the orthogonal complement in $\mathbb C^2$ of $Q\,\mathbb C^2$
  is spanned by a unit vector
along  $\begin{bmatrix} 17\\
1\end{bmatrix}.$ Hence, the orthogonal projection matrix $\tilde Q_{\text{\rm{i}}}$ is given by
\begin{equation}
\label{10.148}
\tilde Q_{\text{\rm{i}}}= \ds\frac{1}{290}\,
\begin{bmatrix} 289&17\\
17&1\end{bmatrix}.
 \end{equation}
We let  
\begin{equation}
\label{10.149}
\tilde{\mathbf G}_{\text{\rm{i}}}=
\begin{bmatrix} a& b+ic
\\
\noalign{\medskip}
b-ic&d\end{bmatrix},
 \end{equation}
where $a,$ $b,$ $c,$ and $d$ are real constants.
The two equalities in \eqref{8.3} impose the restrictions on \eqref{10.149} so that we have
\begin{equation}
\label{10.150}
a=289\, d,\quad b=17d, \quad c=0.
\end{equation}
Using \eqref{10.150} in \eqref{10.149} we get
\begin{equation}
\label{10.151}
\tilde{\mathbf G}_{\text{\rm{i}}}=d\,
\begin{bmatrix} 289&17\\
17&1\end{bmatrix}.
 \end{equation}
From \eqref{10.151} we see that the eigenvalues of
$\tilde{\mathbf G}_{\text{\rm{i}}}$ are $0$ and
$290\,d.$ Since $\tilde{\mathbf G}_{\text{\rm{i}}}$ must be nonnegative and must have its rank equal to $1,$ it follows that we
must have $d>0.$
We further require that the restriction of
$\tilde{\mathbf G}_{\text{\rm{i}}}$ to the subspace
$\tilde Q_{\text{\rm{i}}}\,\mathbb C^2$ 
be invertible.
From \eqref{10.148} we see that the subspace
$\tilde Q_{\text{\rm{i}}}\,\mathbb C^2$ is spanned by 
the vector
$\begin{bmatrix}17\\
1\end{bmatrix}.$ Hence, we must have 
\begin{equation}
\label{10.152}
\tilde{\mathbf G}_{\text{\rm{i}}}
\begin{bmatrix}17\\
1\end{bmatrix}=\alpha \begin{bmatrix}17\\
1\end{bmatrix},
 \end{equation}
for some nonzero constant $\alpha.$
We observe that \eqref{10.152} holds because the column vector
$\begin{bmatrix}17\\
1\end{bmatrix}$ is an eigenvector of
$\tilde{\mathbf G}_{\text{\rm{i}}}$
with the eigenvalue $\alpha$ being equal to
$290d.$
Consequently, $\tilde{\mathbf G}_{\text{\rm{i}}}$ is given by the matrix
on the right-hand side of \eqref{10.151} with any positive constant $d.$
Using \eqref{10.148} and \eqref{10.151} in \eqref{8.4}, we get
\begin{equation}
\label{10.153}
\tilde{\mathbf H}_{\text{\rm{i}}}=
\begin{bmatrix}  \ds\frac{1}{290}+289\,d&-\ds\frac{17}{290}+17\,d
\\
\noalign{\medskip}
-\ds\frac{17}{290}+17\,d&-\ds\frac{289}{290}+d\end{bmatrix}.
 \end{equation}
Using \eqref{8.6} and \eqref{10.153}, we obtain the positive matrix $\tilde{\mathbf H}_{\text{\rm{i}}}^{1/2}$ and
its inverse $\tilde{\mathbf H}_{\text{\rm{i}}}^{-1/2}$ as
\begin{equation*}
\tilde{\mathbf H}_{\text{\rm{i}}}^{1/2}=
\begin{bmatrix}  \ds\frac{1}{290}+\ds\frac{289\,\sqrt{d}}{\sqrt{290}}&-\ds\frac{17}{290}+\ds\frac{17\,\sqrt{d}}{\sqrt{290}}
\\
\noalign{\medskip}
-\ds\frac{17}{290}+\ds\frac{17\,\sqrt{d}}{\sqrt{290}}&\ds\frac{289}{290}+\ds\frac{\sqrt{d}}{\sqrt{290}}\end{bmatrix},
 \end{equation*}
\begin{equation}
\label{10.155}
\tilde{\mathbf H}_{\text{\rm{i}}}^{-1/2}=
\begin{bmatrix}  \ds\frac{1}{290}+\ds\frac{289}{290\,\sqrt{290\, d}}&-\ds\frac{17}{290}+\ds\frac{17}{290\,\sqrt{290\,d}}
\\
\noalign{\medskip}
-\ds\frac{17}{290}+\ds\frac{17}{290\,\sqrt{290\,d}}&\ds\frac{289}{290}+\ds\frac{1}{290\,\sqrt{290\,d}}\end{bmatrix}.
 \end{equation}
Using \eqref{10.148} and \eqref{10.155} in \eqref{8.7}, we get the matrix $\tilde C_{\text{\rm{i}}}$ as
\begin{equation}
\label{10.156}
\tilde C_{\text{\rm{i}}}= \ds\frac{1}{290\,\sqrt{290\,d}}\,
\begin{bmatrix} 289&17\\
17&11\end{bmatrix}.
 \end{equation}
Thus, the only freedom we have in the specification of 
$\tilde C_{\text{\rm{i}}}$ is the choice of the positive constant $d.$ We will see that
each specific choice of $d$ yields a particular perturbed potential $\tilde V(x),$ a particular
perturbed regular solution $\tilde\varphi(k,x),$ a particular perturbed Jost solution $\tilde f(k,x),$
a particular perturbed boundary matrix $\tilde B.$ We will also see that
the constructed perturbed boundary matrix $\tilde A,$ perturbed orthogonal projection $\tilde P_{\text{\rm{i}}},$ and
and perturbed Jost matrix
$\tilde J(k)$ are not affected by any particular choice of $d.$
Using \eqref{10.142} and \eqref{10.156} in \eqref{8.16}, we get
\begin{equation*}
\tilde A
=\begin{bmatrix} 1&0\\
0&1\end{bmatrix},\quad \tilde B=\begin{bmatrix} \ds\frac{287}{306}-\ds\frac{289}{84100\,d}
& \ds\frac{1}{18}-\ds\frac{17}{84100\,d}
\\
\noalign{\medskip}
\ds\frac{1}{18}-\ds\frac{17}{84100\,d}&-\ds\frac{17}{18}-\ds\frac{1}{84100\,d}\end{bmatrix}.
 \end{equation*}
In order to construct the orthogonal projection
matrix $\tilde P_{\text{\rm{i}}},$ we first use \eqref{8.17} and \eqref{8.24} and
construct the matrix product
$L_{\text{\rm{i}}}\,Q_{\text{\rm{i}}}.$
With the help of \eqref{10.127}--\eqref{10.132}, for the unperturbed regular solution $\varphi(k,x)$ at $k=i$ we obtain
\begin{equation}
\label{10.158}
\varphi(i,x)= \ds\frac{1}{4(25+26\,e^{2x/3})}\,
\begin{bmatrix} q_{47}(x)&q_{48}(x)\\
q_{49}(x)&q_{50}(x)\end{bmatrix},
 \end{equation}
where we have defined
\begin{equation*}
q_{47}(x):=17 e^x (7 + 5 e^{2 x/3}),\quad
q_{48}(x):=2\left(12 e^{-x}- 24 e^{-x/3}+ 7 e^x + 5 e^{5 x/3}\right), \end{equation*}
\begin{equation*}
q_{49}(x):=
17 e^x (-1 + e^{2 x/3}),\quad 
q_{50}(x):=2\left(84 e^{-x}+ 120 e^{-x/3}- e^x + e^{5 x/3}\right).
 \end{equation*}
From \eqref{10.148} and \eqref{10.158}, we get
\begin{equation}
\label{10.161}
\varphi(i,x)\,Q_{\text{\rm{i}}}
= \ds\frac{1}{1160(25+26\,e^{2x/3})}\,
\begin{bmatrix} q_{47}(x)&q_{48}(x)\\
q_{49}(x)&q_{50}(x)\end{bmatrix} \,
\begin{bmatrix} 289&17\\
17&1\end{bmatrix}.
 \end{equation}
Using the first asymptotics as $x\to+\infty$ in each of \eqref{2.10} and \eqref{7.23}, from \eqref{8.17}
with $\kappa_1=1$
we obtain
the matrix product $L_{\text{\rm{i}}}\,Q_{\text{\rm{i}}}$ as
\begin{equation}
\label{10.162}
L_{\text{\rm{i}}}\,Q_{\text{\rm{i}}}= \ds\frac{1}{104}\,
\begin{bmatrix} 85&5\\
17&1\end{bmatrix}.
 \end{equation}
Hence, \eqref{10.156} and \eqref{10.162} yield the matrix product $L_{\text{\rm{i}}}\,Q_{\text{\rm{i}}}\,\tilde C_{\text{\rm{i}}}$ as
\begin{equation}
\label{10.163}
L_{\text{\rm{i}}}\,Q_{\text{\rm{i}}}\,\tilde C_{\text{\rm{i}}}= \ds\frac{1}{104\,\sqrt{290\,d}}\,
\begin{bmatrix} 85&5\\
17&1\end{bmatrix},
 \end{equation}
where we remark that the right-hand side of \eqref{10.163} contains the positive parameter $d$ even though
the right-hand side of \eqref{10.162} does not contain $d.$ Using \eqref{10.163} in \eqref{8.25}, we obtain the
orthogonal projection matrix $P_{\text{\rm{i}}}$ into the kernel of $\tilde J(i)^\dagger$ as
\begin{equation}
\label{10.164}
P_{\text{\rm{i}}}= \ds\frac{1}{26}\,
\begin{bmatrix} 25&5\\
5&1\end{bmatrix}.
 \end{equation}
From \eqref{10.164} we observe that $P_{\text{\rm{i}}}$
is independent of the positive parameter $d.$ Using \eqref{10.144} and \eqref{10.164} in \eqref{8.23} we get the
perturbed Jost matrix $\tilde J(k)$ as
\begin{equation}
\label{10.165}
\tilde J(k)= \ds\frac{1}{2\,(3k+i)}\,
\begin{bmatrix} -i(k-i)(6k+i)&-k+i\\
\noalign{\medskip}
-k+i& -i(k-i)(6k+i)\end{bmatrix}.
 \end{equation}
From \eqref{10.165} we see that $\tilde J(k)$ is independent of the positive parameter $d.$
The determinant of the perturbed Jost matrix $\tilde J(k)$ appearing in \eqref{10.165} is given by
\begin{equation}
\label{10.166}
\det[\tilde J(k)]= \ds\frac{-3k(k-i)^2}{3k+i}.
 \end{equation}
From the factor $(k-i)^2$ in the numerator
on the right-hand side of \eqref{10.166}, we see that the perturbed operator has a bound state at $k=i$ with the multiplicity $2.$
Using \eqref{10.156} and \eqref{10.161} in \eqref{8.9}, we obtain $\xi_{\text{\rm{i}}}(x)$ as
\begin{equation}
\label{10.167}
\xi_{\text{\rm{i}}}(x)= \ds\frac{1}{580\,\sqrt{290\,d}\,(25+26\,e^{2x/3})}\,
\begin{bmatrix}17\,q_{51}(x)&q_{51}(x)\\
\noalign{\medskip}
17\,q_{52}(x)&q_{52}(x)\end{bmatrix},
 \end{equation}
where we have defined
\begin{equation*}
q_{51}(x):=6\,e^{-x} - 12 \,e^{- x/3}+ 1015 \,e^x+ 725 e^{5 x/3},
 \end{equation*}
\begin{equation*}
q_{52}(x):=42\,e^{-x} +60\, \,e^{- x/3}-145 \,e^x+ 145 e^{5 x/3}.
 \end{equation*}
We remark that $\xi_{\text{\rm{i}}}(x)$ contains the positive parameter $d$ in its denominator.
Using \eqref{10.148} and \eqref{10.167} in \eqref{8.12}, we obtain
$\Omega_{\text{\rm{i}}}(x)$ as
\begin{equation}
\label{10.170}
\Omega_{\text{\rm{i}}}(x)=
\begin{bmatrix} \ds\frac{289}{290}-289\,q_{53}(x,d)& \ds\frac{17}{290}-17\,q_{53}(x,d)\\
\noalign{\medskip}
\ds\frac{17}{290}-17\,q_{53}(x,d)&\ds\frac{1}{290}-q_{53}(x,d)\end{bmatrix},
 \end{equation}
where we have defined
\begin{equation*}
q_{53}(x,d):=\ds\frac{ 29585 + 72 \,e^{-2 x} + 144 \,e^{-4 x/3} + 33274 \,e^{2 x/3} - 
 42050 \,e^{2 x} - 21025 \,e^{8 x/3}}{ 195112000\, d\, (25 + 26\,e^{2 x/3})  }.
 \end{equation*}
We note that $\Omega_{\text{\rm{i}}}(x)$ contains the positive parameter $d.$
From \eqref{10.170}, as also done in Example~\ref{example10.4}, with the help of the symbolic software Mathematica,
we obtain the Moore--Penrose inverse of $\Omega_{\text{\rm{i}}}(x)$ as
\begin{equation}
\label{10.172}
\Omega_{\text{\rm{i}}}(x)^+
= \ds\frac{d\,e^{2x}(25+26 e^{2x/3})}{q_{54}(x,d)}
\begin{bmatrix}670480 & 39440\\
\noalign{\medskip}
39440&2320\end{bmatrix},
 \end{equation}
where we have let
\begin{equation*}
\begin{split}
q_{54}(x,d):=&-72 - 144 \,e^{2 x/3} + 5 (-5917 + 3364000\, d)\, e^{2 x} \\
&+ 
 2 (-16637 + 8746400\, d) \,e^{8 x/3} + 42050 \,e^{4 x} + 21025 \,e^{14 x/3}.
\end{split}
 \end{equation*}
 We verify that the matrix 
 $\Omega_{\text{\rm{i}}}(x)^+$ appearing in \eqref{10.172} is indeed 
 the Moore--Penrose inverse of
 the matrix $\Omega_{\text{\rm{i}}}(x)$ by verifying that
 those two matrices satisfy the four equalities given in (3.2) of
 \cite{AW2025}
 used in the definition of the
 Moore--Penrose inverse \cite{BG2003} of
 a matrix.
With the help of \eqref{10.141}, \eqref{10.167}, and \eqref{10.172}, from \eqref{8.13} we recover the perturbed potential
$\tilde V(x)$ as
\begin{equation*}
\tilde V(x)= -\ds\frac{8\, e^{2x/3}}{9\,[q_{54}(x,d)]^2}
\begin{bmatrix}q_{55}(x,d) & q_{56}(x,d)\\
\noalign{\medskip}
 q_{56}(x,d)& q_{57}(x,d)\end{bmatrix},
 \end{equation*}
where we have defined
\begin{equation*}
\begin{split}
q_{55}(x,d):=&5184 - 324 \, (-209 + 672800 \,d) \, e^{4 x/3}- 
 1440 \, (-1601 + 672800\, d)  \,e^{2 x}\\
&- 
 1296 \, (-1949 + 672800 \,d)  \,e^{8 x/3} - 
 54496800 \,e^{ 10 x/3}\\
& + (541979641
- 657593374400\, d + 163410202240000\, d^2)
 \,e^{ 4 x}\\
&  - 54496800 \, e^{14 x/3}+ 
 1324575 \, (-8423 + 4709600\, d)  \,e^{16 x/3} \\
& + 
 336400 \, (-45143 + 24893600\, d) \,e^{6 x}\\
& + 
 4730625 (-1253 + 672800\, d) \,e^{20 x/3} + 442050625 e^{8 x},
\end{split}
 \end{equation*}
\begin{equation*}
\begin{split}
q_{56}(x,d):=&
5184 - 324  \,(-4943 + 4709600  \,d) \, e^{4 x/3}+ 
 1152 \, (-3689 + 672800 \, d) \,  e^{2 x}\\
& + 
 6480 \, (-1601 + 672800 \, d) \,  e^{8 x/3}\\
&+ (-455148503 + 528668747200  \,d 
 - 146209128320000  \,d^2)  \,e^{ 4 x}\\
& - 189225 (-4943 + 4709600 \, d) \, e^{16 x/3} 
+ 
 336400 (-3689 + 672800 d)  \,e^{6 x} \\
& + 
 946125 \, (-1601 + 672800 d) \,e^{20 x/3} + 442050625 \, e^{8 x},
\end{split}
 \end{equation*}
\begin{equation*}
\begin{split}
q_{57}(x,d):=&
5184 - 2268\, (-8423 + 4709600\, d)\, e^{4 x/3}\\
& - 
 288\, (-184261 + 100247200\, d)\, e^{2 x} \\
& - 
 32400 \,(-1253 + 672800\, d)\,  e^{8 x/3} - 
 54496800 \,e^{10 x/3}\\
&  + (216875449 - 419599793600\, d + 130818693760000\, d^2)\, e^{ 4 x}\\
&  - 54496800\, e^{14 x/3} + 189225\, (-209 + 672800\, d) \,e^{16 x/3}\\
&   + 
 336400\, (-1079 + 672800\, d)\, e^{6 x} \\
& + 
 189225\, (-1949 + 672800\, d)  \,e^{20 x/3} + 442050625\, e^{8 x}.
\end{split}
 \end{equation*}
We remark that the potential $\tilde V(x)$ contains the parameter $d.$
In this example, we have $V(x),$ $\tilde V(x),$ and $\tilde V(x)-V(x)$ all have the behavior
$O(e^{-2x/3})$ as $x\to+\infty.$ That asymptotic behavior agrees with the predicted asymptotic behavior described in \eqref{8.18}.
In this case, our estimate given in \eqref{8.18} is sharp.
Finally, using \eqref{10.143}, \eqref{10.164}, \eqref{10.167}, and \eqref{10.172} in \eqref{8.29}, we obtain the perturbed
Jost solution $\tilde f(k,x)$ explicitly. Since that expression is too lengthy, we do not display it here.

\end{example}

In the next example we illustrate the removal of a discrete eigenvalue in the scalar case where the unperturbed potential
decays as $O(1/x^2)$ when $x\to+\infty.$ We remark that the potential in this example is integrable but does not belong to the class
$L^1_1(\mathbb R^+).$ Hence, it is outside the class
of potentials for which
the theory of removing bound states is developed of Section~\ref{section5}.
Nevertheless, we follow the procedure of Section~\ref{section5} to remove a bound state in this example.

\begin{example}
\label{example10.8} 
\normalfont
Consider the scalar case, i.e. $n=1,$ where the unperturbed potential is given by
\begin{equation}
\label{10.178}
V (x)=\ds\frac{2}
{\left(1+x\right)^2},
 \end{equation}
 and the boundary matrices $A$ and $B,$ which are scalars in this case, are given by
\begin{equation}
\label{10.179}
A=1,\quad B=-2.
 \end{equation}
 The corresponding Jost solution $f(k,x),$  Jost matrix $J(k),$ scattering matrix $S(k),$ and regular solution $\varphi(k,x)$ are evaluated by proceeding
 as in Example~\ref{example9.2}, and we get
 \begin{equation}
\label{10.180}
f(k,x)= e^{ikx}\left[1+\ds\frac{i}{k(1+x)}\right],
 \end{equation}
  \begin{equation}
\label{10.181}
J(k)=\ds\frac{[k+i(\sqrt{5}+1)/2][k-i(\sqrt{5}+1)/2]}{ik},
 \end{equation}
 \begin{equation*}
S(k)=\ds\frac{k^2+ik+1}{k^2-ik+1},
 \end{equation*} 
  \begin{equation*}
\varphi(k,x)=\ds\frac{q_{58}(k,x)}{2k^3\left(1+x\right)},
 \end{equation*}
where we have defined
 \begin{equation*}
q_{58}(k,x):=e^{-ikx}\left(k^2-ik+1\right)\left(kx+k-i\right)+e^{ikx}\left(k^2+ik+1\right)\left(kx+k+i\right).
 \end{equation*}
 From the zeros of the Jost matrix $J(k)$ given in \eqref{10.181} we observe that
there is a simple bound state at 
$k=i\kappa_1,$ 
where the value of $\kappa_1$ is given by
\begin{equation}
\label{10.185}
\kappa_1=\ds\frac{1+\sqrt{5}}{2}.
 \end{equation}
The Gel'fand--Levitan constant $C_1$ corresponding to the bound state at $k=i\kappa_1$ is evaluated as described
in Example~\ref{example10.2}, and we have
\begin{equation}
\label{10.186}
C_1=\sqrt{2+\ds\frac{4}{\sqrt{5}}}.
 \end{equation}
 By proceeding as in
 Example~\ref{example10.2}, 
we remove the bound state and obtain the corresponding perturbed
 Jost matrix $\tilde J(k),$ perturbed scattering matrix $\tilde S(k),$ perturbed boundary matrices $\tilde A$ and $\tilde B,$ and perturbed potential $\tilde V(x)$ as
 \begin{equation}
\label{10.187}
\tilde J(k)=\ds\frac{[k+i(\sqrt{5}+1)/2][k+i(\sqrt{5}-1)/2]}{ik},
 \end{equation}
  \begin{equation*}
\tilde S(k)=\ds\frac{(k^2+ik+1)[k-i(\sqrt{5}+1)/2]^2}{(k^2-ik+1)[k+i(\sqrt{5}+1)/2]^2},
 \end{equation*}
 \begin{equation*}
\tilde A=1,\quad \tilde B=\ds\frac{4}{\sqrt{5}},
 \end{equation*}
 \begin{equation*}
\tilde V (x)=\ds\frac{2}
{\left(\sqrt{5}+x\right)^2},
 \end{equation*}
  \begin{equation*}
\tilde f(k,x)=\ds\frac{e^{ikx} (2ik-1-\sqrt{5})\, q_{59}(k,x)}
{ ik(2ik+1+\sqrt{5})(5+\sqrt{5}\,x)(2k^2+3+\sqrt{5})},
 \end{equation*}
where we have defined
  \begin{equation*}
\begin{split}
q_{59}(k,x):=&5 + 3 \sqrt{5}-ik(5 + 3 \sqrt{5})(x+1)\\
&+k^2[10+8\,\sqrt{5}+(10+2\,\sqrt{5})\,x]+ik^3(10+2\,\sqrt{5}\,x).
\end{split}
 \end{equation*}
 From \eqref{10.187} we see that $\tilde J(k)$ does not have any zeros on the positive
 imaginary axis, and hence the perturbed operator does not have any bound states.
In this example, the unperturbed potential $V(x)$ and
the perturbed potential $\tilde V(x)$ each behave
as $O(1/x^2)$ as $x\to+\infty,$ and the potential difference
$\tilde V(x)-V(x)$ behaves as
 $O(1/x^3)$ as $x\to+\infty.$
 
\end{example}

In the next example, we revisit
Example~\ref{example10.8}, but we add a bound state to the unperturbed problem instead of removing the existing
bound state in the unperturbed problem.

\begin{example}
\label{example10.9} 
\normalfont
We consider the scalar case, i.e. $n=1,$ where the unperturbed quantities
are the same as in Example~\ref{example10.5}. Therefore, the unperturbed potential
$V(x),$ boundary matrices $A$ and $B,$ Jost solution $f(k,x),$ Jost matrix $J(k),$ scattering matrix $S(k),$ and regular solution $\varphi(k,x)$
are as in \eqref{10.178}--\eqref{10.181}, respectively. The unperturbed operator has one bound state at
$k=i\kappa_1$ with the Gel'fand--Levitan normalization constant $C_1,$
where the values of $\kappa_1$ and $C_1$ are as in \eqref{10.185} and \eqref{10.186}, respectively.
We add one bound state to the unperturbed problem at $k=i\tilde\kappa_2$ with the
 Gel'fand--Levitan normalization constant $\tilde C_2,$
where we use
\begin{equation*}
\tilde\kappa_2=1,\quad \tilde C_2=2.
 \end{equation*}
By proceeding as in Example~\ref{example10.1}, we determine the perturbed quantities. We obtain the
perturbed potential $\tilde V(x),$ the boundary matrices $\tilde A$ and $\tilde B,$ 
Jost matrix $\tilde J(k),$ scattering matrix $\tilde S(k),$ and Jost solution $\tilde f(k,x)$ as
 \begin{equation}
\label{10.194}
\tilde V (x)=\ds\frac{q_{60}(x)+q_{61}(x)}
{\left[3+x+e^{2x}(-6-6x+4 x^2)+e^{4x}(1-x)\right]^2},
 \end{equation}
 where we have defined
  \begin{equation*}
q_{60}(x):=2  + 2\, e^{2 x}\left(84 + 8 x - 40 x^2 - 16 x^3\right),
 \end{equation*}
  \begin{equation*}
q_{61}(x):=4\, e^{4 x}\left(33 - 8 x + 24 x^2\right)+8\, e^{6 x}\left(-5 + 6 x - 14 x^2 + 4 x^3\right)+2 \,e^{8x},
 \end{equation*}
 \begin{equation*}
\tilde A=1,\quad \tilde B=-6,
 \end{equation*}
\begin{equation}
\label{10.198}
\tilde J(k)=\ds\frac{(k-i)(k-i(\sqrt{5}+1)/2)}{ik(k+i)(k+i(\sqrt{5}+1)/2)},
 \end{equation}
\begin{equation*}
\tilde S(k)=\ds\frac{(k+i)^2(k^2+ik+1)}{(k-i)^2(k^2-ik+1)},
 \end{equation*}
\begin{equation*}
\tilde f(k,x)=\ds\frac{e^{ikx}\left[ -(k-i)^2[i+k(3+x)]+e^{2x}\, q_{62}(k,x)+ e^{4x}\,q_{63}(k,x)\right]}{k(k+i)^2[-3-x+e^{2x}(6+6x-4x^2)+e^{4x}(-1+x)},
 \end{equation*}
 \begin{equation*}
q_{62}(k,x):=-2 i (-5 + 2 x) +k (10 + 6 x - 4 x^2) - 2 i k^2 (-3 + 4 x) + k^3 (6 + 6 x - 4 x^2),
 \end{equation*}
  \begin{equation*}
q_{63}(k,x):=(i + k)^2 [i + k (-1 + x)].
 \end{equation*}
From \eqref{10.198} we observe that $\tilde J(k)$ has two zeros 
on the positive imaginary axis at $k=i$ and $k=i(\sqrt{5}+1)/2,$
confirming that the perturbed operator has two bound states. We remark that from \eqref{10.178} and \eqref{10.194}, we obtain
$V(x)=O(1/x^2),$ $\tilde V(x)=O(1/x^2),$ and
$\tilde V(x)-V(x)=O(1/x^3)$ as $x\to+\infty.$

\end{example}

\noindent {\bf Acknowledgments.} Ricardo Weder is an Emeritus National Researcher of 
SNII-SECIHTI, M\'exico.

\end{document}